\documentclass[moor, nonblindrev]{informs3_hide} 
\DoubleSpacedXI

\usepackage[utf8]{inputenc}

\let\theoremstyle\relax 
\usepackage{amsfonts,amsmath,color,amsthm,amssymb, url, enumerate, bbm}
\usepackage[margin=1in]{geometry}
\usepackage[dvipsnames]{xcolor}
\usepackage{caption}
\usepackage[font=small]{subcaption}
\usepackage{float}
\usepackage{wrapfig}
\usepackage{tikz, etoolbox}
\usepackage{thmtools, thm-restate}
\usetikzlibrary{shapes}
\usetikzlibrary{arrows}
\usetikzlibrary{decorations.markings}

\tikzstyle{vertex}=[circle, draw, fill, inner sep=0pt, minimum size=5pt]
\newcommand{\vertex}{\node[vertex]}

\tikzstyle{smvertex}=[circle, draw, fill, inner sep=0pt, minimum size=2pt]
\newcommand{\smvertex}{\node[smvertex]}

\tikzstyle{mvertex}=[circle, draw, fill, inner sep=0pt, minimum size=4pt]
\newcommand{\mvertex}{\node[mvertex]}

\newcommand{\abs}[1]{\left\lvert{#1}\right\rvert}

\newcommand{\f}{\frac}

\newcommand{\x}{\times}

\newcommand{\Z}{\mathbb{Z}}

\newcommand{\cost}{\mathrm{cost}}

\definecolor{purple}{RGB}{128,0,128}

\def\qed{\vbox{\hrule\hbox{\vrule\kern3pt\vbox{\kern6pt}\kern3pt\vrule}\hrule}}

\def\n{\noindent}
\newtheorem{thm}{Theorem}[section]

\newtheorem{cor}[thm]{Corollary}
\newtheorem{rem}[thm]{Remark}
\newtheorem{prop}[thm]{Proposition}
\newtheorem{defn}[thm]{Definition}
\newtheorem{cm}[thm]{Claim}

\newtheorem*{thm*}{Theorem}
\newtheorem*{cor*}{Corollary}
\newtheorem*{lm*}{Lemma}
\newtheorem*{cm*}{Claim}
\newtheorem*{prop*}{Proposition}

\theoremstyle{definition}

\newtheoremstyle%
 {Aside}%
 {}{}%
 {\color{purple}\itshape}
 {}%
 {\color{purple}\bfseries}%
 {\color{purple}.}%
 { }{}
\theoremstyle{Aside}

\usepackage{natbib}
 \NatBibNumeric
 \def\bibfont{\small}%
 \bibpunct[, ]{[}{]}{,}{n}{}{,}%

\usepackage[colorlinks=true,breaklinks=true,bookmarks=true,urlcolor=blue,
     citecolor=blue,linkcolor=blue,bookmarksopen=false,draft=false]{hyperref}
\usepackage{cleveref}

\def\EMAIL#1{\href{mailto:#1}{#1}}
\def\URL#1{\href{#1}{#1}}         

\TheoremsNumberedThrough     

\EquationsNumberedThrough    

\begin{document}


\RUNAUTHOR{Gutekunst, Jin, and Williamson}


\RUNTITLE{The Two-Stripe Symmetric Circulant TSP is in P}

\TITLE{The Two-Stripe Symmetric Circulant TSP is in P}

\ARTICLEAUTHORS{%
\AUTHOR{Samuel C. Gutekunst}
\AFF{Bucknell University, \EMAIL{s.gutekunst@bucknell.edu}, \URL{https://samgutekunst.com/}}
\AUTHOR{Billy Jin}
\AFF{Cornell University, \EMAIL{bzj3@cornell.edu}, \URL{https://people.orie.cornell.edu/bzj3/}}
\AUTHOR{David P. Williamson}
\AFF{Cornell University, \EMAIL{davidpwilliamson@cornell.edu}, \URL{http://www.davidpwilliamson.net/work/}}
} 

\ABSTRACT{%
The symmetric circulant TSP is a special case of the traveling salesman problem in which edge costs are symmetric and obey circulant symmetry.   Despite the substantial symmetry of the input, remarkably little is known about the symmetric circulant TSP, and the complexity of the problem has been an often-cited open question.  Considerable effort has been made to  understand the case in which only edges of two lengths are allowed to have finite cost: the two-stripe symmetric circulant TSP.  In this paper, we resolve the complexity of the two-stripe symmetric circulant TSP. To do so, we reduce two-stripe symmetric circulant TSP to the problem of finding certain minimum-cost Hamiltonian paths on \emph{cylindrical graphs}.  We then solve this Hamiltonian path problem.  Our results show that the two-stripe symmetric circulant TSP is in P.  Note that a two-stripe symmetric circulant TSP instance consists of a constant number of inputs (including $n$, the number of cities), so that a polynomial-time algorithm for the decision problem must run in time polylogarithmic in $n$, and a polynomial-time algorithm for the optimization problem cannot output the tour.  We address this latter difficulty by showing that the optimal tour must fall into one of two parameterized classes of tours, and that we can output the class and the parameters in polynomial time.  Thus we make a substantial contribution to the set of polynomial-time solvable special cases of the TSP, and take an important step towards resolving the complexity of the general symmetric circulant TSP.
}%



\maketitle

\section{Introduction}
The traveling salesman problem (TSP) is one of the most famous problems in combinatorial optimization.  An input to the TSP consists of a set of $n$ cities $[n]:=\{1, 2, ..., n\}$ and edge costs $c_{ij}$ for each pair of distinct $i, j \in [n]$ representing the cost of traveling from city $i$ to city $j$.  Given this information, the TSP is to find a minimum-cost tour visiting every city exactly once.    Throughout this paper, we implicitly assume that the edge costs are \emph{symmetric} (so that $c_{ij}=c_{ji}$ for all distinct $i, j\in [n]$)  and  interpret the $n$ cities as vertices of the complete undirected graph $K_n$ with edge costs $c_e=c_{ij}$ for edge $e=\{i, j\}$.  In this setting, the TSP is to find a minimum-cost Hamiltonian cycle on $K_n.$

With just this set-up, the TSP is well known to be NP-hard.  An algorithm that could approximate TSP solutions in polynomial time to within any constant factor $\alpha$ would imply P=NP (see, e.g., Theorem 2.9 in Williamson and Shmoys \cite{DDBook}).   Thus it is common to consider special cases, such as requiring costs to obey the triangle inequality (i.e. requiring costs to be \emph{metric}, so that $c_{ij}+c_{jk} \geq c_{ik}$ for all $i, j, k\in [n]$), or using costs that are distances between points in Euclidean space.  Another special case that has been considered is the $(1, 2)$-TSP which restricts $c_{ij}\in\{1, 2\}$ for every edge $\{i, j\}$ (see, e.g., Papadimitriou and Yannakakis \cite{Pap93}, Berman and Karpinski \cite{Ber08}, Karpinski and Schmied  \cite{Karp12}). 

In this paper, we consider a different class of instances: circulant TSP.  This class can be described by {\bf circulant matrices,}  matrices of the form
\begin{equation}\label{eq:circMat}
\begin{pmatrix} m_0 & m_1 & m_2 & m_3 & \cdots & m_{n-1} \\ m_{n-1} & m_0 & m_1 & m_2 & \cdots & m_{n-2} \\ m_{n-2} & m_{n-1} & m_0 & m_1 & \ddots & m_{n-3} \\ \vdots & \vdots & \vdots & \vdots & \ddots & \vdots \\ m_1 & m_2 & m_3 & m_4 & \cdots & m_0 \end{pmatrix} = \left(m_{(t-s) \text{ mod } n}\right)_{s, t=1}^n.
\end{equation}
In {\bf circulant TSP}, the matrix of edge costs $C=(c_{ij})_{i, j=1}^n$ is circulant; the cost of edge $\{i, j\}$ only depends on $i-j$ mod $n$.  Our assumption that the edge costs are symmetric and that $K_n$ is a simple graph implies that, for symmetric circulant TSP instances, we can write our cost matrix in terms of $ \lfloor \frac{n}{2}\rfloor$ parameters: \begin{equation}\label{eq:ScircMat} C=(c_{(j-i) \text{ mod } n})_{i, j=1}^n=\begin{pmatrix} 0 & c_1 & c_2 & c_3 & \cdots & c_{1} \\ c_{1} & 0 & c_1 & c_2 & \cdots & c_2 \\ c_{2} & c_{1} &0 & c_1 & \ddots & c_{3} \\ \vdots & \vdots & \vdots & \vdots & \ddots & \vdots \\ c_1 & c_2 & c_3 & c_4 & \cdots & 0\end{pmatrix},\end{equation} with $c_0=0$ and  $c_i=c_{n-i}$ for $i=1, ..., \lfloor \frac{n}{2}\rfloor.$  See, e.g., Figure \ref{fig:sym} for an picture of (symmetric) circulant symmetry.  Importantly, in circulant TSP we do not implicitly assume that the edge costs are  metric.  The original motivations for circulant TSP stem from minimizing wallpaper waste (Garfinkel \cite{Gar77}) and  reconfigurable network design (Medova \cite{Med93}).

\begin{figure}[th!]
	\begin{center}
\begin{tikzpicture}
\tikzset{vertex/.style = {shape=circle,draw,minimum size=2.5em}}
\tikzset{edge/.style = {->,> = latex'}}
\node[draw=none,minimum size=8cm,regular polygon,regular polygon sides=12] (a) {};

\foreach [evaluate ={\j=int(mod(\i, 12)+1)}] \i in {1,2, 3, ..., 12}
\draw[line width=1pt] (a.corner \i) -- (a.corner \j);

\foreach [evaluate ={\j=int(mod(\i+1, 12)+1)}] \i in {1,2, 3, ..., 12}
\draw[dotted, line width=1pt, red] (a.corner \i) -- (a.corner \j);

\foreach [evaluate ={\j=int(mod(\i+2, 12)+1)}] \i in {1,2, 3, ..., 12}
\draw[densely dotted, line width=1pt, blue] (a.corner \i) -- (a.corner \j);

\foreach [evaluate ={\j=int(mod(\i+3, 12)+1)}] \i in {1,2, 3, ..., 12}
\draw[loosely dotted, line width=1pt, magenta] (a.corner \i) -- (a.corner \j);

\foreach [evaluate ={\j=int(mod(\i+4, 12)+1)}] \i in {1,2, 3, ..., 12}
\draw[dashed, line width=1pt, lime] (a.corner \i) -- (a.corner \j);

\foreach [evaluate ={\j=int(mod(\i+5, 12)+1)}] \i in {1,2, 3, ..., 12}
\draw[dashdotted, line width=1pt, pink] (a.corner \i) -- (a.corner \j);

\foreach \n [count=\nu from 1, remember=\n as \lastn, evaluate={\nu+\lastn}] in {12, 11, ..., 1} 
\node[vertex, fill=white]at(a.corner \n){$\nu$};

\end{tikzpicture}

	\end{center}
	\caption{Circulant symmetry.  Edges of a fixed length are indistinguishable and have the same cost.  E.g. all edges of the form $\{v, v+1\}$ (where $v+1$ is taken mod $n$) have the same appearance and cost. }\label{fig:sym}\end{figure}
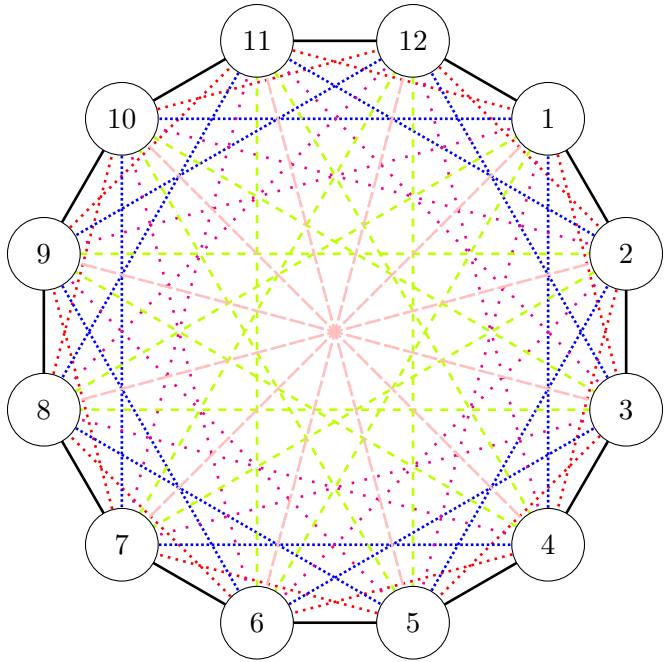

 Circulant TSP is a compelling open problem because of the intriguing structure circulant symmetry broadly provides to  combinatorial optimization problems: It seems to provide just enough structure to make an  ambiguous set of instances, as it is unclear whether or not a given combinatorial optimization problem should remain hard or become easy when restricted to circulant instances.  Some classic combinatorial optimization problems are easy when restricted to circulant instances: in the late 70's, Garfinkel \cite{Gar77} considered a restricted set of circulant TSP instances motivated by minimizing wallpaper waste and argued that, for these instances, the canonical greedy algorithm for TSP (the nearest neighbor heuristic) provides an optimal solution.  
In the late 80's, Burkard and Sandholzer \cite{Burk91} showed that the decidability question for whether or not a symmetric circulant graph (i.e. a graph whose adjacency matrix is circulant) is Hamiltonian can be solved in polynomial time and showed that  bottleneck TSP is polynomial-time solvable on symmetric circulant graphs.  Bach, Luby, and Goldwasser (cited in Gilmore, Lawler, and Shmoys \cite{Gil85}) showed that one could find minimum-cost Hamiltonian paths in (not-necessarily-symmetric) circulant graphs in polynomial time.  In contrast, Codenotti, Gerace, and Vigna \cite{Code98} show that Max Clique and Graph Coloring remain NP-hard when restricted to circulant graphs and do not admit constant-factor approximation algorithms unless P=NP.

Because of this ambiguity, the complexity of circulant TSP has often been cited as an open problem (see, e.g., Burkhard \cite{Burk97}, Burkhard, De\u{\i}neko, Van Dal, Van der Veen, and Woeginger \cite{Burk98}, and Lawler, Lenstra, Rinnooy Kan, and Shmoys  \cite{Law07}). 
It is not known if the circulant TSP is solvable in polynomial-time or is NP-hard. Prior to this work, the complexity of circulant TSP was not understood even when restricted to instances where only two of the edge costs $c_1, ...., c_{\lfloor \frac{n}{2}\rfloor}$ are finite: the \emph{symmetric two-stripe circulant TSP} (which we will henceforth abbreviate as \emph{two-stripe TSP}
) (see Greco and Gerace \cite{Grec07} and Gerace and Greco \cite{Ger08b}).  Yang, Burkard, \c{C}ela, and Woeginger \cite{Yang97} provide a polynomial-time algorithm for \emph{asymmetric} TSP in circulant graphs with only two stripes having finite edge costs.  The symmetric two-stripe circulant TSP is not, however, a special case of the asymmetric two-stripe version.\footnote{In the symmetric case, edges $\{v, v + i\}$ and $\{v, v - i\}$ of cost $c_i$ connect $v$ to both $v + i$ and $v - i$; in the asymmetric case, there are edge costs $c_1, \ldots, c_{n-1}$ and an edge $(v, v + i)$ of cost $c_i$ only connects from $v$ to $v + i$. To encode two general symmetric circulant edges would require four asymmetric circulant edges.}

Despite substantial structure and symmetry, the complexity of two-stripe circulant TSP had previously remained elusive.    General upper- and lower-bounds for circulant TSP stem from Van der Veen, Van Dal, and Sierksma \cite{VDV91}; Greco and Gerace \cite{Grec07} and Gerace and Greco \cite{Ger08b} focus specifically on the two-stripe TSP, and  prove sufficient (but not necessary) conditions for these upper- and lower-bounds to apply.  Van der Veen, Van Dal, and Sierksma \cite{VDV91} and Gerace and Greco \cite{Ger08} give a general heuristic for circulant TSP that provides a tour within a factor of two of the optimal solution; no improvements to this general heuristic have been made when constrained to the two-stripe version.

In this paper, we take the first step toward resolving the polynomial-time solvability of circulant TSP by showing that symmetric two-stripe circulant TSP is solvable in polynomial time.  We need to be clear on what we mean by ``solvable in polynomial time'' for this problem.  The input is the number $n$ (represented in binary by $O(\log n)$ bits), the indices of the two finite cost edges, and the corresponding costs.  So to run in polynomial time in this case, we should run in time \emph{polylogarithmic} in $n$. We will show how to compute the cost of the optimal tour in $O(\log^2 n)$ time, and thus the decision problem of whether the cost of the optimal tour is at most a bound given as input is solvable in polynomial time; this places the decision version of the problem in the class $P$.   Notice, however, time polynomial in the input size is not sufficient to output the complete sequence of $n$ vertices to visit in the tour.  Nevertheless, given two parameterized classes of tours that we will later describe, we are able to compute in polynomial time to which class the optimal tour belongs, as well as the values of the parameters.

Thus our main contribution is to make a substantial addition to the set of polynomial-time solvable special cases of the TSP \cite{Gil85,Kabadi07}, and to take an important step towards resolving the complexity of the symmetric circulant TSP.

In Section \ref{sec:BG}, we begin by providing background on circulant TSP.  This includes previous work on two-stripe TSP, a useful tool from number theory, and results on the structure of circulant graphs.   In Section \ref{sec:reduction}, we introduce cylinder graphs and reduce the two-stripe TSP to finding certain minimum-cost Hamiltonian paths on cylinder graphs.  We also introduce a class of Hamiltonian paths, called GG paths after work by  Gerace and Greco.  We formally state our algorithm in Section \ref{sec:alg}. Assuming a characterization theorem for Hamiltonian paths between the first and last column in cylinder graphs, we prove the correctness of our algorithm and show that it runs in polynomial time. In Section \ref{sec:mainpf}, we prove the aforementioned characterization theorem. We conclude in Section \ref{sec:conclusion}.

\section{Background and Notation}\label{sec:BG}

In this section, we formalize our notation for the two-stripe TSP.  We also describe the pertinent background results we will use.

Throughout this paper, we use $\equiv_n$ to denote equivalence mod $n$.  All calculations are implicitly mod $n$, unless indicated otherwise. For example, we use $v+a_1$ to denote the vertex reachable from $v$ by following an edge of length $a_1$; the label $v+a_1$ is implicitly taken mod $n$. We also frequently rely on the following basic fact about linear congruences.  
See, e.g., Theorem 57 of Hardy and Wright \cite{hardy_wright}.

\begin{prop}\label{prop:solv}
The linear congruence $$ax\equiv_n b$$ has a solution if and only if $\gcd(a, n)$ divides $b$.  Moreover, there are exactly $\gcd(a, n)$ solutions which take the form 
$$x_0 + \lambda~ \f{n}{\gcd(a, n)}, \, \lambda = 0, 1, ..., \gcd(a, n)-1$$ 
for some $0\leq x_0 < \f{n}{\gcd(a, n)}$.
\end{prop}

\subsection{Circulant Graphs}

Recall that a {\bf circulant graph} is a simple graph whose adjacency matrix is circulant.  In circulant TSP, all edges $\{i, j\}$ such that $i-j\equiv_n k$ or $i-j\equiv_n (n-k)$ have the same cost $c_k.$  Such edges are typically referred to as in the {\bf $k$-th stripe}, or as of {\bf length} $k$.  In the two-stripe TSP, only two of the edge costs $c_1, c_2, ..., c_{\lfloor \f{n}{2}\rfloor}$ are finite.  We refer to those two edge lengths as $a_1$ and $a_2$, so that $0\leq c_{a_1} \leq c_{a_2} < \infty,$ and for $i\notin \{a_1, a_2\},$ $c_i = \infty$. 

We use the following definition to describe circulant graphs including exactly the edges associated with some set of stripes $S$.  In the two-stripe TSP, we are generally interested in subsets $S$ of size 1 or 2. 

\begin{defn}
Let $S\subset \{1, ..., \lfloor \f{n}{2}\rfloor\}.$ The {\bf circulant graph $C\langle S\rangle$} is the (simple, undirected, unweighted) graph including exactly the edges associated  with the stripes $S$.  I.e., the graph with adjacency matrix $$A=(a_{ij})_{i, j=1}^n, \hspace{5mm} a_{ij} = \begin{cases} 1, & (i-j) \bmod n \in S \text{ or } (j-i) \bmod n \in S \\ 0, & \text{ else.} \end{cases}$$
\end{defn}

Provided that $C\langle\{a_1, a_2\}\rangle$ is Hamiltonian, the {\bf two-stripe TSP} is to find a minimum-cost Hamiltonian cycle in the graph $C\langle\{a_1, a_2\}\rangle.$  Since $c_{a_1}\leq c_{a_2},$ the problem is to find a Hamiltonian cycle in $C\langle\{a_1, a_2\}\rangle$ using as few edges of length $a_2$ as possible.

Early work from  Burkard and Sandholzer \cite{Burk91} provides necessary and sufficient conditions for $C\langle\{a_1, a_2\}\rangle$  to be Hamiltonian: $\gcd(n, a_1, a_2)$ must equal 1. That this condition is necessary follows directly: If $\gcd(n, a_1, a_2):=g_2>1$, then for any $n, m \in \Z,$ $na_1+ma_2 \equiv_{g_2} 0.$ Thus any combination of edges of length $a_1$ and $a_2$ (starting at vertex 0) will remain within the vertices $\{v\in [n]: v\equiv_{g_2} 0\},$ which is a strict subset of $[n].$  
We state the result of  Burkard and Sandholzer for general circulant TSP below; we are again primarily interested in the cases where $t\in \{1, 2\}$.

\begin{prop}[Burkard and Sandholzer \cite{Burk91}] \label{prop:Ham} 
Let $\{a_1, ..., a_t\}\subset \left[\lfloor \f{n}{2}\rfloor\right]$ and let $\mathcal{G}=\gcd(n, a_1, ..., a_t).$  The circulant graph $C\langle \{a_1, ..., a_t\} \rangle$ has $\mathcal{G}$ components. The $i$th component, for $0\leq i\leq \mathcal{G}-1$, consists of $n/\mathcal{G}$ nodes 
$$\left\{i+\lambda \mathcal{G} \bmod n:~ 0\leq \lambda \leq \f{n}{\mathcal{G}}-1\right\}.$$
$C\langle \{a_1, ..., a_t\} \rangle$  is Hamiltonian if and only if $\mathcal{G}=1.$
\end{prop}

Throughout this paper, we let  $g_1 := \gcd(n, a_1)$ and $g_2 := \gcd(n, a_1, a_2).$  By Proposition \ref{prop:Ham}, an instance to the two-stripe TSP has a solution with finite cost if and only if $g_2=1.$  

Circulant graphs have a rich structure that allows us to understand $C\langle\{a_1, a_2\}\rangle$ in terms of $C\langle\{a_1\}\rangle.$ First, by Proposition \ref{prop:Ham}, the graph $C\langle\{a_1\}\rangle$ consists of $g_1$ components.  Each component is a cycle of size $\f{n}{g_1}$: Start at some vertex $i$ for $0\leq i<g_1$, and continue following edges of length $a_1$ until reaching $i+\left(\f{n}{g_1}-1\right)a_1$.  Taking one more length-$a_1$ edge wraps back to vertex $i$, since $$i+\f{n}{g_1}a_1 = i + n\f{a_1}{g_1} \equiv_n i$$ (where the final equivalence follows because $g_1$ divides $a_1$ by definition)\footnote{In the special case that $n$ is even and $a_1=n/2,$ we think of each of the $g_1=\gcd(n, n/2)=n/2$ components of $C\langle\{a_1\}\rangle$ as a cycle on a pair of vertices.}.

Provided $g_2=1$, Proposition \ref{prop:Ham} indicates $C\langle\{a_1, a_2\}\rangle$ is Hamiltonian (and therefore connected).  Thus edges of length $a_2$ connect the components of $C\langle\{a_1\}\rangle$.  They do so, moreover, in  an extremely structured manner. Consider, e.g., Figure \ref{fig:ex}, which shows  $C\langle\{a_1, a_2\}\rangle$ for three different two-stripe instances on $n=12$ vertices.  In all cases, $a_1=3$.  Since $g_1=\gcd(12, 3)=3,$ there are three components of $C\langle\{a_1\}\rangle$ in all instances (specifically, $\{0, 3, 6, 9\}, \{1, 4, 7, 10\},$ and $\{2, 5, 8, 11\}$).  The instances are drawn so that each component is in its own column.  While the columns are drawn in slightly different orders for each instance, and while the specific red length-$a_2$ edges connect different pairs vertices, all three instances are extremely structured.  For example, let $v$ be any vertex in the first column.  Regardless of the instance or which vertex $v$ is within the first column, $v+a_2$ is always in the second column.  Similarly, if $v$ is in the second column, $v+a_2$ is always in the third; if $v$ is in the third column, $v+a_2$ is always in the first.

Claim \ref{cm:mergeCi} below, from Gutekunst and Williamson \cite{Gut19b}, makes this  structure precise: Consider the graph $C\langle\{a_1, a_2\}\rangle$ and contract each component of $C\langle\{a_1\}\rangle$ into a single vertex. Then, provided $g_2=1$, the resulting graph is a cycle.  See Lemma 3.6 in  Gutekunst and Williamson \cite{Gut19b} for a more general statement.

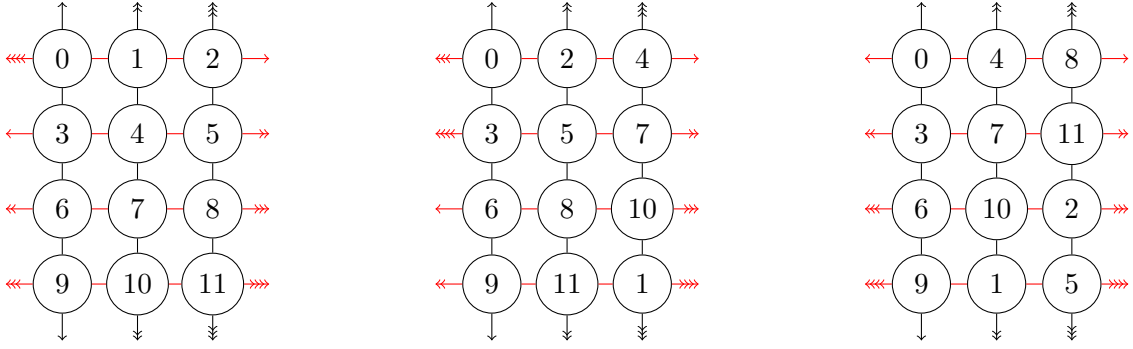
\begin{figure}[t!]
    \centering
        \begin{subfigure}[t]{0.3\textwidth}
        \centering
        \begin{tikzpicture}[scale=1]
        \tikzset{vertex/.style = {shape=circle,draw,minimum size=2em}}
        \vertex (a1) at (0, 3) {$0$};
        \vertex (a2) at (0, 2) {$3$};
        \vertex (a3) at (0, 1) {$6$};
        \vertex (a4) at (0, 0) {$9$};

        \vertex (b1) at (1, 3) {$1$};
        \vertex (b2) at (1, 2) {$4$};
        \vertex (b3) at (1, 1) {$7$};
        \vertex (b4) at (1, 0) {$10$};
        
        \vertex (c1) at (2, 3) {$2$};
        \vertex (c2) at (2, 2) {$5$};
        \vertex (c3) at (2, 1) {$8$};
        \vertex (c4) at (2, 0) {$11$};
        
        \draw (a1) to (a2) to (a3) to (a4);
        \draw (b1) to (b2) to (b3) to (b4);
        \draw (c1) to (c2) to (c3) to (c4);
        
        \draw[red] (a1) to (b1) to (c1);
        \draw[red] (a2) to (b2) to (c2);
        \draw[red] (a3) to (b3) to (c3);
        \draw[red] (a4) to (b4) to (c4);
        
        \draw[->] (a4) to (0, -0.75);
        \draw[->] (a1) to (0, 3.75);
        
        \draw[->>] (b4) to (1, -0.75);
        \draw[->>] (b1) to (1, 3.75);
        
        \draw[->>>] (c4) to (2, -0.75);
        \draw[->>>] (c1) to (2, 3.75);
              
        \draw[->, red] (a2) to (-0.75, 2);
        \draw[->, red] (c1) to (2.75, 3);
              
        \draw[->>, red] (a3) to (-0.75, 1);
        \draw[->>, red] (c2) to (2.75, 2);
                      
        \draw[->>>, red] (a4) to (-0.75, 0);
        \draw[->>>, red] (c3) to (2.75, 1);
        
        \draw[->>>>, red] (a1) to (-0.75, 3);
        \draw[->>>>, red] (c4) to (2.75, 0);
        
        \end{tikzpicture}
        \end{subfigure}
        \hspace{5mm}
 \begin{subfigure}[t]{0.3\textwidth}
        \centering
        \begin{tikzpicture}[scale=1]
        \tikzset{vertex/.style = {shape=circle,draw,minimum size=2em}}
        \vertex (a1) at (0, 3) {$0$};
        \vertex (a2) at (0, 2) {$3$};
        \vertex (a3) at (0, 1) {$6$};
        \vertex (a4) at (0, 0) {$9$};

        \vertex (b1) at (1, 3) {$2$};
        \vertex (b2) at (1, 2) {$5$};
        \vertex (b3) at (1, 1) {$8$};
        \vertex (b4) at (1, 0) {$11$};
        
        \vertex (c1) at (2, 3) {$4$};
        \vertex (c2) at (2, 2) {$7$};
        \vertex (c3) at (2, 1) {$10$};
        \vertex (c4) at (2, 0) {$1$};
        
        \draw (a1) to (a2) to (a3) to (a4);
        \draw (b1) to (b2) to (b3) to (b4);
        \draw (c1) to (c2) to (c3) to (c4);
        
        \draw[red] (a1) to (b1) to (c1);
        \draw[red] (a2) to (b2) to (c2);
        \draw[red] (a3) to (b3) to (c3);
        \draw[red] (a4) to (b4) to (c4);
        
        \draw[->] (a4) to (0, -0.75);
        \draw[->] (a1) to (0, 3.75);
        
        \draw[->>] (b4) to (1, -0.75);
        \draw[->>] (b1) to (1, 3.75);
        
        \draw[->>>] (c4) to (2, -0.75);
        \draw[->>>] (c1) to (2, 3.75);
              
        \draw[->, red] (a3) to (-0.75, 1);
        \draw[->, red] (c1) to (2.75, 3);
              
        \draw[->>, red] (a4) to (-0.75, 0);
        \draw[->>, red] (c2) to (2.75, 2);
                      
        \draw[->>>, red] (a1) to (-0.75, 3);
        \draw[->>>, red] (c3) to (2.75, 1);
        
        \draw[->>>>, red] (a2) to (-0.75, 2);
        \draw[->>>>, red] (c4) to (2.75, 0);
        
        \end{tikzpicture}
        \end{subfigure}
 \hspace{5mm}
 \begin{subfigure}[t]{0.3\textwidth}
        \centering
        \begin{tikzpicture}[scale=1]
        \tikzset{vertex/.style = {shape=circle,draw,minimum size=2em}}
        \vertex (a1) at (0, 3) {$0$};
        \vertex (a2) at (0, 2) {$3$};
        \vertex (a3) at (0, 1) {$6$};
        \vertex (a4) at (0, 0) {$9$};

        \vertex (b1) at (1, 3) {$4$};
        \vertex (b2) at (1, 2) {$7$};
        \vertex (b3) at (1, 1) {$10$};
        \vertex (b4) at (1, 0) {$1$};
        
        \vertex (c1) at (2, 3) {$8$};
        \vertex (c2) at (2, 2) {$11$};
        \vertex (c3) at (2, 1) {$2$};
        \vertex (c4) at (2, 0) {$5$};
        
        \draw (a1) to (a2) to (a3) to (a4);
        \draw (b1) to (b2) to (b3) to (b4);
        \draw (c1) to (c2) to (c3) to (c4);
        
        \draw[red] (a1) to (b1) to (c1);
        \draw[red] (a2) to (b2) to (c2);
        \draw[red] (a3) to (b3) to (c3);
        \draw[red] (a4) to (b4) to (c4);
        
        \draw[->] (a4) to (0, -0.75);
        \draw[->] (a1) to (0, 3.75);
        
        \draw[->>] (b4) to (1, -0.75);
        \draw[->>] (b1) to (1, 3.75);
        
        \draw[->>>] (c4) to (2, -0.75);
        \draw[->>>] (c1) to (2, 3.75);
              
        \draw[->, red] (a1) to (-0.75, 3);
        \draw[->, red] (c1) to (2.75, 3);
              
        \draw[->>, red] (a2) to (-0.75, 2);
        \draw[->>, red] (c2) to (2.75, 2);
                      
        \draw[->>>, red] (a3) to (-0.75, 1);
        \draw[->>>, red] (c3) to (2.75, 1);
        
        \draw[->>>>, red] (a4) to (-0.75, 0);
        \draw[->>>>, red] (c4) to (2.75, 0);
        
        \end{tikzpicture}
        \end{subfigure}
        
    \caption{Three drawings of two-stripe instances on $n=12$ vertices with $a_1=3.$  In the left instance $a_2=1$, in the middle instance $a_2=2,$ and on the right instance $a_2=4$.  Black edges are of length $a_1$ and red edges are of length $a_2$.  Arrows denote horizontal/vertical ``wrapping around,'' and same-colored arrows with the same number of arrowheads wrap to each other.  In all instances, e.g., an edge of length $a_1$ wraps around vertically, connecting $0$ (top-left) to $9$ (bottom-left). Note that the three cases are almost identical, up to the labelling of the vertices and the way that they wrap around horizontally.}
    \label{fig:ex}
\end{figure}

\begin{cm}\label{cm:mergeCi}
Let $C_1,C_2, C_3, ..., C_{g_1}$ denote the connected components of $C\langle\{a_1\}\rangle.$   Consider the directed graph $G'=(V', E')$ on $V'=\left[g_1\right]$  where $(u, v) \in E'$ if and only if there are vertices $x\in C_u$ and $y\in C_v$ with $x-y\equiv_n a_2.$ Then $G'$ is a directed cycle cover. Moreover, if $g_2 = 1$, then $G'$ is a directed cycle.
\end{cm}

\begin{proof}
First, consider any two vertices $x$ and $y$ in the same component $C_i$.  Then $x\equiv_{g_1} y$.  Moreover, $x+a_2 \equiv_{g_1} y+a_2,$ so that $x+a_2$ and $y+a_2$ are in the same component.   Hence, the vertex $i\in V'$ has a single outgoing edge. Analogously $x-a_2\equiv_{g_1} y-a_2$ so that the vertex $i\in V'$ has a single incoming edge.  These facts establish that every vertex of $G'$ has a single outgoing edge and a single incoming edge, so that $G'$ is a directed cycle cover.  However, if $g_2=1$, $G'$ must also be connected. The only connected, directed cycle cover is a directed cycle.
\end{proof}

Claim \ref{cm:mergeCi} suggests a convention for drawing two-stripe circulant graphs (used in  Figure \ref{fig:ex}), where each column corresponds to a component of $C\langle\{a_1\}\rangle.$  We further arrange the columns (and the ordering of vertices within each column) so that edges of $a_2$ can be generally drawn horizontally.  Specifically, we  take the convention that $0$ is the top-left vertex.  The first column will proceed vertically-down as $0, a_1, 2a_1, 3a_1, ...., \left(\f{n}{g_1}-1\right)a_1$.  Our second column ``translates'' the first column by $a_2$, so that the top vertex is $a_2$, the second vertex is $a_1+a_2$, the third vertex is $2a_1+a_2$, and so on.  Our third column translates by another $a_2$, and so on.  See Figure \ref{fig:conv} for our general labeling scheme.

\begin{figure}[t!]

    \centering
        \begin{tikzpicture}[scale=0.85]
        
        \node (a1) at (0, 4) {$0$};
        \node (a2) at (0, 3) {$a_1$};
        \node (a3) at (0, 2) {$2a_1$};
        \node (a4) at (0, 1) {$\vdots$};
        \node (ab) at (0, 0) {$\left(\f{n}{g_1}-1\right)a_1$};

        \node (a1) at (4, 4) {$0+a_2$};
        \node (a2) at (4, 3) {$a_1+a_2$};
        \node (a3) at (4, 2) {$2a_1+a_2$};
        \node (a4) at (4, 1) {$\vdots$};
        \node (ab) at (4, 0) {$\left(\f{n}{g_1}-1\right)a_1+a_2$};

        \node (a1) at (8, 4) {$0+2a_2$};
        \node (a2) at (8, 3) {$a_1+2a_2$};
        \node (a3) at (8, 2) {$2a_1+2a_2$};
        \node (a4) at (8, 1) {$\vdots$};
        \node (ab) at (8, 0) {$\left(\f{n}{g_1}-1\right)a_1+2a_2$};

        \node (a1) at (11, 4) {$\cdots$};
        \node (a2) at (11, 3) {$\cdots$};
        \node (a3) at (11, 2) {$\cdots$};
        \node (a4) at (11, 1) {$\ddots$};
        \node (ab) at (11, 0) {$\cdots$};
        
        \node (a1) at (14, 4) {$0+(g_1-1)a_2$};
        \node (a2) at (14, 3) {$a_1+(g_1-1)a_2$};
        \node (a3) at (14, 2) {$2a_1+(g_1-1)a_2$};
        \node (a4) at (14, 1) {$\vdots$};
        \node (ab) at (14, 0) {$\left(\f{n}{g_1}-1\right)a_1+(g_1-1)a_2$};
            
        
        \end{tikzpicture}    \caption{Our convention for drawing graphs of two-stripe TSP instances.  All vertex labels should be implicitly taken mod $n$.}
    \label{fig:conv}
\end{figure}

Returning to Figure \ref{fig:ex}, recall that the three instances have $n=12$ and $a_1=3.$   In the left instance $a_2=1$, in the middle instance $a_2=2,$ and in the right instance $a_2=4.$   Up to a re-labelling of vertices, the only structural change occurs in how the right-most column is connected to the left-most column: the edges that wrap around horizontally, connecting the last column back to the first.  On the left, taking an edge of length $a_2$ from a vertex in the last column wraps around to the first column, but one row lower (2 in the top row connects to 3 in the second row, 5 in the second row connects to 6 in the third row, etc).  In the middle, wrapping around shifts down two rows (4 in the top row connects to 6 in the third row, etc).  On the right, edges wrap back to the same row.

\subsection{Cylinder Graphs}
Part of the challenge of two-stripe TSP is the differing ways that horizontal edges wrap around between the last and first column.  Our first result, in Section \ref{sec:reduction}, removes this difficulty, and allows us to work on ``cylinder graphs'': graphs that are similar to those  in Figure \ref{fig:ex}, but without any horizontal edges wrapping around between the last and last column.  
\begin{defn}
Let $n=r\times c.$  A {\bf $r\times c$ cylinder graph} is a graph with $n$ vertices arranged into an $r\times c$ grid.  For $0\leq i\leq r-1$ and $0\leq j\leq c-1$, a vertex in row $i$ and column $j$ is adjacent to:
\begin{itemize}
\item The vertex in row $i$ and column $j+1$, provided $j<c-1$,
\item The vertex in row $i$ and column $j-1$, provided $j>0$,
\item The vertex in row $i-1 \mod r$ and column $j$, and the vertex in row $i+1\mod r$ and column $j$.
\end{itemize}
\end{defn}

\noindent See, e.g., Figure \ref{fig:cyl}.  It will often be helpful to refer to a vertex in row $x$ and column $y$ as $(x, y)$, indexed by its row and column, starting from $0$.  We take the convention that the top-left vertex is $(0, 0)$. Hence, the bottom-right vertex is $(r-1, c-1)$. In general, a vertex $v = xa_1 + ya_2$ where $0 \leq x \leq (r-1)$ and $0 \leq y \leq (c-1)$ will have cylindrical coordinates $(x, y)$. We similarly treat the cylinder graph as inheriting the cost structure of a two-stripe circulant instance. That is, all the vertical edges have cost $c_{a_1}$, and all the horizontal edges have cost $c_{a_2}$. 

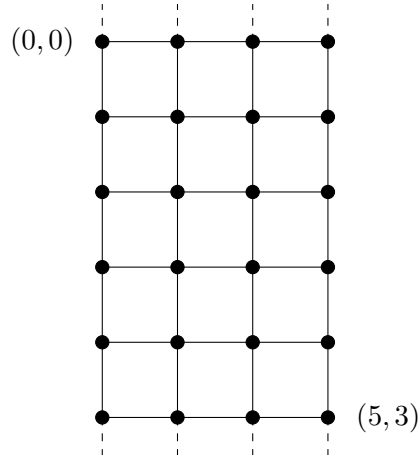
\begin{figure}[t!]
    \centering
        \centering
        \begin{tikzpicture}
        \foreach \x in {0,1,...,3} {
            \foreach \y in {0,1,...,5} {
                \vertex (\x\y) at (\x, \y) {};
            }
        }   
        \foreach \i in {0, 1, 2} {
            \foreach \j in {0, 1, ..., 4} {
            \draw (\i, \j) to (\i, \j+1);
            \draw (\i, \j) to (\i+1, \j);}
        }
        
        \foreach \i in {0, 1, 2} {
         \draw (\i, 5) to (\i+1, 5);}
         
        \foreach \j in {0, 1, ..., 4} {
         \draw (3, \j) to (3, \j+1);}
         
        \foreach \i in {0, 1, 2, 3} {
         \draw[dashed] (\i, 5) to (\i, 5.5);
         \draw[dashed] (\i, 0) to (\i, -0.5);}
         
         \node (b) at (-0.8, 5) {$(0, 0)$};
         \node (b) at (3.8, 0) {$(5, 3)$};
        \end{tikzpicture}

    \caption{A $6\times4$ cylinder graph with vertices $(0, 0)$ and $(5, 3)$ indicated. Dashed edges wrap around vertically to create a cylindrical structure.}
    \label{fig:cyl}
\end{figure}

\subsection{Conventions for the Two-Stripe TSP}
We note several cases where the two-stripe TSP is trivial:
\begin{itemize}
\item If $g_2 >1,$ Proposition \ref{prop:Ham} implies that $G\langle\{a_1, a_2\}\rangle$ is not Hamiltonian (and thus no Hamiltonian cycles exist).
\item If $g_1 =1$, then, by Proposition \ref{prop:Ham}, the graph $G\langle\{a_1\}\rangle$  is Hamiltonian.  Hence, since $c_{a_1} \leq c_{a_2},$ a cheapest Hamiltonian cycle costs $nc_{a_1}$ and consists of a Hamiltonian cycle on $G\langle\{a_1\}\rangle$.  Note that any time $n$ is prime, $g_1=1$, so that two-stripe (and indeed, general circulant) TSP is trivial any time $n$ is prime.
\item If $c_{a_1}=c_{a_2}$ and $g_2=1$, then $G\langle\{a_1, a_2\}\rangle$  is Hamiltonian and any Hamiltonian cycle costs $nc_{a_1}=nc_{a_2}.$ 
\end{itemize}
These observations are also made in  Greco and Gerace \cite{Grec07} and Gerace and Greco \cite{Ger08b}. Hence, we restrict ourselves to the case where $g_1>g_2 = 1$ and $c_{a_1}<c_{a_2}.$  

Following the example of Figure \ref{fig:ex}, we think of the ``expensive'' length-$a_2$ edges as  ``horizontal edges.''  The ``cheap'' length-$a_1$ edges are analogously considered ``vertical edges.''  Then the goal is to minimize the total number of expensive horizontal edges.  From here on we assume without loss of generality that $c_{a_1}=0$ and $c_{a_2}=1$, so every vertical edge has cost 0 and every horizontal edge has cost 1.

\subsection{Main Result}
The main result of our paper, stated below, is a characterization of the optimal tour in any two-stripe circulant TSP instance. Section \ref{sec:main_result_proof} contains its proof, and Section \ref{sec:alg} shows how this result naturally leads to a polynomial-time algorithm for two-stripe TSP. 
\begin{restatable}{thm}{result}
\label{thm:result}
Let $r = \frac{n}{g_1}$ and $c = g_1$. Suppose the cylindrical coordinates of $-a_2$ are $(x, c-1)$. 
Let $m^*$ be the smallest integer value of $m$ such that $m \geq -\frac{c}{2}$, and $x \in \{2m+c, -(2m+c)\} \mod r$.  Then:
\begin{itemize}
    \item If $m^* \leq 0$, the cost of the optimal tour is $c$.
    \item If $0 < 2m^* < c - 2$, the cost of the optimal tour is $c+2m^*$. 
    \item If $2m^* \geq c - 2$ or $m^*$ does not exist, the cost of the optimal tour is $2c - 2$.
\end{itemize}
\end{restatable}

\subsection{Previous Results on Two-Stripe TSP}
Despite the restrictive structure of two-stripe TSP, many of the main previous results for it are inherited from the more-general circulant TSP.   From an approximation algorithms perspective, the state-of-the-art remains a heuristic from Van der Veen, Van Dal, and Sierksma \cite{VDV91}; on any two-stripe instance, this heuristic provides a tour of cost within a factor of two of the optimal cost (see also Gerace and Greco \cite{Ger08}, for a more general 2-approximation algorithm).  The performance guarantee uses a combinatorial lower bound of  Van der Veen, Van Dal, and Sierksma \cite{VDV91}, which when specialized to the two-stripe case states that any Hamiltonian cycle must cost at least $g_1$; this is clear, since there must be at least $c=g_1$ edges to join the $c$ columns in a cycle. If a tour of cost $g_1$ exists, we will refer to it as a ``lower bound tour.''

\begin{prop}[Van der Veen, Van Dal, and Sierksma \cite{VDV91}, specialized to two-stripe TSP]\label{prop:LB}
Any Hamiltonian cycle on a two-stripe instance costs at least $g_1.$
\end{prop}


We can  exhibit an upper bound of $2(g_1-1)$ on the optimal solution to any two-stripe instance by providing a feasible tour (as usual, provided that $g_2=1$); such tours are illustrated in Figure \ref{fig:UB}. Henceforth, we will refer to these tours  as ``upper bound tours.''  Such tours immediately give a 2-approximation algorithm: they provides a feasible solution costing at most $2(g_1-1)<2g_1,$ while $g_1$ is a lower bound on the optimal cost.

\begin{prop}\label{prop:UB}
There exists a Hamiltonian cycle of cost $2(g_1-1)$ on any two-stripe instance.
\end{prop}



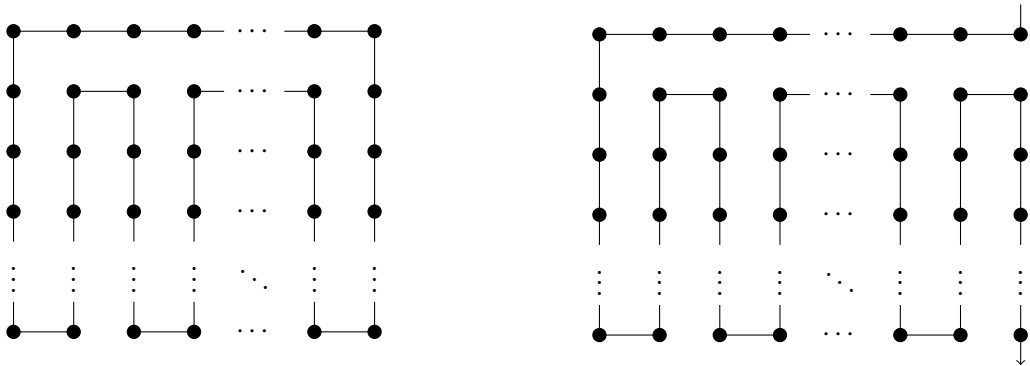
\begin{figure}[t!]
    \centering
        \begin{subfigure}[t]{0.45\textwidth}
        \centering
        \begin{tikzpicture}[scale=0.8]
        \foreach \x in {0,1,2, 3, 5, 6} {
            \foreach \y in {0, 2, 3, 4, 5} {
                \vertex (\x\y) at (\x, \y) {};
            }
        }   
        
        \foreach \x in {0, 1, 2, 3, 5, 6} {\node (\x) at (\x, 1) {$\vdots$};}
        \foreach \y in {0, 2, 3, 4, 5} {\node (\y) at (4, \y) {$\cdots$};}
        \node (d) at (4, 1) {$\ddots$};
        
        \foreach \x in {0, 1, 2, 3, 5, 6} {
        \draw (\x, 4) to (\x, 3) to (\x, 2) to (\x, 1.5);
        \draw (\x, 0.5) to (\x, 0);}
        
        \draw (0, 0) to (1, 0);
        \draw (2, 0) to (3, 0);
        \draw (5, 0) to (6, 0);
        \draw (0, 4) to (0, 5) to (1, 5) to (2, 5) to (3, 5) to (3.5, 5);
        \draw (4.5, 5) to (5, 5) to (6, 5) to (6, 4);
        \draw (1, 4) to (2, 4);
        \draw (3, 4) to (3.5, 4);
        \draw (4.5, 4) to (5, 4);
        \node (filler) at (0, -0.4) {};
        \end{tikzpicture}
        \end{subfigure}
        \hspace{5mm}
         \begin{subfigure}[t]{0.45\textwidth}
        \centering
        \begin{tikzpicture}[scale=0.8]
        \foreach \x in {0,1,2, 3, 5, 6, 7} {
            \foreach \y in {0, 2, 3, 4, 5} {
                \vertex (\x\y) at (\x, \y) {};
            }
        }   
        
        \foreach \x in {0, 1, 2, 3, 5, 6, 7} {\node (\x) at (\x, 1) {$\vdots$};}
        \foreach \y in {0, 2, 3, 4, 5} {\node (\y) at (4, \y) {$\cdots$};}
        \node (d) at (4, 1) {$\ddots$};
        
        \foreach \x in {0, 1, 2, 3, 5, 6, 7} {
        \draw (\x, 4) to (\x, 3) to (\x, 2) to (\x, 1.5);
        \draw (\x, 0.5) to (\x, 0);}
        
        \draw (0, 0) to (1, 0);
        \draw (2, 0) to (3, 0);
        \draw (5, 0) to (6, 0);
        \draw (0, 4) to (0, 5) to (1, 5) to (2, 5) to (3, 5) to (3.5, 5);
        \draw (4.5, 5) to (5, 5) to (6, 5);
        \draw (1, 4) to (2, 4);
        \draw (3, 4) to (3.5, 4);
        \draw (4.5, 4) to (5, 4);
        \draw[->] (7, 0) to (7, -0.5);
        \draw (7, 5.5) to (7, 5);
        \draw (6, 4) to (7, 4);
        \draw (6, 5) to (7, 5);
        \end{tikzpicture}
        \end{subfigure}

    \caption{Feasible Hamiltonian cycles of cost $2(g_1-1)$ when $g_1$ is even (left) and odd (right).}
    \label{fig:UB}
\end{figure}

More specific results on the two-stripe TSP come from  Greco and Gerace \cite{Grec07} and Gerace and Greco \cite{Ger08b}.  Both papers provide a sufficient, number-theoretic condition for  a tour of cost $g_1$ to exist. In Theorem \ref{thm:result}, we strengthen their result to a necessary and sufficient condition for a tour of cost $g_1$ to exist.

\begin{prop}[Greco and Gerace \cite{Grec07} and Gerace and Greco \cite{Ger08b}]\label{prop:getLB}
Consider a two-stripe instance, and suppose that there exists an integer $y$, with $0\leq y\leq g_1$, such that $$(2y-g_1)a_1+g_1a_2 \equiv_n 0.$$  Then the instance admits a Hamiltonian cycle of cost $g_1.$ 
\end{prop}

In this proof, and throughout this paper, it is often helpful to treat our tours as arbitrarily directed and starting at 0.  Doing so allows us to differentiate between an edge of length $+a_i$ (proceeding from $v$ to $v+a_i$) or $-a_i$ (proceeding from $v$ to $v-a_i$).  Recall that vertex labels $v \pm a_i$ are implicitly understood to be taken mod $n$.

\begin{proof}
Suppose such a $y$ exists.  We create a Hamiltonian path as follows: Start at $0.$  In each of the first $y$ columns, we take edges of length $-a_1$ until every vertex in that column is visited (using $\f{n}{g_1}-1$ total edges of length $-a_1$ in that column), and then take an edge of length $+a_2$ to move to the next column.  In any columns remaining after the first $y$, take edges of length $+a_1$ until every vertex in that column is visited (using $\f{n}{g_1}-1$ total edges of length $+a_1$), and take an edge of length $+a_2$ to move to the next column.  We finish this process in the last column, after using $g_1-1$ total horizontal edges.  See Figure \ref{fig:getLB}. 

By construction, this path is Hamiltonian.  The final vertex reached is computed as follows: we used $(g_1-1)$ edges of length $+a_2$ to get to the final column, we used $\f{n}{g_1}-1$ edges of length $-a_1$ in each of $y$ columns, and we used  $\f{n}{g_1}-1$ edges of length $+a_1$ in the remaining $g_1-y$ columns.  Hence, we end at:
\begin{align*}
(g_1-1)a_2 - \left(\f{n}{g_1}-1\right)y a_1 + (g_1-y)\left(\f{n}{g_1}-1\right)a_1
&= (g_1-1)a_2 -2y \left(\f{n}{g_1}-1\right) a_1 + g_1 \left(\f{n}{g_1}-1\right)a_1 \\
&=(g_1-1)a_2 +2ya_1 - 2yn\f{a_1}{g_1} + (n-g_1)a_1 
\intertext{Recalling that $g_1$ divides $a_1$:}
&\equiv_n (g_1-1)a_2 +2ya_1  -g_1a_1 \\
&= (g_1-1)a_2 + (2y-g_1)a_1.
\end{align*}
We now extend our Hamiltonian path from vertex $(g_1-1)a_2 + (2y-g_1)a_1$ by taking a horizontal edge $+a_2$.  By the assumed conditions on $y$,
$$(g_1-1)a_2 + (2y-g_1)a_1 + a_2 = g_1 a_2 + (2y-g_1)a_1 \equiv_n 0.$$  Hence, this final edge extends our Hamiltonian path to a Hamiltonian cycle, and since it adds one horizontal edge, the resultant cycle costs $g_1$.
\end{proof}

\begin{figure}[t!]
    \centering
        \centering
        \begin{tikzpicture}[decoration={
    markings,
    mark=at position 0.5 with {\arrow{>}}}]
        
        \foreach \x in {0,1,...,3} {
            \foreach \y in {0,1,...,4} {
                \vertex (\x\y) at (\x, \y) {};
            }
        }   
        \draw[->] (0, 4) to (0, 4.5);
        \draw[->] (1, 4) to (1, 4.5);
        \draw[->] (2, 4) to (2, 4.5);
        
        \draw[postaction={decorate}](0, -0.5) to (0, 0);
        \draw[postaction={decorate}]  (0, 0) to (0, 1);
        \draw[postaction={decorate}]  (0, 1) to (0, 2);
        \draw[postaction={decorate}]  (0, 2) to (0, 3);
        
        \draw[postaction={decorate}]  (0, 3) to (1, 3);
        
        \draw[postaction={decorate}](1, -0.5) to (1, 0);
        \draw[postaction={decorate}]  (1, 0) to (1, 1);        
        \draw[postaction={decorate}]  (1, 1) to (1, 2);  
        \draw[postaction={decorate}]  (1, 3) to (1, 4);

        \draw[postaction={decorate}](2, -0.5) to (2, 0);
        \draw[postaction={decorate}]  (2, 0) to (2, 1);        
        \draw[postaction={decorate}]  (2, 2) to (2, 3);  
        \draw[postaction={decorate}]  (2, 3) to (2, 4);
        
        \draw[postaction={decorate}](1, 2) to (2, 2);
        \draw[postaction={decorate}](2, 1) to (3, 1);
        
        \draw[postaction={decorate}](3, 1) to (3, 0);
        \draw[->] (3, 0) to (3, -0.5);
        \draw[postaction={decorate}](3, 4.5) to (3, 4);
        \draw[postaction={decorate}](3, 4) to (3, 3);
        \draw[postaction={decorate}](3, 3) to (3, 2);                
        \end{tikzpicture}

    \caption{The idea behind Proposition \ref{prop:getLB}.  The integer $y$ corresponds to the number of columns where arrows point up (corresponding to $-a_1$ edges), so that in this picture, $y=3.$}
    \label{fig:getLB}
\end{figure}
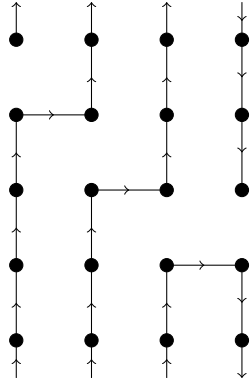

Greco and Gerace \cite{Grec07} and Gerace and Greco \cite{Ger08b} make two other observations that will be useful to us.  These results relate the number of $+a_2$ and $-a_2$ edges used in an optimal Hamiltonian cycle.  Let $a_+$ and $a_-$ respectively denote the number of $+a_2$ and $-a_2$ edges used in an optimal Hamiltonian cycle (again, directed arbitrarily).  Note that the cost of the tour is $a_++a_-.$ 

\begin{prop}[Greco and Gerace \cite{Grec07} and Gerace and Greco \cite{Ger08b}]\label{prop:alpha}
There exists a minimum-cost Hamiltonian cycle in which $a_+ - a_- \in \{0, g_1\}.$
\end{prop}

\begin{proof}
Consider the graph of $C\langle\{a_1\}\rangle.$  Recall that edges of length $\pm a_2$ connect one connected component of this graph to another. Following our convention for drawing graphs of two-stripe instances, each connected components corresponds to a column of the two-stripe graph.  Edges of length $+a_2$ join a column to the column on its right (wrapping around from the last column to the first if need be), whereas edges of length $-a_2$ join a column to  the column on its left (wrapping around from the first column to the last if need be).  Since any Hamiltonian cycle starts and ends in the same column, we have that $a_+ - a_- \equiv_{g_1} 0.$  Any Hamiltonian cycle where $|a_+ - a_-| > g_1$ must cost at least $2g_1$ and cannot be optimal by Proposition \ref{prop:UB}.  Hence, in a minimum-cost Hamiltonan cycle, $a_+ - a_- \in \{0, g_1, -g_1\}.$  If $a_+ - a_- = -g_1,$ the tour can be reversed to attain a tour of the same cost, where $a_+-a_- = g_1.$ 
\end{proof}

\begin{prop}[Greco and Gerace \cite{Grec07} and Gerace and Greco \cite{Ger08b}]\label{prop:forcedUB}
If $a_+ = a_-$ in a minimum-cost Hamiltonian cycle, then the minimum-cost Hamiltonian cycle costs $2(g_1-1).$ 
\end{prop}

\begin{proof}
We analogize an argument from Claim 4.3 of Gutekunst and Williamson \cite{Gut19b}. Consider the columns of our circulant graph (i.e., the components of $C\langle\{a_1\}\rangle$).  Each edge of length $+a_2$ moves one column to the right, while edges of length $-a_2$ move us one column to the left (in both cases, wrapping from column $g_1-1$ to $0$ or vice versa when necessary).   A Hamiltonian tour must visit all columns, and the only way we move between columns is using edges of length $+a_2$ and $-a_2$.  If we view the Hamiltonian cycle as a sequence of edges $e_1, ...., e_n$, we can delete all the edges of length $\pm a_1$ to attain a subsequence $L'=(e_{i_1}, ..., e_{i_{\ell}})$ where $i_1<i_2<\cdots i_{\ell}, e_{i_j}\in \{\pm a_2\}.$ We upper bound the number of columns visited as follows: Set $U=1,$ corresponding to starting at some column.  Until $L'$ is either all $+a_2$s or all $-a_2$s, find an occurrence of a  $+a_2$ followed by a $-a_2$ in $L'$ (or a  $-a_2$ followed by an $+a_2$); delete these two elements and increment $U$ by 1.  Once this process terminates, increment $U$ by $|L'|$ (the number of $+a_2$s or $-a_2$s remaining when $L'$ is either all $+a_2$s or all $-a_2$s).  Note that, at the end, $U=\max\{a_+, a_-\}+1.$  $U$ provides an upper bound on the number of columns visited: Any time an $+a_2$ is followed by a $-a_2$ in $L'$ (or vice versa), the effect is to move to one adjacent column and then move back. Hence we visit at most one new column.  Thus in order to visit the $g_1$ columns, we need $U=g_1$, so if $a_+ = a_-$, then we must have $a_+=a_-=g_1 -1$.
\end{proof}

Based on the above results, Greco and Gerace \cite{Grec07} and Gerace and Greco \cite{Ger08b} are able to identify a few types of instances for which two-stripe TSP is solvable.  For example, if $g_1=2$, then the upper and lower bounds match, so that an upper bound tour is optimal.  Finally, they present their main theorem, which provides a  sufficient (but not necessary) condition for the upper bound to be optimal, and constructs an additional tour. We state it without proof below, as we will provide a full resolution to two-stripe TSP that does not rely on it.  

\begin{prop}[Theorem 4.4 in  Gerace and Greco \cite{Ger08b}]\label{prop:GGFull}
Let $$S=\left\{y: 0\leq y < \f{n}{g_1}, \, (2y-g_1)a_1+g_1a_2 \equiv_n 0\right\}.$$  If $S$ is empty, the upper bound of $2(g_1-1)$ is optimal.  Otherwise, let $y_1=\min\{S\}$ and $y_2 = \max\{S\}.$  Let $m=\min\{y_1 - g_1, \f{n}{g_1}-y_2\}.$  If $m\leq 0$, the conditions of Proposition \ref{prop:getLB} are satisfied and the lower bound of $g_1$ is optimal.  Otherwise, there exists a tour of cost $g_1+2m.$
\end{prop}

Note that the congruence defining $S$ is the same as in Proposition \ref{prop:getLB}.  Hence, if $\f{n}{g_1}-1\leq g_1$, this theorem implies that the optimal solution will always either be the upper-bound tour (if $S$ is empty) or the lower-bound tour (if $S$ is nonempty, in which case any element must be at most $g_1$ and $m$ will be non-positive).  

To prove sufficiency of the upper-bound condition, Greco and Gerace do a small amount of algebra (in the style of Proposition \ref{prop:getLB}'s proof): For the upper bound to be non-optimal, there must exist a tour where $a_+-a_- = g_1$. Greco and Gerace argue that any such tour implies that $S$ is non-empty. 

Ultimately, the main results of  Greco and Gerace \cite{Grec07} and Gerace and Greco \cite{Ger08b} provide sufficient (but not necessary) condition for the upper- and lower-bounds in Propositions \ref{prop:LB} and \ref{prop:UB}. Thus they provide a new, algebraic test that can identify certain instances where the upper- or lower-bounds occur, but leave as open fully characterizing either extreme, as well as determining  the optimal solution to any instance where the optimal value is between those bounds. 

Finally,  Greco and Gerace \cite{Grec07} and Gerace and Greco \cite{Ger08b} conjecture that, in cases where the lower bound is not achievable, the optimal solution is either the upper bound or their tours of cost $g_1+2m$ identified in their main theorem.  Our resolution of the two-stripe TSP confirms this conjecture.  We show how Theorem \ref{thm:main} implies Proposition \ref{prop:GGFull} in Appendix \ref{app:weimplygg}.

\section{GG Paths and a Reduction to Hamiltonian Paths on Cylinder Graphs}
\label{sec:reduction}
In this section, we first introduce specific Hamiltonian paths on cylinder graphs, which we will call \emph{GG paths}. We name these paths after Gerace and Greco (\cite{Ger08b}, \cite{Grec07}), who conjectured that they are sufficient to describe an optimal tour in settings where the upper bound is not optimal. We then present our first main theorem, Theorem \ref{thm:reduction}. This theorem will be the first step in solving the two-stripe TSP. Specifically, Theorem \ref{thm:reduction} reduces the two-stripe TSP to finding a minimum-cost Hamiltonian path (between two specified vertices) on an associated cylinder graph. It will then remain to argue that the optimal Hamiltonian path (between those two specified vertices) on that cylinder graph will always be a GG path.

Theorem \ref{thm:reduction} shows that, to solve two-stripe TSP instances, it suffices to study Hamiltonian paths on cylinder graphs that start in the first column and end in the last column. Without loss of generality, we will assume (and draw) such paths as starting at the top-left vertex: the vertex $(0, 0).$  The cost of such a path will be determined by the number of horizontal edges used.  Theorem \ref{thm:main} will then be used to argue that, among such all 
Hamiltonian  paths (on a cylindrical graph, starting at $(0, 0)$ and ending at a particular vertex in the last column), there is always an extremely structured optimal Hamiltonian path: a \emph{GG path}.  

We will formally define GG paths in Section \ref{s:GG}.  Informally speaking, a GG path is a path where ``all of the funky business'' happens between the first and second column:  For a Hamiltonian path to start in the first column and end in the last column, it must use an odd number of edges between every pair of consecutive columns.  In a GG path, an odd number of horizontal edges are used between the first and second column, but then a single edge is used between every remaining pair successive columns.  Moreover, if a GG path uses $2m+1$ horizontal edges between the first and second column, it will first use $r-2m-1$ vertical edges, and will then alternate between vertical and horizontal edges until all vertices in the first column have been visited.  See Figure \ref{fig:GG3}.

\begin{figure}[t!]

    \centering
        \begin{subfigure}[t]{0.45\textwidth}
        \centering
        \begin{tikzpicture}[scale=0.65]
        \foreach \x in {0,1} {
            \foreach \y in {0,1,...,8} {
                \vertex (\x\y) at (\x, \y) {};
            }
        }   
        \foreach \i in {0, 1} {
            \foreach \j in {3, 4, ..., 7} {
            \draw (\i, \j) to (\i, \j+1);}
        }
        
        \foreach \j in {0, 1, 2} {
        \draw (0, \j) to (1, \j);}
        
        \draw (0, 2) to (0, 3);
        \draw (0, 0) to (0, 1);
        \draw (1, 1) to (1, 2);
        \draw [->]  (1, 0) to (1, -0.5);
        \draw    (1, 8) to (1, 8.5);
        
        \end{tikzpicture}
        \end{subfigure}
        \hspace{5mm}
        \centering
        \begin{subfigure}[t]{0.45\textwidth}
        \centering
        \begin{tikzpicture}[scale=0.65]
        \foreach \x in {0,1} {
            \foreach \y in {6, 7,...,14} {
                \vertex (\x\y) at (\x, \y) {};
            }
        }   
        \foreach \i in {0, 1} {
            \foreach \j in {6, 7, ..., 9} {
            \draw (\i, \j) to (\i, \j+1);}
        }
        
        \foreach \j in {11, 12, 13} {
        \draw (0, \j) to (1, \j);}
        
        \draw (0, 10) to (0, 11);
        \draw [->] (0, 14) to (0, 14.5);
        \draw (0, 6) to (0, 5.5);
        \draw (0, 12) to (0, 13);
        \draw [->] (1, 14) to (1, 14.5);
        \draw (1, 6) to (1, 5.5);       
        \draw (1, 11) to (1, 12);
        \draw (1, 13) to (1, 14);

        \end{tikzpicture}
        \end{subfigure}
        
    \caption{The start of GG paths using 3 horizontal edges between the first and second column.  Note that there are two choices: the first edge goes from $(0, 0)$ to $(1, 0)$ or wraps vertically from $(0, 0)$ to $(r-1, 0)$.}
    \label{fig:GG3}
\end{figure}
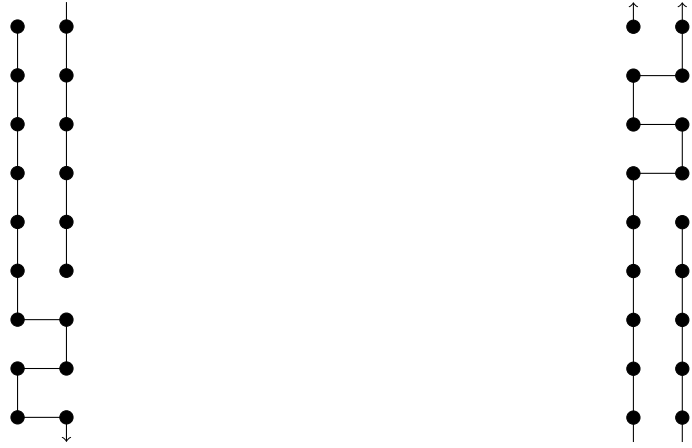

\subsection{Reduction of 2-stripe problem to cylinder graph problem}\label{s:red1}

We first reduce two-stripe TSP instances to Hamiltonian paths on cylinder graphs that start in the first column and end in the last column.  By Proposition \ref{prop:UB}, we can always attain a tour of cost $2(g_1-1)$. Theorem \ref{thm:reduction} says that we can find a cheaper tour costing $k<2(g_1-1)$ if and only if there there is a corresponding  Hamiltonian path on a cylinder graph using $k-1$ horizontal edges.  In the theorem, we use $x$ to denote the row of the vertex $-a_2$; a Hamiltonian path from $0$ to $-a_2$ on the cylinder graph can readily be converted to a Hamiltonian cycle by adding a single length-$a_2$ edge.


\begin{thm}
\label{thm:reduction}
Let $k<2(g_1-1)$ and let $x$ be the unique solution to $$xa_1 \equiv_n -g_1 a_2$$ in $\{0, 1, ..., \f{n}{g_1}-1\}.$  Consider a 2-stripe instance with $n$ vertices and edges of length $a_1$ and $a_2.$  Then there exists a Hamiltonian cycle costing $k$  if and only if there is a Hamiltonian path on an $r \times c$ cylinder graph with $c=g_1$ and $r=\f{n}{g_1}$ using $k-1$ horizontal edges, starting at $(0, 0)$, and ending at $(x, c-1).$
\end{thm}

The idea of the proof is that a tour of cost less than $2(g_1-1)=2(c-1)$ must have some pair of consecutive columns such that only one horizontal edge of the tour crosses between the two columns.  We can ``shift'' the tour so that this edge wraps around from column $c-1$ to 0, and so that this edge passes through $(0,0)$.  Removing the wraparound edge leaves a Hamiltonian path from (0,0) to $(x,c-1)$, where the endpoint is equivalent to the vertex $-a_2$. The next two propositions make this precise.


\begin{prop}\label{prop:a1shift}
Let $H=(v_0, v_1, v_2, ..., v_{n-1}, v_n=v_0)$ be a Hamiltonian cycle in $C\langle\{a_1, a_2\}\rangle.$  Let $H'=(u_0, u_1, ..., u_{n-1}, u_n=0)$ be the Hamiltonian cycle where $u_i=v_i+a_1.$  Then  for any two columns $C_s$ and $C_{s+1\mod g_1}$, the number of horizontal edges between $C_s$ and $C_{s+1 \mod g_1}$ is the same in both $H$ and $H'$.  Moreover, if a horizontal edge is in row $k$ of $H$, then it is in row $k+1 \mod r$ of $H'$.
\end{prop}

See, e.g, Figure \ref{fig:a1shifts}, where  every vertex in left cycle $H$ has its label shifted by $a_1$ to get the right cycle (e.g. the edge between 15 and 18 is now between 20 and 23, one row lower).  Recall also that the columns of $C\langle\{a_1, a_2\}\rangle$ are the components of $C\langle\{a_1\}\rangle.$

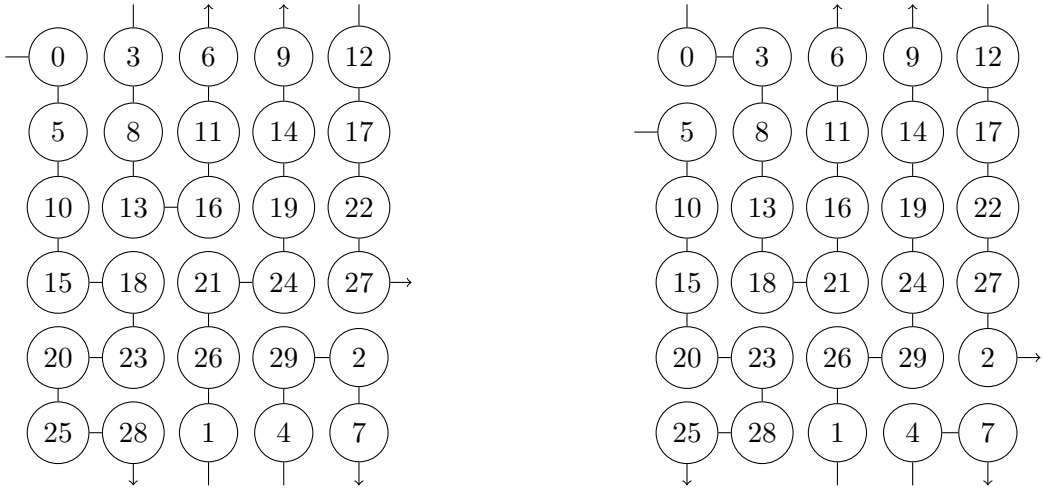
\begin{figure}[t!]
    \centering
    \begin{subfigure}[t]{0.4\textwidth}
        \centering
        \begin{tikzpicture}[scale=1]
        \tikzset{vertex/.style = {shape=circle,draw,minimum size=2em}}
        \foreach \x in {0,1,...,4} {
            \foreach[evaluate = {\l = int(mod(\x*3+(5-\y)*5, 30))}] \y in {0,1,...,5}{
                \vertex (\x\y) at (\x, \y) {\l};
            }
        }

        \draw[->] (10) to (1, -0.7);
        \draw (1, 5.7) to (15);        
        \draw[->] (40) to (4, -0.7);
        \draw (4, 5.7) to (45);
        
        \draw (2, -0.7) to (20);
        \draw[->] (25) to (2, 5.7);
        \draw (3, -0.7) to (30);
        \draw[->] (35) to (3, 5.7);
        
        \draw (05) to (04) to (03) to (02) to (12) to (11) to (01) to (00) to (10);
        \draw (15) to (14) to (13) to (23) to (24) to (25);
        \draw (20) to (21) to (22) to (32) to (33) to (34) to (35);
        \draw (30) to (31) to (41) to (40);
        \draw (45) to (44) to (43) to (42);
        \draw[->] (42) to (4.7, 2);
        \draw (-0.7, 5) to (05);
       
        \end{tikzpicture}
        \end{subfigure}        
        \hspace{15mm}
        \begin{subfigure}[t]{0.4\textwidth}
        \centering
        \begin{tikzpicture}[scale=1]
        \tikzset{vertex/.style = {shape=circle,draw,minimum size=2em}}

        \foreach \x in {0,1,...,4} {
            \foreach[evaluate = {\l = int(mod(\x*3+(5-\y)*5, 30))}] \y in {0,1,...,5}{
                \vertex (\x\y) at (\x, \y) {\l};
            }
        }   
        
        \draw (05) to (15) to (14) to (13) to (12) to (22) to (23) to (24) to (25);
        \draw[->] (25) to (2, 5.7);
        \draw (2, -0.7) to (20);
        \draw (20) to (21) to (31) to (32) to (33) to (34) to (35);
        \draw[->] (35) to (3, 5.7);
        \draw (3, -0.7) to (30);
        \draw (30) to (40);
        \draw[->] (40) to (4, -0.7);
        \draw (4, 5.7) to (45);
        \draw (45) to (44) to (43) to (42) to (41);
        \draw[->] (41) to (4.7, 1);
        \draw (-0.7, 4) to (04);
        \draw (04) to (03) to (02) to (01) to (11) to (10) to (00);
        \draw[->] (00) to (0, -0.7);
        \draw (0, 5.7) to (05);
        \end{tikzpicture}
        \end{subfigure}        
    \caption{The left image is a Hamiltonian cycle in a two-stripe instance with $n=30, a_1=5$ and $a_2=3.$  The right shows the effect of adding $a_1$ to every vertex label (so that, e.g., the edge between $15$ and $18$ becomes an edge between $20$ and $23$). }
    \label{fig:a1shifts}
\end{figure}

\begin{proof}
The first statement follows from the fact that, $$u_i = v_i + a_1 \equiv_{g_1} v_i,$$ since $g_1$ divides $a_1.$  Hence $u_i$ and $v_i$ are in the same column.  Further, by our drawing convention, $u_i$ is exactly one row below $v_i$ (wrapping around from row $r-1$ to row $0$ if need be). Note also that any edge from $u_i$ to $u_{i+1}$ has the same length as the edge from $v_i$ to $v_{i+1}$, so that this shift does not change the total number of horizontal/vertical edges.
\end{proof}

\begin{prop}\label{prop:a2shift}
Let $H=(v_0, v_1, v_2, ..., v_{n-1}, v_n=v_0)$ be a Hamiltonian cycle in $C\langle\{a_1, a_2\}\rangle.$  Let $H'=(u_0, u_1, ..., u_{n-1}, u_n=0)$ be the Hamiltonian cycle where $u_i=v_i+a_2.$  Then  for any two columns $C_s$ and $C_{s+1\mod g_1}$, the number of horizontal edges between $C_s$ and $C_{s+1 \mod g_1}$ in $H$ equals
the number of horizontal edges between $C_{s+1 \mod g_1}$ and $C_{s+2 \mod g_1}$ in $H'$.
\end{prop}

\begin{proof}
Again, any edge from $u_i$ to $u_{i+1}$ has the same length as the edge from $v_i$ to $v_{i+1}$, so that this shift does not change the total number of horizontal/vertical edges.  By construction, if $v_i$ is in column $C_s$, then $u_i=v_i+a_2$ is in column $C_{s+1 \mod g_1}$. 
\end{proof}

\begin{proof}[Proof (of Theorem \ref{thm:reduction})]
First, by Proposition \ref{prop:UB}, there always exists a tour of cost $2(g_1-1)$.  Now suppose we have a tour of cost $k<2(g_1-1)$. We treat this tour as starting at 0 and following in some arbitrary direction, so that it proceeds $v_0=0, v_1, v_2, ..., v_{n-1}, v_n=0$ where $v_i-v_{i-1}\in\{a_1, -a_1, a_2, -a_2\}$.  Let $a_+$ and $a_-$ respectively denote the number of $a_2$ and $-a_2$ edges in this tour. By Proposition \ref{prop:alpha}, we can assume that $a_+ - a_- \in \{0, g_1\}.$  By Proposition \ref{prop:forcedUB}, for our tour to have cost strictly less than $2(g_1-1)$, we must have that $a_+-a_- = g_1.$

Our tour starts in column 0.  Each $+a_2$ edge moves us one column to the right, and each $-a_2$ edge moves us one column to the left (wrapping around from column $0$ or $g_1-1$ if need be).  Since $a_+ - a_- = g_1,$  this tour ``wraps around'' horizontally: our tour starts in the first column, moving through all columns until the last column, and then uses one more horizontal edge to wrap around to the first column.

We note that there is at most one pair of consecutive columns (possibly the last and first column) between which the tour can use zero horizontal edges: if the tour starts in column 0, didn't use any horizontal edges between columns $i, i+1$ and between $j, j+1$ with $i<j$, the tour would never be able to reach vertices in column $i+1$.  Hence, there are at least $g_1-1$ pairs of consecutive columns where the tour uses at least one edge.  In one of these pairs, the tour must use exactly one edge: if not, the tour would cost at least $2(g_1-1)>k.$


By Proposition \ref{prop:a2shift}, we can shift vertex labels by an arbitrary multiple of $a_2$ so that this pair of columns is the first and last: it wraps from column $g_1-1$ to column $0$.  By Proposition \ref{prop:a1shift}, we can further assume that this edge wraps from $-a_2$ (in the last column) to $0$ (in the first column): we can shift all vertex labels by some multiple of $a_1$, until its endpoint in column $g_1-1$ is the vertex $-a_2.$  Note that shifting the labels of all vertices by the same constant does not change the cost of the tour, by circulant symmetry.

We have now argued that if there is a Hamiltonian cycle costing $k<2(g_1-1)$, there must be a Hamiltonian cycle of cost $k$ with exactly one edge between the last and first column, going from vertex $-a_2$ to vertex $0$.  Deleting this edge yields a Hamiltonian path on a cylinder graph (with $c=g_1$ and $r=\f{n}{g_1}$), starting in the first column (at vertex $(0, 0)$) and ending in the last.  Moreover, this Hamiltonian path uses one fewer horizontal edge, and so costs  $k-1.$  To determine the final vertex $-a_2$ of this Hamiltonian path in terms of the cylindrical coordinates, we recall that vertices in the last column and row $x$ have a label of the form $ (g_1-1)a_2 + x a_1$.   Hence, 
$$-a_2 \equiv_n (g_1-1)a_2 + x a_1.$$  
Rearranging, 
$$xa_1 \equiv_n -g_1 a_2.$$  
By Proposition \ref{prop:solv}, there is a unique solution $x$ to this equivalence in $\{0, 1, ..., \f{n}{g_1}-1\}.$  Thus, in cylindrical coordinates, the final vertex is as stated: $(x, c-1).$  

In all, we have shown that any optimal Hamiltonian tour of cost $k<2(g_1-1)$  gives rise to a Hamiltonian path on a cylinder graph with $c=g_1$ and $r=\f{n}{g_1}$ edges using $k-1$ horizontal edges, starting at $(0, 0)$, and ending at $(x, c-1)$.

Completing the proof requires showing that a Hamiltonian path on that cylinder graph, between vertex $0$ and $-a_2$, of cost $k-1 < 2(g_1-1)-1$, gives rise to a Hamiltonian cycle on the circulant instance of cost $k$.  This claim is more direct: our arguments above show that, treating the cylindrical graph as a circulant graph, it has a Hamiltonian path from $0$ to $-a_2$ of cost $k-1.$  Adding an edge of length $a_2$ turns this path into a Hamiltonian tour on a circulant graph of cost $k$, as desired.
\end{proof}

\begin{rem}
We are only considering non-trivial instances of two-stripe TSP.  Hence we assume that $g_1>1$, which implies that $c\geq 2.$  Similarly, because $g_1=\gcd(n, a_1)$ and $a_1 \leq \lfloor\f{n}{2}\rfloor,$ we have that $g_1<n$ and so $r\geq 2.$
\end{rem}

\subsection{GG Paths}\label{s:GG}
Theorem \ref{thm:reduction} reduces two-stripe TSP to the problem of finding a minimum-cost Hamiltonian path between 0 and $-a_2$ in a cylinder graph. Next, we introduce a specific class of Hamiltonian paths: \emph{GG paths}. We will later use Theorem \ref{thm:main} to argue that, if there is a Hamilonian path from $(0, 0)$ to a vertex in the last column, there is also a GG path from $(0, 0)$ to that vertex (and the GG path is of no greater cost).

GG paths from (0,0) to the last column are highly structured paths where ``all the funky business'' happens between the first two columns.  See, e.g., Figure \ref{fig:GG3}: GG paths start with a sequence of vertical edges, use all the extra pairs of horizontal edges to alternate between the first two columns, and then use one horizontal edge between every remaining pair of consecutive columns.  More formally:

\begin{defn}\label{defn:GG}
A {\bf GG path} with $2k+1+(c-2)$ horizontal edges in a cylinder graph is a Hamiltonian path that 
\begin{enumerate}
\item Starts at $(0, 0)$.
\item Then uses $r-2k-1$ consecutive vertical edges, moving along $(0, 0), (1, 0), ..., (r-2k-1, 0)$.
\item Then alternates horizontal and vertical edges, using $2k+1$ total horizontal edges, moving along $(r-2k-1, 0), (r-2k-1, 1), (r-2k, 1), (r-2k, 0), (r-2k+1, 0), ..., (r-1, 0), (r-1, 1)$ and ending at $(r-1, 1)$.
\item Then traverses the remaining vertices in the graph, by traversing a column, taking a horizontal edge to the next column, then traversing that column, and so on until all the vertices have been visited. This part of the path ends in the last column and uses $c-2$ additional horizontal edges (one between every remaining pair of consecutive columns).
\end{enumerate}
or
\begin{enumerate}
\item Starts at $(0, 0)$.
\item Then uses $r-2k-1$ consecutive vertical edges, moving along $(0, 0), (r-1, 0), (r-2, 0), ... (2k+1, 0)$.
\item Then alternates horizontal and vertical edges, using $2k+1$ total horizontal edges, moving along $(2k+1, 0), (2k+1, 1), (2k, 1), (2k, 0), ..., (1, 0), (1, 1)$ and ending at $(1, 1)$.
\item Then traverses the remaining vertices in the graph, by traversing a column, taking a horizontal edge to the next column, then traversing that column, and so on until all the vertices have been visited. This part of the path ends in the last column and uses $c-2$ additional horizontal edges (one between every remaining pair of consecutive columns).
\end{enumerate}  
See Figure \ref{fig:GG3}.
\end{defn}

In a GG path with $2k+1$ horizontal edges between the first and second column, and where $k\geq 1$, there are exactly two options for how to start the tour: the first edge goes from $(0, 0)$ to $(1, 0)$ or wraps vertically from $(0, 0)$ to $(r-1, 0)$.  Once that choice is made, the edges in the first and second column are fully determined: After following the edges  in Definition \ref{defn:GG}, the prescribed part of the path ends at either $(1, 1)$ or $(r-1, 1)$.  Suppose the sequence ends at $(r-1, 1),$ as in the left of Figure \ref{fig:GG3}.  Then, since $k\geq 1$, the vertex $(r-2, 1)$ must have degree 2 and the only possible edge goes from $(r-1, 1)$ to $(0, 1).$  If the path ends at $(1, 1)$, the situation is analogous.  When $k=0,$ there are four options as shown in Figure \ref{fig:GG1}.  Two of these (the first and third in Figure \ref{fig:GG1}) end at the same vertex $(0, 1)$.

\begin{figure}[t!]

    \centering
        \begin{subfigure}[t]{0.2\textwidth}
        \centering
        \begin{tikzpicture}[scale=0.7]
        \foreach \x in {0,1} {
            \foreach \y in {0,1,...,8} {
                \vertex (\x\y) at (\x, \y) {};
            }
        }   
        \foreach \i in {0, 1} {
            \foreach \j in {0, 1, ..., 6} {
            \draw (\i, \j) to (\i, \j+1);}
        }
        
        \draw   (0, -0.5) to (0, 0);
        \draw  [->]  (0, 8) to (0, 8.5);
        \draw (0, 7) to (1, 7);
        
        \draw [->]  (1, 0) to (1, -0.5);
        \draw    (1, 8) to (1, 8.5);
        
        \end{tikzpicture}
        \end{subfigure}
        \hspace{5mm}
        \begin{subfigure}[t]{0.2\textwidth}
        \centering
        \begin{tikzpicture}[scale=0.7]
        \foreach \x in {0,1} {
            \foreach \y in {0,1,...,8} {
                \vertex (\x\y) at (\x, \y) {};
            }
        }   
        \foreach \i in {0, 1} {
            \foreach \j in {0, 1, ..., 5} {
            \draw (\i, \j) to (\i, \j+1);}
        }
        
        \draw [->]  (0, -0.5) to (0, 0);
        \draw  [->]  (0, 8) to (0, 8.5);
        \draw (0, 7) to (1, 7);
        \draw (0, 6) to (0, 7);
        
        \draw [->]  (1, -0.5) to (1, 0);
        \draw  [->]  (1, 8) to (1, 8.5);
        \draw (1, 7) to (1, 8);
        
        \end{tikzpicture}
        \end{subfigure}
\hspace{5mm}
        \begin{subfigure}[t]{0.2\textwidth}
        \centering
        \begin{tikzpicture}[scale=0.7]
        \foreach \x in {0,1} {
            \foreach \y in {0,1,...,8} {
                \vertex (\x\y) at (\x, \y) {};
            }
        }   
        \foreach \i in {0, 1} {
            \foreach \j in {0, 1, ..., 7} {
            \draw (\i, \j) to (\i, \j+1);}
        }
        
        \draw (0, 0) to (1, 0);

        \end{tikzpicture}
        \end{subfigure}
        \hspace{5mm}
        \begin{subfigure}[t]{0.2\textwidth}
        \centering
        \begin{tikzpicture}[scale=0.7]
        \foreach \x in {0,1} {
            \foreach \y in {0,1,...,8} {
                \vertex (\x\y) at (\x, \y) {};
            }
        }   
        \foreach \i in {0, 1} {
            \foreach \j in {1, 2, ..., 7} {
            \draw (\i, \j) to (\i, \j+1);}
        }
        \draw (0, 0) to (0, 1);
        \draw (0, 0) to (1, 0);
        \draw [->] (1, 0) to (1, -0.5);
        \draw (1, 8.5) to (1, 8);

        \end{tikzpicture}
        \end{subfigure}

    \caption{GG paths using 1 horizontal edge between the first and second column.}
    \label{fig:GG1}
\end{figure}
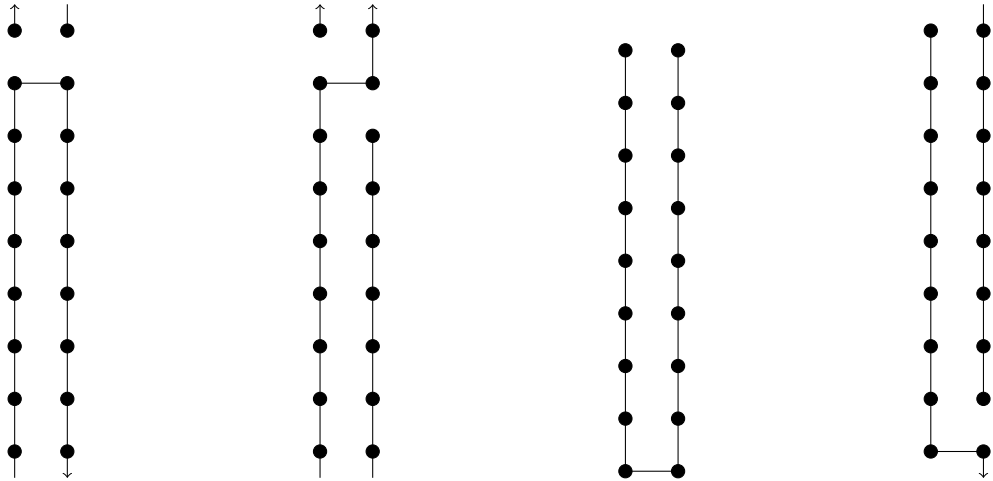

Because GG paths only use one horizontal edge between consecutive columns starting from the second column onward, they are extremely structured.  Moreover, when moving to a new column (starting at the third column), a GG path has exactly two options:  the first edge in that column vertically is either vertically ``up'' (going from $(i, j)$ to $(i-1, j)$) or vertically ``down'' (going from $(i, j)$ to $(i+1, j)$).  The net effect is that, if the last vertex visited in the second column is $(i, j-1)$, the last vertex visited in the third column will be $(i+1, j)$ or $(i-1, j)$.  See Figure \ref{fig:GG3cont}.

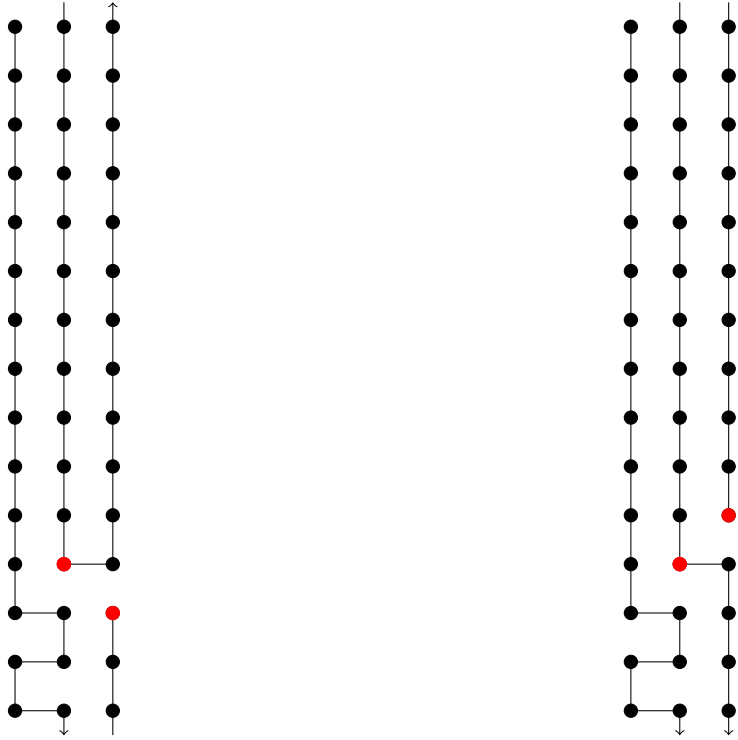
\begin{figure}[t!]

    \centering
        \begin{subfigure}[t]{0.45\textwidth}
        \centering
        \begin{tikzpicture}[scale=0.65]
        \foreach \x in {0,1, 2} {
            \foreach \y in {0,1,...,14} {
                \vertex (\x\y) at (\x, \y) {};
            }
        }   
        \foreach \i in {0, 1} {
            \foreach \j in {3, 4, ..., 13} {
            \draw (\i, \j) to (\i, \j+1);}
        }
        
        \foreach \j in {0, 1, 2} {
        \draw (0, \j) to (1, \j);}
        
        \draw (0, 2) to (0, 3);
        \draw (0, 0) to (0, 1);
        \draw (1, 1) to (1, 2);
        \draw [->]  (1, 0) to (1, -0.5);
        \draw    (1, 14) to (1, 14.5);
        
        \vertex[red] (r1) at (1, 3) {};
        
        \foreach \j in {3, 4, ..., 13} {
        \draw (2, \j) to (2, \j+1);}

        \foreach \j in {0, 1} {
        \draw (2, \j) to (2, \j+1);}
        
        \draw (1, 3) to (2, 3);
        \draw[->] (2, 14) to (2, 14.5);
        \draw (2, -0.5) to (2, 0);
        
        \vertex[red] (r1) at (1, 3) {};
        \vertex[red] (r2) at (2, 2) {};
        
        \end{tikzpicture}
        \end{subfigure}
        \hspace{5mm}
        \begin{subfigure}[t]{0.45\textwidth}
        \centering
        \begin{tikzpicture}[scale=0.65]
        \foreach \x in {0,1, 2} {
            \foreach \y in {0,1,...,14} {
                \vertex (\x\y) at (\x, \y) {};
            }
        }   
        \foreach \i in {0, 1} {
            \foreach \j in {3, 4, ..., 13} {
            \draw (\i, \j) to (\i, \j+1);}
        }
        
        \foreach \j in {0, 1, 2} {
        \draw (0, \j) to (1, \j);}
        
        \draw (0, 2) to (0, 3);
        \draw (0, 0) to (0, 1);
        \draw (1, 1) to (1, 2);
        \draw [->]  (1, 0) to (1, -0.5);
        \draw    (1, 14) to (1, 14.5);
        
        \vertex[red] (r1) at (1, 3) {};
        
        \foreach \j in {4, 5, ..., 13} {
        \draw (2, \j) to (2, \j+1);}

        \foreach \j in {0, 1, 2} {
        \draw (2, \j) to (2, \j+1);}
        
        \draw (1, 3) to (2, 3);
        \draw (2, 14) to (2, 14.5);
        \draw[->] (2, 0) to (2, -0.5);
        
        \vertex[red] (r1) at (1, 3) {};
        \vertex[red] (r2) at (2, 4) {};
        
        \end{tikzpicture}
        \end{subfigure}
    \caption{Extending a GG path on 2 columns to a GG path on 3 columns.  Notice that, when entering the third column, there are exactly two choices: the first vertical edge in the third column either moves up or down.  From there, the third column is completely determined.  The last vertex visited in the second and third columns are shown in red.  Note that the red vertex in the third column is exactly one row above or below the the red vertex in the second column. }
    \label{fig:GG3cont}
\end{figure}

This structure of GG paths applies inductively. Starting in the third column, each time a GG path enters a new column, it has the exact same two options.  Tracing out this process allows us to quickly determine where a GG path can end, based on its path through the first two columns.  See Figure \ref{fig:GG3full}.

\begin{figure}[t!]

    \centering
        \begin{subfigure}[t]{0.45\textwidth}
        \centering
        \begin{tikzpicture}[scale=0.65]
        \foreach \x in {0,1, ..., 5} {
            \foreach \y in {0,1,...,14} {
                \vertex (\x\y) at (\x, \y) {};
            }
        }   
        \foreach \i in {0, 1} {
            \foreach \j in {3, 4, ..., 13} {
            \draw (\i, \j) to (\i, \j+1);}
        }
        
        \foreach \j in {0, 1, 2} {
        \draw (0, \j) to (1, \j);}
        
        \draw (0, 2) to (0, 3);
        \draw (0, 0) to (0, 1);
        \draw (1, 1) to (1, 2);
        \draw [->]  (1, 0) to (1, -0.5);
        \draw    (1, 14) to (1, 14.5);
        \draw (1, 3) to (2, 3);
        \vertex[red] (r2) at (1, 3) {};
        
        \vertex[red] (r2) at (2, 2) {};
        \vertex[red] (r2) at (2, 4) {};
        
        \vertex[red] (r2) at (3, 3) {};
        \vertex[red] (r2) at (3, 1) {};
        \vertex[red] (r2) at (3, 5) {};
        
        \vertex[red] (r2) at (4, 0) {};
        \vertex[red] (r2) at (4, 2) {};
        \vertex[red] (r2) at (4, 4) {};
        \vertex[red] (r2) at (4, 6) {};
        
        \vertex[red] (r2) at (5, 14) {};
        \vertex[red] (r2) at (5, 1) {};
        \vertex[red] (r2) at (5, 3) {};
        \vertex[red] (r2) at (5, 5) {};
        \vertex[red] (r2) at (5, 7) {};
        
        \end{tikzpicture}
        \end{subfigure}
        \hspace{5mm}
        \centering
        \begin{subfigure}[t]{0.45\textwidth}
        \centering
        \begin{tikzpicture}[scale=0.65]
        \foreach \x in {0,1, ..., 5} {
            \foreach \y in {0,1,...,14} {
                \vertex (\x\y) at (\x, \y) {};
            }
        }   
        \foreach \i in {0, 1} {
            \foreach \j in {0, 1, ..., 9} {
            \draw (\i, \j) to (\i, \j+1);}
        }
        
        \foreach \j in {11, 12, 13} {
        \draw (0, \j) to (1, \j);}
        
        \draw (0, 10) to (0, 11);
        \draw [->] (0, 14) to (0, 14.5);
        \draw (0, 0) to (0, -0.5);
        \draw (0, 12) to (0, 13);
        \draw [->] (1, 14) to (1, 14.5);
        \draw (1, 0) to (1, -0.5);       
        \draw (1, 11) to (1, 12);
        \draw (1, 13) to (1, 14);
            
        \draw (1, 10) to (2, 10);       
        
        \vertex[red] (r2) at (1, 10) {};
        
        \vertex[red] (r2) at (2, 11) {};
        \vertex[red] (r2) at (2, 9) {};
        
        \vertex[red] (r2) at (3, 8) {};
        \vertex[red] (r2) at (3, 10) {};
        \vertex[red] (r2) at (3, 12) {};
        
        \vertex[red] (r2) at (4, 7) {};
        \vertex[red] (r2) at (4, 9) {};
        \vertex[red] (r2) at (4, 11) {};
        \vertex[red] (r2) at (4, 13) {};
        
        \vertex[red] (r2) at (5, 6) {};
        \vertex[red] (r2) at (5, 8) {};
        \vertex[red] (r2) at (5, 10) {};
        \vertex[red] (r2) at (5, 12) {};
        \vertex[red] (r2) at (5, 14) {};

        \end{tikzpicture}
        \end{subfigure}
        
    \caption{Possible ending vertices for a GG path with 3 edges between the first two columns.}
    \label{fig:GG3full}
\end{figure}
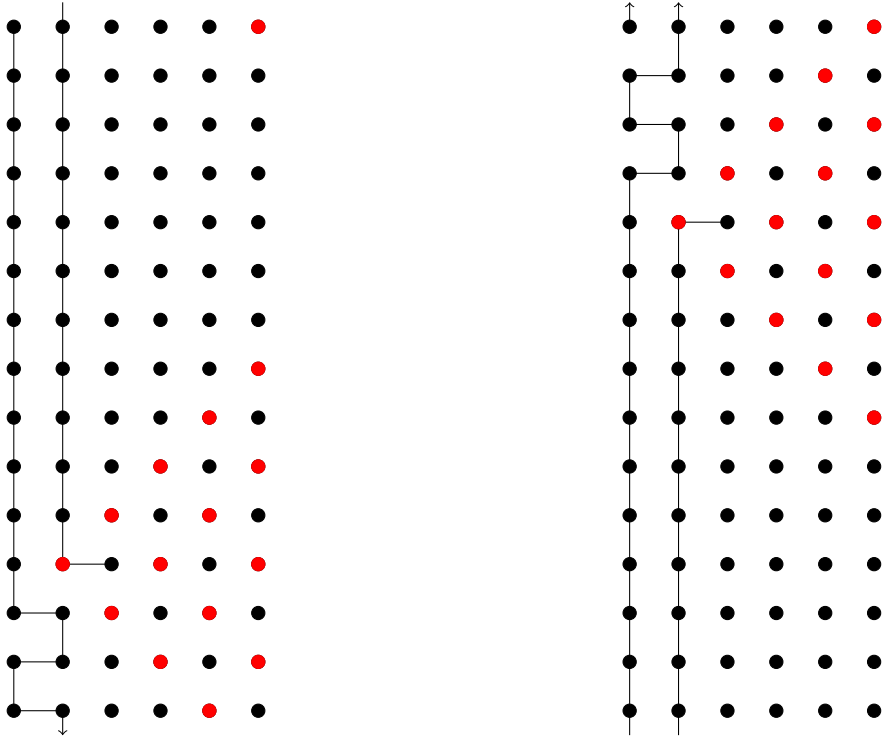

We will argue (in Theorem \ref{thm:main}) that, in solving the minimum-cost Hamiltonian path problem on a cylinder graph between $(0, 0)$ and a vertex in the last column, it suffices to consider GG paths: if a minimum-cost path between $(0, 0)$ and $(i, c-1)$ requires $k$ horizontal edges for any $0\leq i\leq r-1$, there is a GG path from $(0, 0)$ to $(i, c-1)$ using at most $k$ horizontal edges.  Hence, we characterize exactly where a GG path, with some fixed number of horizontal edges, can end up.  We define this set below.

\begin{defn}
Let $A_{r, c, m}$ denote the set of row indices of vertices in the last column reachable on a GG path in an $r\times c$ cylinder graph using {\bf at most} $2m+(c-1)$ total horizontal edges.
\end{defn}
In the definition of $A_{r, c, m},$ it is useful to think of the $m$ as the number of extra ``expensive pairs'' of horizontal edges used between the first and second column.  
Note also that a GG path can have at most $\lfloor\f{r-1}{2}\rfloor$ extra horizontal pairs, because the extra horizontal edges are all between the first two columns. The following proposition characterizes $A_{r,c,m}$. 

\begin{prop}
\label{prop:Arcm} Let $0 \leq m\leq \lfloor \f{r-1}{2}\rfloor.$  Then
$$A_{r, c, m} = \{c+2m-2i \mod r: 0\leq i\leq c+2m\}.$$
\end{prop}

Note that, by our definition of $A_{r, c, m}$, this means that the set of vertices reachable in a GG path on an $r\times c$ cylinder graph using at most $2m+(c-1)$ total horizontal edges is:
$$\{(c+2m-2i \; \text{({mod} $r$)}, \,\, c-1): 0\leq i\leq c+2m\}.$$

The proof of Proposition \ref{prop:Arcm} follows the structure of GG paths discussed above: we need only identify the vertices reachable when $c=2.$  Then we follow the ``triangular growth'' of the red vertices in  Figure \ref{fig:GG3full}, where a GG path ending in some row $k$ of column $c$ corresponds to a GG path ending in column $c+1$ in row $k-1$ or $k+1$ (taken mod $r$). 

Here and throughout, we use the notation  $\{a_1,a_2,\ldots,a_k\}$ mod $r$ as a shorthand for $\{${$a_1$ mod $r$, $a_2$ mod $r$, \ldots, $a_k$ mod $r$}$\}$.

\begin{proof}
Our proof proceeds using two rounds of induction.  By definition, $A_{r, c, m} \subset A_{r, c, m+1}$.  We first induct on $m$ to characterize $A_{r, 2, m}.$  Then we induct on $c$ to be fully general.

Hence, consider $A_{r, 2, 0}.$  These are exactly the paths in Figure \ref{fig:GG1}. We can thus see that $$A_{r, 2, 0} =\{0, 2, r-2\}.$$  This matches the proposition statement.  

Now we induct on $m$.  For $m\geq 1,$  by definition $A_{r, 2, m}$ consists of the vertices in $A_{r, 2, m-1}$, together with the vertices that can be reached using a GG path with exactly $2m+1$ horizontal edges.  Here it is useful to view Figure  \ref{fig:GG3}.  There are exactly two GG paths with $2m$ extra  horizontal edges that we need to consider, corresponding to the first edge going ``down'' from $(0, 0)$ to $(1, 0)$, or to the first edge going ``up'' and wrapping vertically from $(0, 0)$ to $(r-1, 0).$  In the former, we use all $2m+1$ horizontal edges at rows $r-1, r-2, ..., r-(2m+1)$ and thus end at  $(r-2m-2, 1)$.  In the latter, we use horizontal edges at rows $1, 2, ..., 2m+1$ and end at $(2m+2, 1).$  Hence
\begin{align*}
A_{r, 2, m} &= A_{r, 2, m-1} \cup \{ r-2m-2 \} \cup \{ 2m+2\} \\
&= \{ 2+2(m-1)-2i \mod r : 0\leq i \leq 2+2(m-1)\}\cup  \{ r-2m-2\} \cup \{ 2m+2 \}\\
&= \{2m, 2m-2, \ldots, -2m+2, -2m\} \cup \{-(2m+2)\} \cup \{2m+2\} \mod r \\
&= \{ 2+2m-2i \mod r : 0\leq i \leq 2+2m\}.
\end{align*}

Now that we have established the proposition for $A_{r, 2, m}$, we induct on $c$ to prove the proposition for $A_{r,c,m}$. For $c\geq 2,$ we have that
$$A_{r, c+1, m} = \{x\pm 1: x \in A_{r, c, m}\}.$$  This is the ``triangular growth'' of the red vertices in  Figure \ref{fig:GG3full}, i.e.,  because we can only use one horizontal edge between columns $c-1$ and $c$.  Thus by induction, since $A_{r, c, m} = \{c+2m-2i \mod r: 0 \leq i \leq c+2m\}$, we have that
$$A_{r, c+1, m} = \{ c+2m-2i \pm 1 \mod r: 0 \leq i \leq c+2m\}.$$
We now trace through the values of $c+2m-2i \pm 1$ as $i$ ranges from $0$ to $c+2m$.  Since $c+2m-2i -1 = c + 2m -2(i+1) + 1$, we see that the values trace through $$c+2m -2(0)+1, c+2m-2(0)-1, c+2m-2(1)-1, c+2m-2(2)-1,...., c+2m-2(c+2m)-1.$$  That is, $$(c+1)+2m-2(0), (c+1)+2m-2, (c+1)+2m-4, ..., (c+1)+2m-2(c+1+2m).$$  That is, $\{(c+1)+2m-2i: 0 \leq i \leq c+1+2m\}.$  This inductively establishes that $A_{r, c+1, m}$ is correctly stated.
\end{proof}
The above characterization of $A_{r,c,m}$ immediately implies the following corollary, which essentially states that for each additional pair of horizontal edges used, the number of vertices reachable in the last column increases by at most 2. Furthermore, we can characterize exactly which new vertices (if any), become reachable. 
\begin{cor}\label{cor:newrows}
We have $A_{r,c,m} \setminus A_{r,c,m-1} \subseteq \{c+2m, -(c+2m) \} \mod r$. 
\end{cor}
Finally, we show that GG paths ``capture" all Hamiltonian paths that only use extra pairs of horizontal edges in the first two columns, in the sense made precise below. 

\begin{restatable}{prop}{ggall}
\label{prop:ggall}
Let $P$ be a Hamiltonian path in an $r$ by $c$ cylinder graph, starting at $(0, 0)$ and ending at $(x, c-1)$ for some row $x$.  Suppose that $P$ uses $2m+1$ horizontal edges between the first pair of columns, and then 1 horizontal edge between every subsequent pair of columns.  Then $x \in A_{r, c, m}$.
\end{restatable}
\begin{proof}
We have moved the proof of this proposition to Appendix \ref{app:ggall} because it is long, yet relatively unenlightening. Intuitively, it considers all possible configurations of $P$ in the first 2 columns, and shows that in each case, the set of reachable vertices in the last column is contained in $A_{r,c,m}$. 
\end{proof}

\subsection{Roadmap for Rest of the Paper}
Before moving on, we step back to briefly summarize what we have shown so far, and outline the rest of the paper. 
So far, we have argued that for two-stripe TSP, the minimum-cost tour is either 1) the upper-bound tour, or 2) a minimum-cost Hamiltonian path from 0 to $-a_2$, plus the wraparound edge from $-a_2$ to 0. In Theorem \ref{thm:main}, we will show that the Hamiltonian path in case 2) can always be taken to be a GG path. Since GG paths are extremely structured, this leads us to the characterization of the minimum-cost tour in Theorem \ref{thm:result}. 

In Section \ref{sec:main_result_proof}, we prove Theorem \ref{thm:result} assuming Theorem \ref{thm:main} holds. Then in Section \ref{sec:alg}, we show how the characterization in Theorem \ref{thm:result} leads to a polynomial-time algorithm for determining the tour. Finally Section \ref{sec:mainpf} contains the proof of Theorem \ref{thm:main}. 

\section{Proof of Main Result}
\label{sec:main_result_proof}
In this section, we prove our main result, Theorem \ref{thm:result}, which characterizes the cost of an optimal tour for the two-stripe TSP. As a reminder, the notation $\{a_1,a_2,\ldots,a_k\}$ mod $r$ is shorthand for $\{${$a_1$ mod $r$, $a_2$ mod $r$, \ldots, $a_k$ mod $r$}$\}$.

\result*
Our proof of Theorem \ref{thm:result} relies on Theorem \ref{thm:main}. This is one of the main technical results of the paper, and we devote Section \ref{sec:mainpf} to its proof. We state it without proof below. 

\begin{restatable}{thm}{main}
\label{thm:main}
Consider a cylinder graph on $n=r\times c$ vertices, with $r$ rows and $c$ columns.  Suppose we have a Hamiltonian path, starting at 0 and ending in the last column, and suppose it uses at most $(c-1)+2m$ horizontal edges.  Then it must end at a row in $A_{r, c, m}$.
\end{restatable}
\begin{proof}[Proof of Theorem \ref{thm:result}]
Suppose $m^* \leq 0$. This implies $x \in \{-c, -c+2, \ldots, c-2, c\}$ mod $r$. Note that this set is equal to $A_{r,c,0}$ by Proposition \ref{prop:Arcm}. By definition of $A_{r,c,0}$, there exists a GG path, using exactly $c-1$ horizontal edges, from 0 to $-a_2$. Appending the edge from $-a_2$ to 0 gives a Hamiltonian cycle with cost $c$. This cycle is optimal, since $c$ is a lower bound on the cost of any tour. 

Next, suppose $0 < 2m^* < c-2$. By Proposition \ref{prop:Arcm}, we have that $x \in A_{r,c,m^*}$. Moreover, the minimality of $m^*$ implies that $m^*$ is the \emph{smallest} value of $m$ for which $x \in A_{r,c,m}$. By definition of $A_{r,c,m^*}$, there is a GG path from 0 to $-a_2$ with cost $(c-1) + 2m^*$, and appending the wraparound edge from $-a_2$ to 0 gives a tour of cost $c + 2m^*$. To show that this is optimal, note first that since $2m^* < c-2$, we have $c+2m^* < 2(c-1)$. Therefore by Theorem \ref{thm:reduction}, the optimal tour consists of a cheapest Hamiltonian path $P$ from 0 to $-a_2$, plus the wraparound edge from $-a_2$ to 0. Now, Theorem \ref{thm:main} says that we can always take $P$ to be a GG path. Since $m^*$ is the smallest value of $m$ for which $x \in A_{r,c,m}$, it follows that the cost of $P$ is $(c-1) + 2m^*$, so we are done. 

Finally, suppose that $2m^* \geq c-2$ or $m^*$ does not exist. By Proposition \ref{prop:Arcm}, this implies that $x \not\in A_{r,c,m}$ for any $m$ with $0 \leq 2m < c-2$. This, together with Theorems \ref{thm:reduction} and \ref{thm:main}, implies that the upper bound tour is optimal, the cost of which is $2c-2$. 
\end{proof}

Figure \ref{fig:lbggub} shows the three types of tours that correspond to the three cases in Theorem \ref{thm:result}. 

    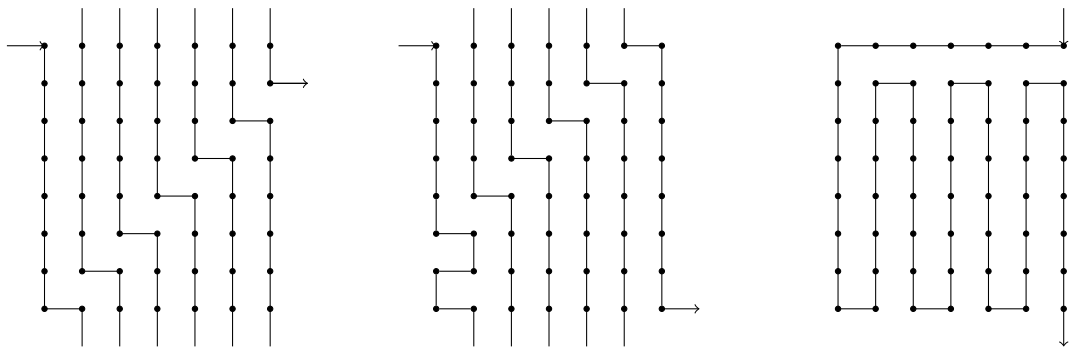
\begin{figure}
    \centering
    \begin{subfigure}[t]{0.3\textwidth}
        \centering
        \begin{tikzpicture}[scale=0.5]
        \foreach \x in {0,1,...,6} {
            \foreach \y in {0,...,7} {
                \smvertex (\x\y) at (\x, \y) {};
            }
        }
        
        \foreach \x in {0,1,...,6} {
            \foreach \y in {\x,...,6} {
                \draw (\x, \y) -- (\x, \y+1);
            }
            \draw (\x, \x) -- (\x+1, \x);
        }
        
        \foreach \x in {1,2,...,6} {
            \draw (\x, 8) -- (\x, 7);
            \foreach \y in {1,...,\x} {
                \draw (\x, \y-2) -- (\x, \y-1);
            }
        }
        \draw[->] (-1,7) -- (0, 7);
        \draw[->] (6,6) -- (7,6);
        
        \end{tikzpicture}
        \caption{A lower bound tour, which has cost $c$.}
        \label{subfig:lb}
    \end{subfigure}%
    ~ 
    \begin{subfigure}[t]{0.3\textwidth}
        \centering
        \begin{tikzpicture}[scale=0.5]
        \foreach \x in {0,1,...,6} {
            \foreach \y in {0,...,7} {
                \smvertex (\x\y) at (\x, \y) {};
            }
        }
        
        \draw (07) -- (06)--(05)--(04)--(03)--(02);
        \draw (17) -- (16)--(15)--(14)--(13);
        \draw (27) --(26)--(25)--(24);
        \draw (37) --(36)--(35);
        \draw (47) --(46);
        
        \draw (12)--(11)--(01)--(00)--(10)--(1,-1);
        
        \draw (60) -- (67);
        
        \foreach \x in {0,1,...,5} {
            \draw (\x, \x+2) -- (\x+1, \x+2);
        }
        
        \foreach \x in {1,2,...,5} {
            \draw (\x,8) -- (\x,7);
        }
        
        \foreach \x in {2,3,...,5} {
            \foreach \y in {-1,...,\x} {
                \draw (\x, \y) -- (\x, \y+1);
            }
        }
        \draw[->] (-1,7) -- (0, 7);
        \draw[->] (6,0) -- (7,0);
        \end{tikzpicture}
        \caption{A tour that consists of a GG path using $m$ additional horizontal pairs, plus the wraparound edge to 0. Here, $m = 1$.}
        \label{subfig:gg}
    \end{subfigure}
    ~ 
    \begin{subfigure}[t]{0.3\textwidth}
        \centering
        \begin{tikzpicture}[scale=0.5]
        \foreach \x in {0,1,...,6} {
            \foreach \y in {0,...,7} {
                \smvertex (\x\y) at (\x, \y) {};
            }
        }
        
        \foreach \x in {0,1,...,6} {
            \draw (\x,6) -- (\x, 0);
        }
        
        \foreach \x in {1,3,5} {
            \draw (\x,6) -- (\x+1, 6);
            \draw (\x,0) -- (\x-1,0);
        }
        
        \draw (06) -- (07) -- (67);

        \draw[->] (6,8) -- (6, 7);
        \draw[->] (6,0) -- (6,-1);

        \end{tikzpicture}
        \caption{An upper bound tour, which has cost $2c-2$.}
        \label{subfig:ub}
    \end{subfigure}
    \caption{An illustration of the three types of tours in Theorem \ref{thm:result}.}
    \label{fig:lbggub}
\end{figure}

\section{The Algorithm}\label{sec:alg}
We turn now to showing how Theorem \ref{thm:result} naturally leads to an algorithm for two-stripe circulant TSP. In this section, we will describe the algorithm formally and prove that it runs in polynomial time. As mentioned in the introduction, it is possible to represent the input to two-stripe circulant TSP with 3 numbers: 
\begin{itemize}
	\item $n$, the number of nodes,
	\item $a_1, a_2 \leq \frac{n}{2}$, the lengths of the edges with finite cost (i.e. the cost of edge $\{i, j\}$ is finite iff $\min\{n-\abs{i-j}, \abs{i-j}\} \in \{a_1, a_2\}$).
\end{itemize}
(Recall that we assume without loss of generality that $c_{a_1} = 0$ and $c_{a_1} = 1$.)

It is important to discuss what we mean by ``polynomial time." We will take ``polynomial time" to mean polynomial time in the bit size of the input, which in our case is \emph{polylogarithmic} in $n$. Note then that we cannot, strictly speaking, output the entire tour as a sequence of vertices, because this would require listing $n$ numbers, taking $\Omega(n)$ time. Instead, we will show a guarantee on our algorithm that is similar in spirit to the statement of Theorem \ref{thm:result}: There are two parametrized classes of tours to which the optimal tour can belong (an upper bound tour or a GG path plus an edge), and we can determine the class as well as output a set of parameters that describe the tour in polynomial time. (In particular, we are able to compute the \emph{cost} of the optimal tour in polynomial time.) Given the class and the corresponding parameters, it is easy to output the sequence of vertices in the tour in $O(n)$ time. 

Below is the algorithm. It mimics the statement of Theorem \ref{thm:result}. Throughout the remainder of this section, we will let $r := \frac{n}{g_1}$ and $c := g_1$. 
\begin{enumerate}
\item Calculate the row index of $-a_2$. That is, calculate the value of $x \in \{0,1,\ldots, r-1\}$ such that the cylindrical coordinates of $-a_2$ are $(x, c-1)$. (We will explain how to do this step in $O(\log^2n)$ time.)
\item Compute $m^*$, the smallest integer value of $m$ such that $m \geq -\frac{c}{2}$ and $x \in \{c+2m, -(c+2m)\} \mod r$. (We will explain how to do this step in $O(\log ^2n)$ time.)
\item If $2m^* < c-2$, then the cost of the optimal tour is $c + \max\{0, 2m^*\}$. The optimal tour is achieved by taking a cheapest GG path from 0 to $-a_2$ (which will use exactly $2\max\{m^*, 0\} + (c-1)$ horizontal edges), and appending the wraparound edge from $-a_2$ to 0. 
\item Otherwise, if $2m^* \geq c-2$ or $m^*$ does not exist, the cost of the optimal tour is $2c-2$. In this case, the upper bound tour (see Proposition \ref{prop:UB} and Figure \ref{fig:UB}), is optimal. 
\end{enumerate}



\begin{thm}[Correctness of the Algorithm]
\label{thm:alg_correctness}
The algorithm correctly determines a minimum-cost tour. 
\end{thm}
\begin{proof}
This directly follows from Theorem \ref{thm:result} and its proof.
\end{proof}

\begin{thm}[Runtime of the Algorithm]
\label{thm:alg_runtime}
The algorithm calculates the cost of the optimal tour in $O(\log^2n)$ time. Moreover, it can output a set of parameters that describe an optimal tour in an additional $O(\log^2 n)$ time, and it can return the sequence of $n$ vertices in the tour in $O(n)$ time.
\end{thm}
\begin{proof}
To implement the algorithm, we need to first describe how to find the row index $x$ of $-a_2$.  By Theorem \ref{thm:reduction}, we know that $0 \leq x < \frac{n}{g_1}$ is the unique solution to $a_1x \equiv - g_1a_2 \mod n$. Let us see how to compute $x$ in $O(\log^2 n)$ time. Since $g_1 = \gcd(a_1, n)$, we can divide both sides of the above relation by $g_1$ to obtain the equivalent relation
$$\left(\frac{a_1}{g_1}\right)x \equiv -a_2 \mod \frac{n}{g_1}.$$
Since $\gcd(\frac{a_1}{g_1}, \frac{n}{g_1}) = 1$, this congruence has a unique solution $x \in \{0, 1, \ldots, \frac{n}{g_1}\}$, and it can be found using the extended Euclidean algorithm in $O(\log^2 n)$ time. (See, for example, Theorem 44 in \cite{shoup09}.)


In the second step of the algorithm, we need to compute $m^*$, the smallest integer value of $m$ such that $m \geq -\frac{c}{2}$ and $x \in \{c+2m, -(c+2m)\} \mod r$.  As a technical point, we don't actually need to precisely compute $m^*$ if it is nonpositive, since in that case the optimal tour will have cost $c + \max\{0, 2m^*\}=c+0=c$.  Similarly, $m^*$ may not exist. Thus, we will determine if $m^*$ is nonpositive, positive, or does not exist, and we only need to find the explicit value of $m^*$ if it is positive and exists.

To find the value of $m^*$ (or determine that it is nonpositive or does not exist), we solve two congruences for $m$, which we can do with the extended Euclidean algorithm in $O(\log^2 n)$ time.  First, we find the smallest value of $m$ such that $m \geq -\frac{c}{2}$ and $x \equiv_r c+2m$.  Second, we find the smallest value of $m$ such that $m \geq -\frac{c}{2}$ and $x \equiv_r -(c+2m)$. Recall that we computed $x$ in the first step of the algorithm, so $m$ is the only unknown in each congruence.

Lets consider the first congruence, that $x \equiv_r c+2m$. Equivalently, we want to find $m$ where $2m \equiv_r (x-c)$. Let $$r' = \frac{r}{\gcd(r, 2)} = \begin{cases} r, & r \text{ odd} \\ r/2, & r \text{ even}. \end{cases}$$  By Proposition \ref{prop:solv}, the congruence $2m \equiv_r (x-c)$ either has no solution for $m$ or there is a unique solution $m_0$ where $0\leq m_0 < r'$ and all other solutions are of the form $m_0 + \lambda r'$ for $\lambda \in \Z.$  Thus we can use the extended Euclidean algorithm in $O(\log^2 n)$ time to either establish that this congruence is infeasible, or find $m_0$.  Once we find $m_0$, then 
we will have one candidate for $m^*$: (1) $m^*=m_0$ if $m_0 - r' < -\frac{c}{2}$, or (2) $m^*=m_0 - r'$ if $m_0-r' \geq -\frac{c}{2}$.  In the first case where $m_0-r'<-\frac{c}{2}$, we know that any candidate  solution of the form $m_0 + \lambda r'$ for $\lambda \leq -1$ will not satisfy the constraint that it be at least $-\frac{c}{2}$.  In the second case, where $m_0-r' \geq -\frac{c}{2}$, we will have that $m^*=m_0 - r'< 0$ (since $m_0 < r'$), and thus will have established that $m^*$ is negative. 


The situation for the second congruence is analogous: applying the extended Euclidean algorithm again either certifies infeasiblility or gives us another  candidate value for $m^*$. To find the true value of $m^*$, it suffices to take the minimum of the two candidate values. If both congruences are infeasible, then $m^*$ does not exist, and we are in the case where the upper bound tour is optimal.

So far, we have shown how to compute $x$ and $m^*$ in $O(\log^2n)$ time. Note that by Theorem \ref{thm:result}, this is already enough to calculate the cost of the optimal tour. Now, let us see how to output a set of parameters that uniquely describe the tour in $O(\log^2n)$ time, and how to output the sequence of $n$ vertices that comprise the tour in $O(n)$ time.
\begin{enumerate}
	\item First, suppose $2m^* \geq c-2$ or $m^*$ does not exist. Then, as shown in the proof of Theorem \ref{thm:result}, the optimal tour is the upper bound tour, which uses $2g_1 - 1$ horizontal edges. This tour is described in Proposition \ref{prop:UB}, and does not need any additional parameters to describe. The exact sequence of vertices in the upper bound tour can be listed in $O(n)$ time.
	\item On the other hand, suppose $2m^* < c-2$. Then, the optimal tour consists of a cheapest GG path $P$ from 0 to $-a_2$ (which will use exactly $2m^* + c - 1$ horizontal edges), plus the wraparound edge from $-a_2$ to 0. It remains for us to find parameters that describe this GG path, and show that the path can be listed in $O(n)$ time. 
	
	First, suppose $m^* \geq 1$. In Section \ref{sec:reduction}, we essentially showed that any GG path can be described with 2 parameters: 
	\begin{enumerate}
	    \item The way it moves in the first column (either starting by going vertically up or vertically down). Note that this also uniquely defines how the path moves in the second column. 
	    \item In each of the remaining $c-2$ columns, whether the first vertical edge of $P$ in the column is going up or down. Note that when considering the endpoint of $P$, the only thing that matters here is the \emph{number} of columns (from the third column onward) for which the first vertical edge of $P$ in the column is going down. We will let $k$ denote the number of columns from the third column onward for which the first vertical edge of $P$ in the column is going up. 
	\end{enumerate} 
	If $P$ starts in the first column by going vertically up, then (as in the proof of Proposition \ref{prop:Arcm}), the last vertex visited by $P$ in the second column is $(2m^*+2, 1)$. Then, the row index of the last vertex of $P$ (in the last column), is $2m^*+2 + k - (c-2-k) \mod r$. To see where this comes from, let $x_i$ denote the row index of the last vertex that is visited in column $i$, where $i \geq 2$. Then, any column $i \geq 3$ in which the first vertical edge goes down, we have $x_i = x_{i-1} - 1$, and if the first vertical edge goes up, we have $x_i = x_{i-1} + 1$. Therefore, if exactly $k$ of the columns from the third column onward have their first vertical edge going up (and hence $(c-2-k)$ of them have their first vertical edge going down), we know that the row index of the last vertex visited in the last column is $2m^*+2 + k - (c-2-k)$.
	To solve for $k$, we {find the smallest non-negative integer} $k$ that makes the above expression congruent to $x$ (mod $r$):
	$$ 2k \equiv_r x +c-4-2m^*.$$
	 We can solve the congruence (or certify that no solution exists), using the extended Euclidean algorithm in $O(\log^2n)$ time. If a solution $k$ is found {that satisfies $k \leq c-2$}, then we are done: We know that an optimal GG path starts in the first column by moving up, and starts by going down in $k$ of the columns from the third column onward. This is enough information to output the sequence of $n$ vertices of the tour in $O(n)$ time.
	
	On the other hand, if the linear congruence is infeasible {or $k > c-2$}, it must mean that $P$ starts in the first column by going vertically down. In this case (as in the proof of Proposition \ref{prop:Arcm}), the last vertex visited by $P$ in the second column is $(r-2m^*-2, 1)$. Then, the row index of the last vertex of $P$ is $(r-2m^*-2) + k - (c-2-k) \mod r$. As before, we can solve for $k$ by solving the linear congruence
	$$2k \equiv x + 2m^* + c\mod r,$$
	for the smallest non-negative integer solution $k$, which takes $O(\log^2n)$ time using the extended Euclidean algorithm.  With the value of $k$ known, we know that such GG path starts in the first column by moving down, and in the columns from the third onward, starts in $k$ of them by moving down. Again, this allows us to generate the sequence of vertices in the tour in $O(n)$ time. 
	
	Note that (at least) one of the two congruences above must always have a solution, since we know that a GG path of cost $2m^* + (c-1)$ exists.
	
	Finally, we still need to consider the case $m^* \leq 0$. In this case, the optimal GG path uses 1 edge between every pair of consecutive columns. This case is very similar to the case $m^* \geq 1$, so we will omit the details here. The only difference is that when $m^* \leq 0$, there are 3 options for where the GG path can end in the second column (as opposed to 2 options when $m^* \geq 1$). These options are described after Definition \ref{defn:GG} in Section \ref{sec:reduction}. 
\end{enumerate}
\end{proof}

\section{Proof of Theorem \ref{thm:main}}
\label{sec:mainpf}

This section contains the final main proof of the paper, which is the proof of Theorem \ref{thm:main}. Recall the statement of the theorem:

\main*

Our proof of Theorem \ref{thm:main} involves a minimal counterexample-style argument. We begin by sketching the main ideas below; Sections \ref{subsec:cycleprop} to  \ref{subsec:contradiction} fill in the details.  Specifically, we will consider a  counterexample to Theorem \ref{thm:main} that is minimal in two senses.  Suppose that there exists an $r\times c$ cylinder graph with a Hamiltonian path $P$ from $(0, 0)$ to the last column.  Suppose further that $P$ uses $(c-1)+2m$ horizontal edges, but does not end at a row in $A_{r, c, m}.$  Among all such Hamiltonian cylinder graphs and corresponding Hamiltonian paths, we specifically consider an instantiation where:
\begin{enumerate}
\item $r$ and $c$ are minimal with respect to $r+c$, then $r$.  We assume that $r, c\geq 2.$
\item Second, among all counterexamples with $r$ rows and $c$ columns, consider a counterexample path $P$ that is minimal with respect to the {\bf reverse-lexicographic ordering} $(h_{c-2}, h_{c-3}, ..., h_1, h_0)$, where $h_i$ is the number of horizontal edges used in $P$ between the $i$th and $(i+1)$st columns.  
\end{enumerate}

As base cases, we note that if $r=2$ or $c=2$, then Theorem \ref{thm:main} holds: If $r=2,$ then $P$ cannot use any extra horizontal pairs and so must be a GG path.  If $c=2$, then all horizontal edges must be between the first and second column, and this case is handled by Proposition \ref{prop:ggall}.  Finally, for any $r, c$, we note that minimal counterexample $P$ cannot have reverse-lexicographic order $(1, 1, 1, \ldots, 2m-1)$ again by Proposition \ref{prop:ggall}.  

The intuition for considering a counterexample that is minimal with respect to the reverse-lexicographic ordering is as follows: Among all Hamiltonian paths that use $(c-1)+2m$ horizontal edges, GG paths are minimal with respect to this ordering. Our strategy for arriving at a contradiction builds on this intuition. At a high level, the idea of the proof is to devise a sequence of transformations that are applied to $P$, with the intent of creating a Hamiltonian path $P'$ with $\cost(P') \leq \cost(P)$, and such that the reverse lexicographic order of $P'$ is smaller than that of $P$. 

Our transformations (see Section \ref{subsec:cycleprop}) will generate a sequence of subgraphs $P = H_0, H_1, \ldots,H_k$, such that $H_{i+1}$ is obtained by applying a transformation on $H_i$. By properties of the transformations, we will be able to show that $H_i$ satisfies the following three invariants: (1) $H_{i}$ uses at most $(c-1)+2m$ horizontal edges, (2) $H_{i}$ is either a $0$-$v$ Hamiltonian path or the disjoint union of a $0$-$v$ path and a cycle, and (3) $H_{i}$ has a smaller reverse-lexicographic order than that of $P$. 

Observe that if any $H_i$ is a 0-$v$ Hamiltonian path, then this immediately contradicts the minimality of $P$. Hence, we may assume that after applying these transformations, each $H_i$ is the disjoint union of a $0$-$v$  path $P_i$ and a cycle $C_i$.

Let $(P_1, C_1), \ldots, (P_k, C_k)$ be the sequence of path/cycle pairs generated by the sequence of transformations. The next part of the proof (see Sections \ref{subsec:cyclecase1} and \ref{subsec:cyclecase2}) shows that each $C_i$ must be a doubled vertical edge (i.e. a 2-cycle). We do this via a ``backward induction'' argument: We first show that $C_k$ is a 2-cycle. Then, we show inductively that assuming $C_k$, $C_{k-1}, \ldots, C_{i+1}$ are 2-cycles, $C_i$ must also be a 2-cycle. This inductive step is itself a mini proof by contradiction: Assuming $C_i$ is \emph{not} a 2-cycle, we show that the structure of the minimal counterexample implies that that many edges in the graph of $(P_i, C_i)$ are forced. We then use the presence of these forced edges to derive a contradiction. 

Finally, the last step of the proof (see Section  \ref{subsec:contradiction}) completes the argument by showing that in a minimal counterexample, it is not possible for $C_1, \ldots, C_k$ to all be 2-cycles. This is the place where we use the assumption that $r+c$ is minimal: By deleting two specific rows in the cylinder graph, we reduce to a graph with two fewer rows. Then, the fact that Theorem \ref{thm:main} holds on this smaller instance will enable us to show that, in fact the row index of $v$ is in $A_{r,c,m}$ after all. 

Section \ref{s:prelim} gathers some properties of $A_{r,c,m}$ that will be useful. Sections \ref{subsec:cycleprop} to Section \ref{subsec:contradiction} will contain the main proof of Theorem \ref{thm:main}.


\subsection{Preliminaries}
\label{s:prelim}

First, we note some properties of $A_{r,c,m}$ that will be useful later. These largely summarize structural observations of GG paths. See, e.g., Figure \ref{fig:GG3full}.\\
\begin{restatable}{prop}{propAG}
\label{prop:AG}
\begin{enumerate} Let $r, c\geq 2.$  $A_{r,c,m}$ satisfies the following properties.
\item $A_{r, c, m}\subset A_{r, c, m+1}$
\item $ i \in A_{r, c, m} \implies  i \pm 2 \mod r \in A_{r, c, m+1}$
\item $ i \in A_{r, c, m} \implies  i\pm 1 \mod r \in A_{r, c+1, m}.$
\item Let $m'$ be the the largest value of $m$ such that $c+2m < r$. Then $A_{r,c,m} = A_{r,c,m'}$ for all $m \geq m'$. (If $c >r,$ we take $m'=0.$)
\item If $x \in A_{r, c, m}$, then both $x$ and $x +2$ are in $A_{r+2, c, m+1}$. 
\end{enumerate}
\end{restatable}

The proofs of these properties are not especially illuminating and are thus detailed proofs are deferred to Appendix \ref{app:AG}.  Property 1 follows because any vertex reachable in a GG path using at most $m$ extra horizontal pairs is also reachable using at most $m+1$ extra horizontal pairs.  Property 2 summarizes the fact that adding one extra horizontal pair in a GG path allows us to shift the ending row up or down by 2.  Property 3 follows because, when a GG path enters a column past the second, it ends one row above or below where it ended in the previous column.  Property 4 indicates a threshold $m'$ beyond which extra horizontal pairs will not allow us to reach any new vertices in GG paths.  Finally, property 5 summarizes what happens when we both add two rows and one horizontal pair, and will be used in the last step of our proof (see Section  \ref{subsec:contradiction}), where we delete two specific rows in a cylinder graph.
     
Note also that, as a result of our chosen notion of minimality and the observations of Proposition \ref{prop:AG}, we quickly get the following structural result.

\begin{cm} 
\label{cm:3edges}In a minimal counterexample to Theorem \ref{thm:main}, there must be at least 3 edges between the penultimate and final column.  
\end{cm}
This claim says that, in the ordering $(h_{g_1-2}, h_{g_1-3}, ..., h_1, h_0),$ the number of edges $h_{g_1-2} \geq 3$.

\begin{proof}
Suppose not, and let $P$ be a minimal counterexample.  Then $P$ is a Hamiltonian path from $(0, 0)$ to the last column on an $r\times c$ cylinder graph,  $P$ uses $(c-1)+2m$ horizontal edges, and $P$ ends at a vertex $v \not\in A_{r, c, m}.$  Moreover, $c\geq 3.$  Further, $P$ uses exactly one edge between columns $c-2$ and $c-1.$  This edge must be from $(v-1 \mod r, c-2)$ to $(v-1 \mod r, c-1)$ or from $(v+1 \mod r, c-2)$ to $(v+1 \mod r, c-1)$.  In either case, deleting that edge and the last column yields a Hamiltonian path $P'$ on an $r\times (c-1)$ cylinder graph, from $(0, 0)$ to $(v+1 \mod r, c-2)$  or to $(v-1 \mod r, c-2)$.  We assume it is the former, but both cases proceed analogously.  By minimality of $r+c$, $ v+1 \mod r\in A_{r, c-1, m}.$  But by the third statement of Proposition \ref{prop:AG}, this implies that $ v\in A_{r, c, m},$ contradicting our assumptions on $P$. 
\end{proof}

\subsection{Main Argument: Propagating Cycles to the Left} \label{subsec:cycleprop}
With Proposition \ref{prop:Arcm} and Claim \ref{cm:3edges} in hand, we are ready to begin the main proof of Theorem \ref{thm:main}. 
Recall the proof setup described at the beginning of the section: We are assuming a minimal counterexample to Theorem \ref{thm:main} with respect to (1) The number of rows plus columns, then (2) the number of rows, then (3) the reverse lexicographic order of the number of horizontal edges used between each pair of adjacent columns. Throughout this section, we will abbreviate an edge $\{x, y\}$ as $xy$.  



Let $P^*$ be the Hamiltonian path in the minimal counterexample. Then $P^*$ starts at 0 (the top left vertex), and ends at some vertex $v = (v_1, c-1)$ in the last column. Since $P^*$ is a counterexample to the theorem, we have $v_1 \not\in A_{r,c,m}$, where $m$ is the number of extra pairs of horizontal edges used by $P^*$. Our goal is to show that such a minimal counterexample cannot exist by deriving a contradiction. To do this, we use a ``cycle propagating" argument, which is described below. Throughout the course of the argument, we will be gradually transforming $P^*$ into structures that more and more resemble GG paths. The intuition behind the argument is to try to push the horizontal edges used by $P^*$ toward the first two columns, all without increasing the cost. If at any point we end up with a Hamiltonian path $P$ that is either (1) cheaper than $P^*$ or (2) has a smaller reverse-lexicographic ordering than $P^*$, then we will have reached our desired contradiction.

Consider the last horizontal edge in $P^*$: Call it $x_1y_1$, where $x_1$ is the in second-to-last column and $y_1$ is in the last column. At least one of the two horizontal edges immediately above or below $x_1y_1$ must also be in $P^*$. Otherwise, {if both of these horizontal edges were not in $P^*$, then }we would have either (1) $\{y_1, y_1+a_1\} \in P^*$ and $v = y_1-a_1$  or (2) $\{y_1, y_1-a_1\} \in P^*$ and $v = y_1 + a_1$.\footnote{{Since if $\{y_1, y_1 + a_1\} \in P^*$, then $y_1 - a_1$ must have degree 1 since the horizontal edge entering it from the left is not in $P^*$ by assumption; thus $y_1 - a_1 = v$. Similarly if $\{y_1, y_1 - a_1\} \in P^*$, then $y_1 + a_1$ has degree 1 in $P^*$, and so must equal $v$.}} Since $x_1y_1$ is the last horizontal edge in $P^*$,  both cases would imply that there is exactly 1 edge between the last two columns, which contradicts Claim \ref{cm:3edges}. 

So, let $x_2y_2$ be a horizontal edge that is immediately above or below $x_1y_1$, where $x_2$ is in the second-to-last column and $y_2$ is in the last column. Delete the two adjacent horizontal edges $x_1y_1$ and $x_2y_2$, and replace them with the two adjacent vertical edges $x_1x_2$ and $y_1y_2$. (See Figure \ref{fig:cycleprop_start}.) 
\begin{figure}

    \centering
    \begin{tikzpicture}[scale=0.6]
    \foreach \x in {0,1} {
            \foreach \y in {0,1,...,6} {
                \vertex (\x\y) at (\x, \y) {};
            }
        }
    
    \vertex (v) at (1,1) [color=red, label=right:$v$] {};
    
    \draw (14) -- (13) -- (12) -- (11);
    
    \draw (15) -- (16);
    
    \draw[dashed] (04) -- (14);
    \draw[dashed] (05) -- (15);
    
    \draw[blue] (05) -- (04);
    \draw[blue] (15) -- (14);
    
    \vertex (x1) at (05) [label=left:$x_2$] {};
    \vertex (x2) at (04) [label=left:$x_1$] {};
    \vertex (y1) at (15) [label=right:$y_2$] {};
    \vertex (y2) at (14) [label=right:$y_1$] {};
    
    \end{tikzpicture}
    \caption{Here, $x_1y_1$ is the last horizontal edge in $P^*$, and $x_2y_2$ is a horizontal edge in $P^*$ immediately above/below it. To begin the cycle propagating argument, we delete $x_1y_1$ and $x_2y_2$ (dashed) from $P^*$. Then, we add the vertical edges $x_1x_2$ and $y_1y_2$ (blue). The result uses two fewer horizontal edges, and is either a Hamiltonian path or the union of a path and a cycle.}
    \label{fig:cycleprop_start}
\end{figure}
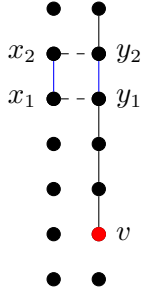

The result is that we have used 2 fewer horizontal edges. Moreover, by Claim \ref{cm:swap} below, we end up with either an 1) $0$-$v$ Hamiltonian path, or 2) the disjoint union of a $0$-$v$ path and a cycle. Note that if $x_1x_2$ was an edge in $P^*$ originally, then we have now used $x_1x_2$ twice; in other words, it is a doubled edge. It is important for the reader to keep in mind that, from now on, the graphs created by our transformation can contain doubled edges. However, because our transformations will preserve the degrees of every node as they were in $P^*$ (so $0$ and $v$ will always have degree 1 and every other node will have degree 2), if a doubled edge is ever created, it will be its own component (i.e. a 2-cycle). 

In Claim \ref{cm:swap} below and in what follows, we say a pair of edges $p_1p_2$, $q_1q_2$ is ``adjacent parallel" if they are either both horizontal and one is immediately above the other, or are both vertical and one is immediately to the right of the other. In other words, $p_1p_2$ and $q_1q_2$ are adjacent parallel if either they are both horizontal and $p_1 = q_1+ a_1$, $p_2 = q_2 + a_1$ (or $p_1 = q_1 - a_1$, $p_2 = q_2 - a_1$), or are both vertical and $p_1 = q_1 + a_2$, $p_2 = q_2 + a_2$ (or $p_1 = q_1 - a_2$, $p_2 = q_2 - a_2$).
 \begin{cm}
    \label{cm:swap}
    Let $P$ be an $s$-$t$ Hamiltonian path in a cylinder graph $G$. Suppose $P$ contains a pair of adjacent parallel edges.   Deleting this pair of edges, and replacing by the opposite pair of adjacent edges yields either an $s$-$t$ Hamiltonian path, or the disjoint union of an $s$-$t$ path and a cycle.   
    \end{cm}
    \begin{proof}
    Suppose the pair of edges is horizontal. (The case where they are vertical is similar.) Let the edges be $x_1y_1$ and $x_2y_2$, where $x_1,x_2$ are in the same column and $y_1,y_2$ are in the same column.
    
    We wish to prove that $Q := P - \{x_1y_1, x_2y_2\} + \{x_1x_2, y_1y_2\}$ is either an $s$-$t$ Hamiltonian path, or the union of an $s$-$t$ path and a cycle. 
    
    Assume, without loss of generality, that in the $s$-$t$ path $P$, $x_1$ is visited before $y_1$, which is visited before $x_2$ or $y_2$. Referring to Figure \ref{fig:cm_swap}, 
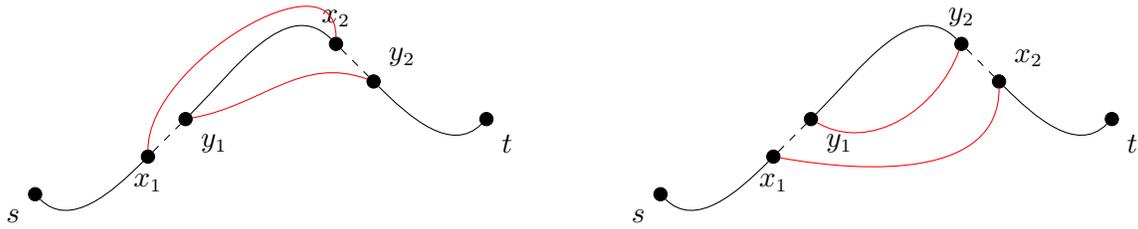
    \begin{itemize}
        \item If $x_2$ is visited before $y_2$, then $Q$ is an $s$-$t$ Hamiltonian path, and 
        \item If $y_2$ is visited before $x_2$, then $Q$ is the disjoint union of an $s$-$t$ path and a cycle. 
    \end{itemize} 
    \end{proof}

    \begin{figure}
    \centering
    \begin{subfigure}[t]{0.45\textwidth}
        \centering
        \begin{tikzpicture}
        \vertex(s) at (0, 0) [label=below left:$s$] {};
        \vertex(t) at (6, 1) [label=below right:$t$] {};
        \vertex(x1) at (1.5, 0.5) [label=below:$x_1$]{};
        \vertex(y1) at (2, 1) [label=below right:$y_1$]{};
        \vertex(x2) at (4, 2) [label=above:$x_2$]{};
        \vertex(y2) at (4.5, 1.5) [label=above right:$y_2$]{};

        
        
        \draw[black] (s) to[out=-45, in=-135] (x1);
        \draw (y1) to[out=45, in=135] (x2); \draw (y2) to[out=-45,in=-135] (t);
        
        \draw[dashed] (x1) -- (y1);
        \draw[dashed] (x2) -- (y2);
        
        \draw[red] (x1) to[out=90,in=90] (x2);
        \draw[red] (y1) to[out=10,in=160] (y2);

        \end{tikzpicture}
        \caption{The case where $x_2$ is visited before $y_2$ in $P$. The black lines (including the two dashed edges) represent $P$. The solid lines (including the two red edges) represent $Q$. In this case, $Q$ is an $s$-$t$ Hamiltonian path.}
    \end{subfigure}%
    ~ \hspace{5mm}
    \begin{subfigure}[t]{0.45\textwidth}
        \centering
        \begin{tikzpicture}
        \vertex(s) at (0, 0) [label=below left:$s$] {};
        \vertex(t) at (6, 1) [label=below right:$t$] {};
        \vertex(x1) at (1.5, 0.5) [label=below:$x_1$]{};
        \vertex(y1) at (2, 1) [label=below right:$y_1$]{};
        \vertex(y2) at (4, 2) [label=above:$y_2$]{};
        \vertex(x2) at (4.5, 1.5) [label=above right:$x_2$]{};

        
        
        \draw[black] (s) to[out=-45, in=-135] (x1);
        \draw (y1) to[out=45, in=135] (y2); \draw (x2) to[out=-45,in=-135] (t);
        
        \draw[dashed] (x1) -- (y1);
        \draw[dashed] (y2) -- (x2);
        
        \draw[red] (x1) to[out=-10,in=-90] (x2);
        \draw[red] (y1) to[out=-30,in=-110] (y2);

        \end{tikzpicture}
        \caption{The case where $y_2$ is visited before $x_2$ in $P$. In this case, $Q$ is the disjoint union of an $s$-$t$ path and a cycle.}
    \end{subfigure}
    \caption{An illustration of the proof of Claim \ref{cm:swap}.}
    \label{fig:cm_swap}
\end{figure}

Observe that the result cannot be a 0-$v$ Hamiltonian path $P$. This is because compared to $P^*$, $P$ uses two fewer horizontal edges and has a smaller reverse-lexicographic ordering. Hence, if $P$ were a 0-$v$ Hamiltonian path, then this would contradict the minimality of $P^*$. So, the result must be a disjoint union of a 0-$v$ path $P$ and a cycle $C$. 

 If a horizontal edge of $C$ is adjacent to a horizontal edge of $P$,  deleting that pair and replacing it with the corresponding pair of vertical edges yields a 0-$v$ Hamiltonian path (see Claim \ref{cm:cycle_path}, which proves this fact in a more general context where $C$ and $P$ are a cycle and path in any graph -- not necessarily cylinder). This Hamiltonian path uses two fewer horizontal edges and has a smaller reverse-lexicographic ordering than $P^*$, which contradicts the minimality of $P^*$. So from now on, we may assume without loss of generality that $C$ and $P$ share no adjacent horizontal edges.

    
    \begin{cm}
    \label{cm:cycle_path}
    Let $P$ be an $s$-$t$ path and $C$ be a cycle in some graph $G$, such that $P$ and $C$ are vertex-disjoint and every vertex of $G$ is contained in either $P$ or $C$. Let $p_1p_2$ be an edge of $P$ and let $q_1q_2$ be an edge of $C$. Assuming that the edges $p_1q_1$ and $p_2q_2$ exist in $G$, deleting the edges $p_1p_2, q_1q_2$ from $P \cup C$ and replacing with the edges $p_1q_1$ and $p_2q_2$ yields an $s$-$t$ Hamiltonian path in $G$.
    \end{cm}
    \begin{proof}
    Without loss of generality, suppose that $p_1$ is visited before $p_2$ in $P$. Then deleting the edges $p_1p_2$ and $q_1q_2$ from $P \cup C$ yields the disjoint union of three paths: 1) An $s$-$p_1$ path $P_1$, 2) an $p_2$-$t$ path $P_2$, and 3) a $q_1$-$q_2$ path $Q$. Adding the edge $p_1q_1$ connects $P_1$ and $Q$ into an $s$-$q_2$ path, and adding the edge $p_2q_2$ connects this path and $P_2$ into a $s$-$t$ Hamiltonian path. For an illustration, see Figure \ref{fig:cm_cyclepath}.
    \end{proof}
    
    \begin{figure}
        \centering
        \begin{tikzpicture}
        \vertex(s) at (0, 0) [label=below left:$s$] {};
        \vertex(t) at (6, 1) [label=below right:$t$] {};
        \vertex(x1) at (1.5, 0.5) [label=above left:$p_1$]{};
        \vertex(y1) at (2, 1) [label=above left:$p_2$]{};
        \coordinate (x2) at (4, 2);
        \coordinate (y2) at (4.5, 1.5);
        
        \vertex(c1) at (2, 0) [label = below:$q_1$] {};
        \vertex(c2) at (2.5, 0.5) [label = below:$q_2$] {};
        \coordinate (c4) at (2.5, -1.5);
        \node (c3) at (4.5, 0.5) {};
        \node (C) at (3, -0.5) [label = $C$]{};
        \node (P) at (4.6, 1.6) [label=$P$]{};

        
        
        \draw[dashed] (c1) -- (c2);
        
        \draw (c2) to[out=45, in=120] (c3) to[out=-60,in=45] (c4) to [out=-135, in=-135] (c1);
        
        \draw[black] (s) to[out=-45, in=-135] (x1);
        \draw (y1) to[out=45, in=135] (x2); 
        \draw (y2) to[out=-45,in=-135] (t);
        
        \draw[dashed] (x1) -- (y1);
        \draw (x2) -- (y2);
        
        \draw[red] (x1) -- (c1);
        \draw[red] (y1) -- (c2);

        \end{tikzpicture}
        \caption{Illustration of the proof of Claim \ref{cm:cycle_path}.}
        \label{fig:cm_cyclepath}
    \end{figure}
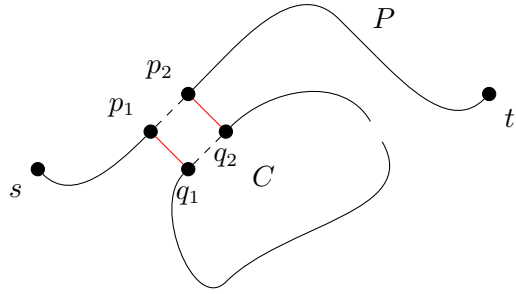


Consider the leftmost vertical edge in $C$: Call it $v_1v_2$. Our cycle propagating argument will terminate when this edge is in the leftmost column, so we assume here that $v_1v_2$ is not in the leftmost column. Note that $C$ has at least one vertical edge, because $x_1x_2$ is one. If there are multiple leftmost vertical edges, break ties by choosing the ``topmost" one. (That is, the one that contains the vertex with the smallest row index. Our convention will be that the wraparound edge from the first vertex of a column to the last vertex is the most top edge.) Henceforth, we will refer to this step as selecting the leftmost, topmost vertical edge of $C$. 

Referring to Figure \ref{subfig:cycle_iter_before}, let $u_1 = v_1 - a_2$, $t_1 = v_1 - 2a_2$, and $w_1 = u_1 - a_1$. Also, let $u_2 = v_2 - a_2$, $t_2 = v_2 - 2a_2$, and $w_2 = u_2 + a_1$. Note that $u_1v_1$ cannot be an edge in $P \cup C$: If it was, it must be an edge of $C$ since $v_1v_2 \in C$. However, since $u_1$ is to the left of $v_1$ and $v_2$, this would imply that $C$ has some vertical edge that is further to the left compared to $v_1v_2$, which contradicts our choice of $v_1v_2$ as a leftmost edge of $C$. 

Furthermore, $u_1u_2$ cannot be an edge in $P \cup C$. If it were, then it must be in $P$ since it is a vertical edge that is to the left of $v_1v_2$. In that case, by Claim \ref{cm:cycle_path}, deleting the two vertical edges $u_1u_2, v_1v_2$ and replacing them with the two horizontal edges $u_1v_1$, $u_2v_2$ results in a 0-$v$ Hamiltonian path. Moreover, this Hamiltonian path has a smaller reverse-lexicographic order than $P^*$,\footnote{{The reason why this Hamiltonian path has a strictly smaller reverse-lexicographic order than $P^*$ is because $v_1v_2$ is not in the last column. (Since $x_1x_2$ is a vertical edge in C that is not in the last column, $v_1v_2$ being the leftmost such edge, cannot be in the last column either.)}} which contradicts the minimality of $P^*$. 

So far, we have established that the edges $u_1v_1$ and $u_1u_2$ cannot be present in $P \cup C$. Consider the vertex $u_1$: It cannot be 0 or $v$, so it must have degree 2 in $P \cup C$. (It cannot be $v$ because it is not in the last column. It cannot be 0 because if it were, then $u_2u_1$ and $u_2w_2$ must be edges in $P$ because (1) $u_2$ has degree 2 and (2) $v_1v_2$ is a leftmost edge in $C$. {This already contradicts the fact that $u_1u_2$ cannot be in $P \cup C$.}) Since $u_1$ has degree 2 in $P \cup C$, this implies that the two edges $u_1w_1$ and $u_1t_1$ are both in $P \cup C$. Moreover, we must have $u_1w_1, u_1t_1 \in P$, because $v_1v_2$ is the leftmost edge of $C$. 

The same argument shows that $u_2w_2, u_2t_2 \in P$. 

Now, perform the following operation: Delete the three edges $v_1v_2, t_1u_1, t_2u_2$, and replace them with the three edges $t_1t_2, u_1v_1, u_2v_2$.  (Refer to Figure \ref{subfig:cycle_iter_after}.) We will refer to this operation as the ``cycle propagating operation".

    \begin{figure}
    \centering
    \begin{subfigure}[t]{0.45\textwidth}
        \centering
        \begin{tikzpicture}
        \vertex(v1) at (4, 4) [label=right:$v_1$] {};
        \vertex(v2) at (4, 3) [label=right:$v_2$] {};
        \vertex(u1) at (3, 4) [label= right:$u_1$] {};
        \vertex(u2) at (3, 3) [label= right:$u_2$]{};
        \vertex(u0) at (3, 5) [label=right:$w_1$]{};
        \vertex(u3) at (3, 2) [label=right:$w_2$]{};
        \vertex(t1) at (2, 4) [label= left:$t_1$] {};
        \vertex(t2) at (2, 3) [label= left:$t_2$] {};
        \draw (v1) -- (v2);
        \draw (u0) -- (u1) -- (t1);
        \draw (u3) -- (u2) -- (t2);
        \end{tikzpicture}
        \caption{Before. $v_1v_2$ is the leftmost, topmost edge of the cycle $C$. The other edges in the figure are part of the path $P$.}
        \label{subfig:cycle_iter_before}
    \end{subfigure}%
    ~ 
    \begin{subfigure}[t]{0.45\textwidth}
        \centering
        \begin{tikzpicture}
        \vertex(v1) at (4, 4) [label=right:$v_1$] {};
        \vertex(v2) at (4, 3) [label=right:$v_2$] {};
        \vertex(u1) at (3, 4) [label= left:$u_1$] {};
        \vertex(u2) at (3, 3) [label= left:$u_2$]{};
        \vertex(u0) at (3, 5) [label=right:$w_1$]{};
        \vertex(u3) at (3, 2) [label=right:$w_2$]{};
        \vertex(t1) at (2, 4) [label= left:$t_1$] {};
        \vertex(t2) at (2, 3) [label= left:$t_2$] {};
        \draw (t1) -- (t2);
        \draw (u0) -- (u1) -- (v1);
        \draw (u3) -- (u2) -- (v2);
        \end{tikzpicture}
        \caption{After. The three edges $\{v_1v_2, t_1u_1, t_2u_2\}$ have been deleted and replaced with $\{t_1t_2, u_1v_1, u_2v_2\}$. The result is either a 0-$v$ Hamiltonian path or the disjoint union of a 0-$v$ path and a cycle.}
        \label{subfig:cycle_iter_after}
    \end{subfigure}
    \caption{An illustration of one step of the cycle propagating procedure.}
    \label{fig:cycle_iter}
\end{figure}

Observe that the resulting subgraph uses the same number of horizontal edges as before, and has a strictly smaller reverse-lexicographic order (since the operation pushed two horizontal edges one column to the left). Furthermore, by Claim \ref{cm:cycle_prop}, the resulting structure is either 1) a 0-$v$ Hamiltonian path, or 2) the union of an 0-$v$ path with a cycle.

\begin{cm}
    \label{cm:cycle_prop}
    After one iteration of the cycle propagating procedure, the resulting subgraph is either 1) a 0-$v$ Hamiltonian path, or 2) the disjoint union of a 0-$v$ path with a cycle. 
    \end{cm}
    \begin{proof}
    Let $Q$ be the resulting subgraph. Note that the cycle propagating operation does not change the degree of any vertex. Hence, in $Q$, the degrees of 0 and $v$ are 1, and the degree of every other vertex is 2. This implies that $Q$ is the disjoint union of a 0-$v$ path and zero or more cycles. To prove the claim, we will show that $Q$ consists of at most 2 connected components. 
    
    Let $Q' = P \cup C \cup \{u_1v_1, u_2v_2, t_1t_2\}$, so that $Q = Q' - \{t_1u_1, t_2u_2, v_1v_2\}$. (Refer to Figure \ref{fig:cycle_iter} for an illustration of the location of these vertices in the cylinder graph.) Observe that $Q'$ consists of one connected component, because the edges $u_1v_1$ and $u_2v_2$ connect $P$ and $C$. Now, imagine deleting the edges $\{t_1u_1, v_1v_2, t_2u_2\}$ one by one from $Q'$ to obtain $Q$. $t_1u_1$ is not a bridge\footnote{A bridge in a graph $G$ is an edge whose removal increases the number of components of $G$.} in $Q'$, because it is contained in the cycle $t_1u_1v_1v_2u_2t_2t_1$. Hence $Q' - \{t_1u_1\}$ is connected. 
    
    Next, $v_1v_2$ is not a bridge in $Q' - \{t_1u_1\}$, because it is contained in the cycle $C$. Hence $Q' - \{t_1u_1, v_1v_2\}$ is connected. 
    Finally, $t_2u_2$ may be a bridge edge in $Q' - \{t_1u_1, v_1v_2\}$, but deleting it can create at most one new component. Thus $Q$ consists of at most 2 components, which completes the proof. 
        \end{proof}
    
If the result is a 0-$v$ Hamiltonian path $P$, then we are done because $P$ uses two fewer horizontal edges than $P^*$, and has a smaller reverse-lexicographic ordering.

Thus, after one step of the cycle propagating procedure, we may assume that the result is the disjoint union of a 0-$v$ path and a cycle. Moreover, observe that the leftmost vertical edge of the new cycle is further to the left than that of the old cycle. (Referring to Figure \ref{fig:cycle_iter}, at least one of the vertical edges $t_1t_2$, $u_1w_1$ or $u_2w_2$ are in the new cycle. All three of these edges are further to the left than $v_1v_2$, which was a leftmost vertical edge in the old cycle.) Hence, iterating this process will either result in a 0-$v$ Hamiltonian path at some step (in which case we are done), or we get a sequence of paths and cycles $(P_1, C_1), \ldots, (P_k, C_k)$ that satisfy the following properties.
\begin{prop}
    \label{prop:cycle_path}
    The sequence of paths and cycles $(P_1, C_1), \ldots, (P_k, C_k)$ resulting from the cycle propagating procedure satisfy the following properties.
    \begin{enumerate}[P1.]
        \item For each $i$, $P_i$ is a 0-$v$ path and $C_i$ is a cycle,
        \item For each $i$, $P_i, C_i$ are vertex-disjoint and together
        cover all the vertices of the graph,
        \item For each $i$, $P_i$ and $C_i$ have no horizontal adjacent edges. (Otherwise by Claim \ref{cm:cycle_path}, we could delete them and replace them with the corresponding pair of vertical edges to get a Hamiltonian path $P$ that (1) uses 4 fewer horizontal edges than $P^*$ and (2) has a smaller reverse lexicographic order.)
        \item $(P_{i+1}, C_{i+1})$ is obtained from $(P_i, C_i)$ by applying the cycle propagating procedure on the leftmost, topmost edge of $C_i$, and 
        \item The last cycle, $C_k$, has a vertical edge in the first column. 
    \end{enumerate}
    Moreover, we have the following property:
    \begin{enumerate}[P1.]
        \setcounter{enumi}{5}
        \item For each $i$, let $b_i$ be the index of the column that contains the leftmost vertical edge of $C_i$. Suppose $e = \{(a, b_i), (a+1, b_i)\} \in C_i$ is some vertical edge in $C_i$ that is in column $b_i$. Let $f = \{(a, b_i-1), (a+1, b_i-1)\}$ be the vertical edge that is immediately to the left of $e$. Then $f$ cannot be in $P_i$. 
    \end{enumerate}
\end{prop}
\begin{proof}
Properties P1 -- P5 are either self-explanatory or have been proved earlier. 

Property P6 is true, because otherwise, we could delete $e$ and $f$, and replace them with the corresponding pair of horizontal edges. This would give a Hamiltonian path that 1) uses the same number of horizontal edges as $P^*$, and 2) has a smaller reverse lexicographic order than $P^*$. This contradicts our choice of minimal counterexample. Note that 2) is true because $b_i$ was the leftmost column containing a vertical edge of $C_i$, so in particular $b_i \leq c-2$.
\end{proof}
    Note that any or all of the cycles $C_i$ can be 2-cycles (i.e. a doubled edge.) Referring to Figure \ref{fig:cycle_iter}, the only way a doubled edge can be created by the cycle propagating procedure is if before some iteration (see Figure \ref{subfig:cycle_iter_before}), the edge $t_1t_2$ is already present in $P \cup C$. Then, after the iteration (see Figure \ref{subfig:cycle_iter_after}), the edge $t_1t_2$ is doubled. In particular, this means that only vertical edges can be doubled; a doubled horizontal edge is not possible.
 
We now look at the sequence of cycles $C_1, C_2, \ldots, C_k$, and consider cases depending on which of them are 2-cycles. In Section \ref{subsec:cyclecase1}, we will show that $C_k$ must be a 2-cycle. In Section \ref{subsec:cyclecase2}, we will extend the argument to show that, in fact, every cycle must be a 2-cycle. Finally, we use this structure in Section \ref{subsec:cyclecase3} to arrive at a contradiction.

\subsection{Last cycle is a 2-cycle}
\label{subsec:cyclecase1}

\begin{cm}
\label{cm:cyclecase1}
The last cycle $C_k$ must be a 2-cycle.
\end{cm}
\begin{proof}
Suppose for a contradiction that $C_k$ is not a 2-cycle. Let $S$ be the maximal contiguous set of vertices in the first column such that
\begin{enumerate}
    \item $0 \in S$, and
    \item No vertex in $S$ is in $C_k$. 
\end{enumerate}
Let $x$ and $y$ be the endpoints of $S$. Let $x'$ and $y'$ be the vertices in the first column, not in  $S$, that are adjacent to $x$ and $y$ respectively. (See Figure \ref{fig:cm_cyclecase1}.) Note that $x', y'$ exist, are distinct, and belong to $C_k$. 

\begin{figure}
    \centering
    \begin{tikzpicture}[scale=0.8]
    \foreach \x in {0,1} {
            \foreach \y in {0,1,...,9} {
                \vertex (\x\y) at (\x, \y) {};
            }
        }
    \foreach \y in {0, 1, 2, 9, 8, 7} {
        \vertex (\y) at (0, \y) [color=red] {};
    }
    
    \vertex (x) at (0, 7) [color=red, label = left:$x$] {};
    \vertex (y) at (0, 2) [color=red, label = left:$y$] {};
    
    \vertex (x') at (0, 6) [label=left: $x'$] {};
    \vertex (y') at (0, 3) [label=left: $y'$] {};
    
    \draw[blue] (04) -- (y') -- (13);
    \draw[blue] (05) -- (x') -- (16);
    
    \end{tikzpicture}
    \caption{An illustration of the proof of Claim \ref{cm:cyclecase1}. The depicted nodes are the first two columns of the graph. The red nodes are the nodes in $S$. The blue edges are in $C_k$. Since the path and the cycle cannot have any adjacent horizontal edges, this implies that $x$ and $y$ both have degree 1 in the path, which is impossible.}
    \label{fig:cm_cyclecase1}
\end{figure}
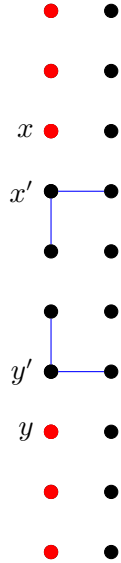

Since $C_k$ is a cycle, both $x'$ and $y'$ have degree 2. Now $x'$ cannot be incident to $x$ in $C_k$, so it must be incident to $x' + a_2$ and $x' + a_1$ in $C_k$. (Here, we use the assumption that $C_k$ is not a 2-cycle.)

Consider $x$. If $x$ has degree 2 in $P_k$, then it would be incident to $x + a_2$. Then $C_k$ and $P_k$ would have a pair of adjacent horizontal edges, which is leads to a contradiction via Claim \ref{cm:cycle_path}. Therefore $x$ has degree 1 in $P_k$. The same argument shows that $y$ must have degree 1 in $P_k$. This is a contradiction, since $P_k$ has exactly one vertex of degree 1 in the first column, namely vertex 0.

\end{proof}

\subsection{Every cycle is a 2-cycle}
\label{subsec:cyclecase2}
In the previous subsection, we showed that the last cycle $C_k$ must be a 2-cycle. Next, we will show that in fact every cycle must be a 2-cycle. Throughout this section, when we write the coordinates of a point in a cylinder graph, we will use $(i, j)$ to implicitly mean $(i \mod r, j \mod c)$. This is to make the notation less cluttered.
\begin{cm}
\label{cm:caseb2}
For all $1\leq i \leq k$, $C_i$ is a 2-cycle.
\end{cm}
\noindent\textit{Proof.}
Suppose for a contradiction that some cycle is not a 2-cycle. 
We consider the last cycle in the sequence that is not a 2-cycle. That is, let $i$ be the index such that 
\begin{enumerate}
    \item $C_i$ is not a 2-cycle, and
    \item $C_{i+1}, \ldots, C_k$ are all 2-cycles. 
\end{enumerate}
Note that by the previous claim, $i \leq k-1$. We will show this case cannot happen using an argument similar to the one used in the proof of the previous claim, but that is more involved. Essentially, we will argue that many connections in the graph of $(P_i, C_i)$ are forced; that is, there are certain edges we can be sure are present in $P_i$ or $C_i$. We will then show that the presence of these forced edges implies that two vertices in the first column have degree one in $P_i$ if $C_i$ is not a 2-cycle, which will give our desired contradiction. 


Let the two vertices of $C_k$ be $(a, 0)$ and $(a+1, 0)$, where $0\leq a \leq r-1$. Note that in fact $1 \leq a \leq r-2$, because $(0, 0)$ cannot be a vertex in $C_k$; $(0, 0)$ is in $P_k$. Because $C_k, C_{k-1}, \ldots, C_{i+1}$ are all 2-cycles, it follows from the way we propagated cycles in Section \ref{subsec:cycleprop} that for all $0 \leq j \leq k-i-1$, the following vertices are consecutive on the path $P_i$ (see Figure \ref{subfig:case_b2a} for an illustration):
$$(a-1, 2j+1), \,(a, 2j+1),\,(a, 2j),\, (a+1, 2j),\, (a+1, 2j+1),\, (a+2, 2j+1)$$
Also, $2(k-i)$ is the index of the leftmost column of $C_i$, and the edge between the vertices $(a, 2(k-i))$ and $(a+1, 2(k-i))$ is the leftmost, topmost vertical edge of $C_i$. (Again, see Figure \ref{subfig:case_b2a}.)

We proceed to argue that certain edges in the graph are forced to be in either $P_i$ or $C_i$. For convenience, let $b = 2(k-i)$ be the index of the column containing the leftmost edge of $C_i$. Then the leftmost, topmost vertical edge of $C_i$ is between $(a, b)$ and $(a+1, b)$. (This is the blue edge in Figure \ref{subfig:case_b2a}.) First, we will show that $a = 1$.

\begin{wrapfigure}[8]{r}{0.3\textwidth}
\label{wrap-fig:1}
\centering
\begin{tikzpicture}[scale=0.5]
    \foreach \x in {0,1,...,6} {
            \foreach \y in {0,...,5} {
                \smvertex (\x\y) at (\x, \y) {};
            }
        }
    
        \foreach \i in {1,3} {
            \draw (\i,4) -- (\i, 3) -- (\i-1,3) -- (\i-1, 2) -- (\i, 2) -- (\i, 1);
        }
        \draw[blue] (53) -- (43) -- (42) -- (52);
        
        
\end{tikzpicture}
\end{wrapfigure} 
To begin, observe that $(a, b)$ cannot be connected to $(a-1, b)$ in $C_i$; otherwise, if $\{(a, b), (a-1, b)\}$ were an edge in $C_i$, then it would be adjacent to the vertical edge $\{(a, b-1), (a-1,b-1)\}$ in $P_i$.  This is impossible by property P6 in Proposition \ref{prop:cycle_path}, because  $\{(a,b-1), (a-1,b-1)\}$ is in $P_i$. Hence, in $C_i$, if $C_i$ is not a 2-cycle, then $(a, b)$ must have a horizontal edge to $(a, b+1)$. A similar argument shows that $(a+1, b)$ must have a horizontal edge to $(a+1, b+1)$ in $C_i$.

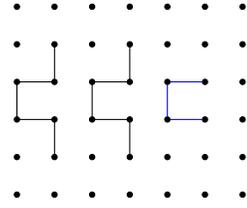
\noindent\begin{wrapfigure}[8]{r}{0.3\textwidth}
\label{wrap-fig:2}
\centering
\begin{tikzpicture}[scale=0.5]
\foreach \x in {0,1,...,6} {
            \foreach \y in {0,...,5} {
                \smvertex (\x\y) at (\x, \y) {};
            }
        }
    
        \foreach \i in {1,3} {
            \draw (\i,4) -- (\i, 3) -- (\i-1,3) -- (\i-1, 2) -- (\i, 2) -- (\i, 1);
        }
        \draw[blue] (53) -- (43) -- (42) -- (52);
        
        
        \draw[dotted] (44) -- (54);
        
\end{tikzpicture}
\end{wrapfigure} 
Next, we claim that the edge between $(a-1, b)$ and $(a-1, b+1)$ (the dotted edge in the figure on the right) cannot exist in either $P_i$ or $C_i$. Indeed, if that edge existed in $P_i$, then $P_i$ and $C_i$ would share an adjacent pair of horizontal edges, a contradiction. On the other hand, if that edge existed in $C_i$, then so must the edge between $(a-1, b)$ and $(a-2, b)$, which (since $a \geq 1$), contradicts the fact that $\{(a, b), (a+1, b)\}$ is the leftmost topmost vertical edge in $C_i$. 

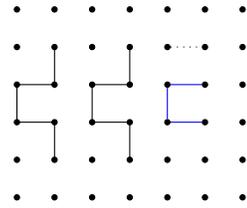
\begin{wrapfigure}[8]{r}{0.3\textwidth}
\label{wrap-fig:3}
\centering
\begin{tikzpicture}[scale=0.5]
\foreach \x in {0,1,...,6} {
            \foreach \y in {0,...,5} {
                \smvertex (\x\y) at (\x, \y) {};
            }
        }
    
        \foreach \i in {1,3} {
            \draw (\i,4) -- (\i, 3) -- (\i-1,3) -- (\i-1, 2) -- (\i, 2) -- (\i, 1);
        }
        \draw[blue] (53) -- (43) -- (42) -- (52);
        
        \draw (34) -- (44) -- (4,5);
        
        \smvertex (x) at (4, 4) [color=red] {};
        
\end{tikzpicture}
\end{wrapfigure} 

Now, note that $(a-1, b)$ (the red vertex in the diagram to the right) must have degree 2 in $P_i \cup C_i$. Since $(a-1, b)$ cannot have an edge to the vertex to its right (by the previous paragraph), or to the vertex below it (since that vertex already has degree 2), it must have edges to the vertices above it and to its left. In other words, $(a-1, b)$ must have edges to $(a-2, b)$ and $(a-1, b-1)$. Note that the edges $\{(a-1, b), (a-2, b)\}$ and $\{(a-1, b), (a-1, b-1)\}$ must both belong to $P_i$, because the vertex $(a-1, b-1)$ is in $P_i$.  

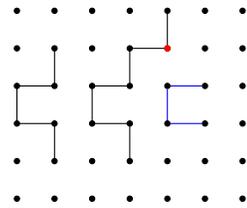
\begin{wrapfigure}[10]{r}{0.3\textwidth}
\label{wrap-fig:3.5}
\centering
\begin{tikzpicture}[scale=0.5]
\foreach \x in {0,1,...,6} {
            \foreach \y in {0,...,5} {
                \smvertex (\x\y) at (\x, \y) {};
            }
        }
    
        \foreach \i in {1,3} {
            \draw (\i,4) -- (\i, 3) -- (\i-1,3) -- (\i-1, 2) -- (\i, 2) -- (\i, 1);
        }
        \draw[blue] (53) -- (43) -- (42) -- (52);
        
        \draw (34) -- (44) -- (4,5);
        \draw (14) -- (24) -- (25);
        
        \smvertex (x) at (2, 4) [color=red] {};
        
\end{tikzpicture}
\end{wrapfigure} 

Now, consider the vertex $(a-1, b-2)$ (the red vertex in the diagram to the right). If $b > 2$, then $(a-1, b-2)$ has degree 2 in $P_i \cup C_i$. Since its right neighbor $(a-1, b-1)$ and bottom neighbor $(a, b-2)$ already have degree 2, it must have edges to its top neighbor and left neighbor. As in the previous paragraph, these two edges $\{(a-1, b-2), (a-2, b-2)\}$ and $\{(a-1, b-2), (a-1, b-3)\}$ must both belong to $P_i$, because the vertex $(a-1, b-3)$ is in $P_i$. Applying this argument repeatedly, we get that for all $j \in \{b-1, b-3, \ldots, 1\}$, the vertex $(a-1, j)$ has edges to $(a-2, j)$ and $(a-1, j-1)$, and that these edges are all in $P_i$.

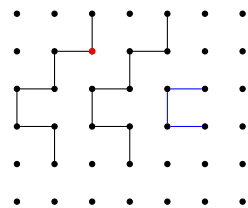
\begin{wrapfigure}[12]{r}{0.3\textwidth}
\label{wrap-fig:3.6}
\centering
\begin{tikzpicture}[scale=0.5]
\foreach \x in {0,1,...,6} {
            \foreach \y in {0,...,5} {
                \smvertex (\x\y) at (\x, \y) {};
            }
        }
    
        \foreach \i in {1,3} {
            \draw (\i,4) -- (\i, 3) -- (\i-1,3) -- (\i-1, 2) -- (\i, 2) -- (\i, 1);
        }
        \draw[blue] (53) -- (43) -- (42) -- (52);
        
        \draw (34) -- (44) -- (4,5);
        \draw (14) -- (24) -- (25);
        \draw (04) -- (05);
        
        \smvertex (x) at (0, 4) [color=red] {};
        
\end{tikzpicture}
\end{wrapfigure} 
We are now ready to prove that $a=1$. Consider the vertex $(a-1, 0)$ (red in the diagram to the right). Since it is in the first column, it cannot have any edge to its left. Moreover, its right neighbor $(a-1, 1)$ already has degree 2 by the arguments in the previous paragraph, and its bottom neighbor $(a, 0)$ also already has degree 2 (we knew this at the start of this proof). Therefore, the only possible edge out of $(a-1, 0)$ is the edge to its top neighbor. This implies that $(a-1, 0)$ has degree 1 in $P_i \cup C_i$. This implies $a = 1$, because $(0, 0)$ is the only vertex in the first column with degree 1. Previously we drew our diagrams with $a > 1$ in order to be general, but now that we have shown $a = 1$, we will draw our diagrams with $a=1$ from now on.

\noindent
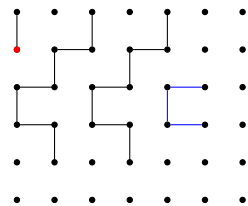
\begin{wrapfigure}[8]{r}{0.3\textwidth}
\label{wrap-fig:4}
\centering
\begin{tikzpicture}[scale=0.5]
\foreach \x in {0,1,...,6} {
            \foreach \y in {0,...,7} {
                \smvertex (\x\y) at (\x, \y) {};
            }
        }
    
        \foreach \i in {1,3} {
            \foreach \j in {7} {
                \draw (\i,\j) -- (\i, \j-1) -- (\i-1,\j-1) -- (\i-1, \j-2) -- (\i, \j-2) -- (\i, \j-3) -- (\i-1, \j-3) -- (\i-1, \j-4);
                \draw (\i, \j) -- (\i+1, \j) -- (\i+1, \j+1);
            }
        }
        \draw[blue] (56) -- (46) -- (45) -- (55);
        
        \draw (07) -- (0,8);
        
        \foreach \i in {0,2,4} {
            \draw (\i,-1) -- (\i, 0);
        }
        
        
        \smvertex (x) at (0, 4) [color=red] {};
        \smvertex (x') at (2, 4) [color=red] {};
        
\end{tikzpicture}
\end{wrapfigure} 
Now, consider the vertex $(3, 0)$. If $r > 3$, then it must have degree 2 in $P_i \cup C_i$. Since it cannot have an edge to $(2, 0)$, it must have edges to $(3, 1)$ and $(4, 0)$. Next, consider the vertex $(3, 2)$. Again, it must have degree 2 in $P_i \cup C_i$. It cannot have edges to $(3, 1)$ or $(2, 2)$, because these already have degree 2. Hence, it must have edges to $(4, 1)$ and $(3, 3)$. Iteratively applying this argument, we get that for all $j \in \{0,2,\ldots, b-2\}$, the vertex $(3, j)$ (red in the diagram) has edges to $(3, j+1)$ and $(4, j)$. 

\noindent\begin{wrapfigure}[8]{r}{0.3\textwidth}
\label{wrap-fig:4.5}
\centering
\begin{tikzpicture}[scale=0.5]
\foreach \x in {0,1,...,6} {
            \foreach \y in {0,...,7} {
                \smvertex (\x\y) at (\x, \y) {};
            }
        }
    
        \foreach \i in {1,3} {
            \foreach \j in {7} {
                \draw (\i,\j) -- (\i, \j-1) -- (\i-1,\j-1) -- (\i-1, \j-2) -- (\i, \j-2) -- (\i, \j-3) -- (\i-1, \j-3) -- (\i-1, \j-4);
                \draw (\i, \j) -- (\i+1, \j) -- (\i+1, \j+1);
            }
        }
        \draw[blue] (56) -- (46) -- (45) -- (55);
        \draw[blue] (54) -- (44) -- (43);
        
        \draw (07) -- (0,8);
        
        \foreach \i in {0,2,4} {
            \draw (\i,-1) -- (\i, 0);
        }
        
        
        \smvertex (x'') at (4, 4) [color=red] {};
\end{tikzpicture}
\end{wrapfigure} 
Now, consider the vertex $(3, b)$ (red). It must have degree 2 in $P_i \cup C_i$. However, it cannot have an edge to its top neighbor $(2,b)$ or to its left neighbor $(3, b-1)$, because these already have degree 2. Thus, it must have edges to $(3, b+1)$ and $(4, b)$. If these two edges belonged to $P_i$, then the horizontal edge $\{(3, b), (3, b+1)\}$ in $P_i$ would be adjacent to the horizontal edge $\{(2, b), (2, b+1)\}$ in $C_i$, which contradicts property P3 in Proposition \ref{prop:cycle_path}. Therefore, these two edges must belong to $C_i$. 

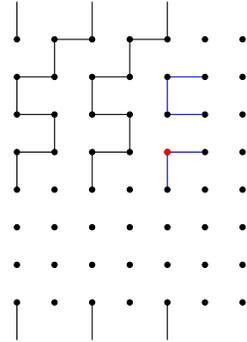
\noindent\begin{wrapfigure}[11]{r}{0.3\textwidth}
\label{wrap-fig:4.6}
\centering
\begin{tikzpicture}[scale=0.5]
\foreach \x in {0,1,...,6} {
            \foreach \y in {0,...,7} {
                \smvertex (\x\y) at (\x, \y) {};
            }
        }
    
        \foreach \i in {1,3} {
            \foreach \j in {7} {
                \draw (\i,\j) -- (\i, \j-1) -- (\i-1,\j-1) -- (\i-1, \j-2) -- (\i, \j-2) -- (\i, \j-3) -- (\i-1, \j-3) -- (\i-1, \j-4);
                \draw (\i, \j) -- (\i+1, \j) -- (\i+1, \j+1);
            }
        }
        \draw[blue] (56) -- (46) -- (45) -- (55);
        \draw[blue] (54) -- (44) -- (43) -- (53);
        
        \draw (07) -- (0,8);
        
        \foreach \i in {0,2,4} {
            \draw (\i,-1) -- (\i, 0);
        }
        
        
        \smvertex (x'') at (4, 3) [color=red] {};
\end{tikzpicture}
\end{wrapfigure} 
Consider the the vertex $(4, b)$ (red). It is in $C_i$ and must have degree 2. It cannot have an edge going to the left to $(4, b-1)$, since otherwise $C_i$ would have an edge further to the left than its leftmost edge. Also, it cannot have an edge going down to $(5, b)$: If $\{(4, b), (5, b)\}$ were an edge in $C_i$, then the the edge $\{(4, b-1), (5, b-1)\}$, immediately left to it, would be forced to be in $P_i$. This would mean that a vertical edge in the leftmost column of $C_i$ is adjacent to a vertical edge of $P_i$, which contradicts property P6 in Proposition \ref{prop:cycle_path}. Therefore $(4, b)$ must have an edge going to the right to $(4, b+1)$, and this edge is in $C_i$ because $(4,b)$ is in $C_i$. 

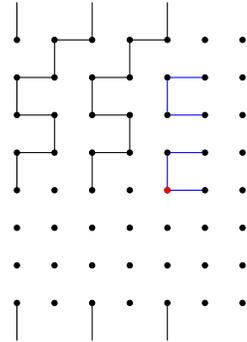
\noindent\begin{wrapfigure}[7]{r}{0.3\textwidth}
\label{wrap-fig:4.7}
\centering
\begin{tikzpicture}[scale=0.5]
\foreach \x in {0,1,...,6} {
            \foreach \y in {0,...,7} {
                \smvertex (\x\y) at (\x, \y) {};
            }
        }
    
        \foreach \i in {1,3} {
            \foreach \j in {7} {
                \draw (\i,\j) -- (\i, \j-1) -- (\i-1,\j-1) -- (\i-1, \j-2) -- (\i, \j-2) -- (\i, \j-3) -- (\i-1, \j-3) -- (\i-1, \j-4);
                \draw (\i, \j) -- (\i+1, \j) -- (\i+1, \j+1);
                \draw (\i-1, \j-4) -- (\i, \j-4) -- (\i, \j-5);
            }
        }
        \draw[blue] (56) -- (46) -- (45) -- (55);
        \draw[blue] (54) -- (44) -- (43) -- (53);
        
        \draw (07) -- (0,8);
        
        \foreach \i in {0,2,4} {
            \draw (\i,-1) -- (\i, 0);
        }
        
        
        \smvertex (x) at (1,3) [color=red] {};
        \smvertex (x') at (3,3) [color=red] {};
\end{tikzpicture}
\end{wrapfigure} 
Consider the the vertex $(4, b-1)$. It must have degree 2, and its top and right neighbors are already have degree 2, which implies it has edges to its left neighbor $(4, b-2)$ and its bottom neighbor $(5, b-1)$. These edges are in $P_i$, because $(4, b-2)$ is. Repeating this argument, we get that for all $j \in \{b-1, b-3, \ldots, 1\}$, the vertex $(4, j)$ (red in the diagram), has edges to $(4, j-1)$ and $(5, j)$, and these edges are all in $P_i$. 

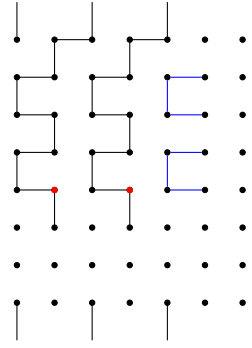
\noindent\begin{wrapfigure}[11]{r}{0.3\textwidth}
\label{wrap-fig:4.8}
\centering
\begin{tikzpicture}[scale=0.5]
\foreach \x in {0,1,...,6} {
            \foreach \y in {0,...,7} {
                \smvertex (\x\y) at (\x, \y) {};
            }
        }
    
        
        \foreach \i in {1,3} {
            \foreach \j in {7,5,3} {
                \draw (\i,\j) -- (\i, \j-1) -- (\i-1,\j-1) -- (\i-1, \j-2) -- (\i, \j-2);
            }
            \draw (\i, 7) -- (\i+1, 7) -- (\i+1, 8);
        }
        
        \foreach \j in {6,4,2} {
            \draw[blue] (5, \j) -- (4, \j) -- (4, \j-1) -- (5, \j-1);
        }
        
        \draw (07) -- (0,8);
        \draw (10) -- (11);
        \draw (30) -- (31);
        
        \foreach \i in {0,2,4} {
            \draw (\i,-1) -- (\i, 0);
        }
        
        
\end{tikzpicture}
\end{wrapfigure} 
Repeating the argument in the previous 4 paragraphs, we can propagate this structure all the way down the rows. The exact way this structure ends depends on the parity of $r$. First, suppose $r$ is even. Then the structure is depicted in the diagram on the right. To be precise, we have that for all $l \in \{1, 3, 5, \ldots, r-3\}$, the vertices $(l, b+1), (l, b), (l+1, b), (l+1, b+1)$ are consecutive in $C_i$. (These are the blue paths of length 3 in the diagram.) Also, for all $l \in \{1,3,5, \ldots, r-3\}$ and $j \in \{1,3,\ldots, b-1\}$, the vertices $(l-1, j), (l, j), (l, j-1), (l+1, j-1), (l+1, j), (l+2, j)$ are consecutive in $P_i$. (These form the zig-zagging subpaths of $P_i$ going downward.)

\noindent
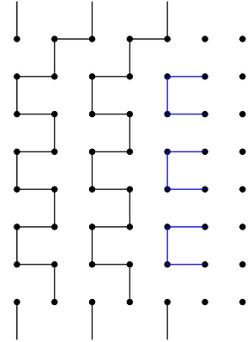
\begin{wrapfigure}[12]{r}{0.3\textwidth}
\label{wrap-fig:4.9}
\centering
\begin{tikzpicture}[scale=0.5]
\foreach \x in {0,1,...,6} {
            \foreach \y in {0,...,7} {
                \smvertex (\x\y) at (\x, \y) {};
            }
        }
    
        
        \foreach \i in {1,3} {
            \foreach \j in {7,5,3} {
                \draw (\i,\j) -- (\i, \j-1) -- (\i-1,\j-1) -- (\i-1, \j-2) -- (\i, \j-2);
            }
            \draw (\i, 7) -- (\i+1, 7) -- (\i+1, 8);
        }
        
        \foreach \j in {6,4,2} {
            \draw[blue] (5, \j) -- (4, \j) -- (4, \j-1) -- (5, \j-1);
        }
        
        \draw (07) -- (0,8);
        \draw (10) -- (11);
        \draw (30) -- (31);

        \foreach \i in {0,2,4} {
            \draw (\i,-1) -- (\i, 0);
            \draw (\i, 0) -- (\i+1, 0);
        }
        
        
        \smvertex (x) at (0,0) [color=red] {};
        \smvertex (x') at (2,0) [color=red] {};
        \smvertex (x'') at (4, 0) [color=red] {};
\end{tikzpicture}
\end{wrapfigure} 
At this point, we have almost exactly specified all the edges in $P_i \cup C_i$ in the first $b$ columns in the case where $r$ is even. To complete the picture, observe that the vertex $(r-1, 0)$ must have degree 2. It cannot be connected to its top neighbor $(r-2, 0)$, because $(r-2, 0)$ already has degree 2. Thus, it is connected to its right neighbor, $(r-1, 1)$. Similarly, the vertex $(r-1, 2)$ has degree 2, so it must be connected to $(r-1, 3)$. Repeating this argument, we get that for all $j \in \{0, 2, \ldots, b\}$, the vertex $(r-1, j)$ (red in the diagram to the right), has an edge to $(r-1, j+1)$. These edges are all in $P_i$. But now, the horizontal edge $\{(r-1, b), (r-1, b+1)\}$ is in $P_i$ and is adjacent to the horizontal edge $\{(r-2, b), (r-2, b+1)\}$ in $C_i$. This is a contradiction. {The final figure is also depicted in \Cref{subfig:case_b2b}.}

\noindent\begin{wrapfigure}[14]{r}{0.3\textwidth}
\label{wrap-fig:4.10}
\centering
\begin{tikzpicture}[scale=0.5]
\foreach \x in {0,1,...,6} {
            \foreach \y in {-1, 0,...,7} {
                \smvertex (\x\y) at (\x, \y) {};
            }
        }
    
        
        \foreach \i in {1,3} {
            \foreach \j in {7,5,3} {
                \draw (\i,\j) -- (\i, \j-1) -- (\i-1,\j-1) -- (\i-1, \j-2) -- (\i, \j-2);
            }
            \draw (\i, 7) -- (\i+1, 7) -- (\i+1, 8);
        }
        
        \foreach \j in {6,4,2} {
            \draw[blue] (5, \j) -- (4, \j) -- (4, \j-1) -- (5, \j-1);
        }
        
        \draw (07) -- (0,8);
        \draw (10) -- (11);
        \draw (30) -- (31);
        \draw (0,-1) -- (00) -- (10);

        \foreach \i in {0,2,4} {
            \draw (\i,-2) -- (\i, -1);
        }
        
        
        \smvertex (x) at (1,-1) [color=green] {};
        \smvertex (x') at (0,0) [color=red] {};
\end{tikzpicture}
\end{wrapfigure} 
The same kind of argument works if $r$ is odd. If $r$ is odd, the resulting structure is depicted on the diagram to the right. To be precise, for all $l \in \{1, 3, 5, \ldots, r-4\}$, the vertices $(l, b+1), (l, b), (l+1, b), (l+1, b+1)$ are consecutive in $C_i$. (These are the blue paths of length 3 in the diagram.) Also, for all $l \in \{1,3,\ldots,r-4\}$ and $j\in\{1,3,\ldots,b-1\}$, the vertices  $(l-1, j), (l, j), (l, j-1), (l+1, j-1), (l+1, j), (l+2, j)$ are consecutive in $P_i$. (These form the zig-zagging subpaths of $P_i$ going downward.) Consider the vertex $(r-2, 0)$ (red in the diagram). Since it has degree 2, it must have an edge to the right to $(r-2, 1)$ and an edge down to $(r-1, 0)$. Now consider the vertex $(r-1, 1)$ (green in the diagram). It must have degree 2, but the only neighbor of the green vertex that does not have degree 2 yet is $(r-1, 2)$ to its right. Hence $(r-1, 1)$ can have degree at most 1, which is a contradiction.  {The final figure is also depicted in \Cref{subfig:case_b2c}.}$\square$

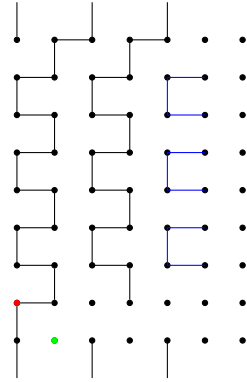
\begin{figure}[t!]
    \centering
    \begin{subfigure}[t]{0.3\textwidth}
        \centering
        \raisebox{6mm}{
        \begin{tikzpicture}[scale=0.7]
        \foreach \x in {0,1,...,6} {
            \foreach \y in {0,1,...,9} {
                \mvertex (\x\y) at (\x, \y) {};
            }
        }
    
        \foreach \i in {1,3} {
            \draw (\i,9) -- (\i, 8) -- (\i-1,8) -- (\i-1, 7) -- (\i, 7) -- (\i, 6);
        }
        \draw[blue] (48) -- (47);
        
        \end{tikzpicture}
        }
        \caption{The picture of $(P_i, C_i)$. The black edges are in $P_i$ and the blue edge is the leftmost, topmost vertical edge of $C_i$.}
        \label{subfig:case_b2a}
    \end{subfigure}%
    \hfill
    \begin{subfigure}[t]{0.3\textwidth}
        \centering
        \begin{tikzpicture}[scale=0.7]
        \foreach \x in {0,1,...,6} {
            \foreach \y in {0,1,...,9} {
                \mvertex (\x\y) at (\x, \y) {};
            }
        }
    
        \foreach \i in {1,3} {
            \foreach \j in {9, 7, 5, 3} {
                \draw (\i,\j) -- (\i, \j-1) -- (\i-1,\j-1) -- (\i-1, \j-2) -- (\i, \j-2);
            }
        }
        
        \foreach \j in {8, 6, 4, 2} {
            \draw[blue] (5, \j) -- (4, \j) -- (4, \j-1) -- (5, \j-1);
        }
        
        \draw (11) -- (10) -- (00) -- (0,-1);
        \draw (31) -- (30) -- (20) -- (2,-1);
        \draw (09) -- (0,10);
        \draw (19) -- (29) -- (2,10);
        \draw (39) -- (49) -- (4,10);
        \draw (4,-1) -- (40) -- (50);
        
        
        
        \end{tikzpicture}
        \caption{The forced connections if $r$ is even. Blue edges are in the cycle, black edges are in the path. At the bottom right, a blue horizontal edge is adjacent to a black horizontal edge, which is a contradiction.}
        \label{subfig:case_b2b}
    \end{subfigure}
    \hfill
    \begin{subfigure}[t]{0.3\textwidth}
        \centering
        \begin{tikzpicture}[scale=0.7]
        \foreach \x in {0,1,...,6} {
            \foreach \y in {-1,0,1,...,9} {
                \mvertex (\x\y) at (\x, \y) {};
            }
        }
    
        \foreach \i in {1,3} {
            \foreach \j in {9, 7, 5, 3} {
                \draw (\i,\j) -- (\i, \j-1) -- (\i-1,\j-1) -- (\i-1, \j-2) -- (\i, \j-2);
            }
        }
        
        \foreach \j in {8, 6, 4, 2} {
            \draw[blue] (5, \j) -- (4, \j) -- (4, \j-1) -- (5, \j-1);
        }
        
        \draw (11) -- (10) -- (00) -- (0,-1) -- (0, -2);
        \draw (2, -2) -- (2, -1);
        \draw (31) -- (30);
        \draw (09) -- (0,10);
        \draw (19) -- (29) -- (2,10);
        \draw (39) -- (49) -- (4,10);
        
        
        \vertex (x) at (1, -1) [color=green] {};
        
        \end{tikzpicture}
        \caption{The forced connections if $r$ is odd. Blue edges are in the cycle, black edges are in the path. At the bottom left, the green vertex must have degree 2, but only has one available neighbor to connect to.}
        \label{subfig:case_b2c}
    \end{subfigure}
    
        
        
    \caption{An illustration of the proof of Claim \ref{cm:caseb2}}
    \label{fig:case_b2}
\end{figure}
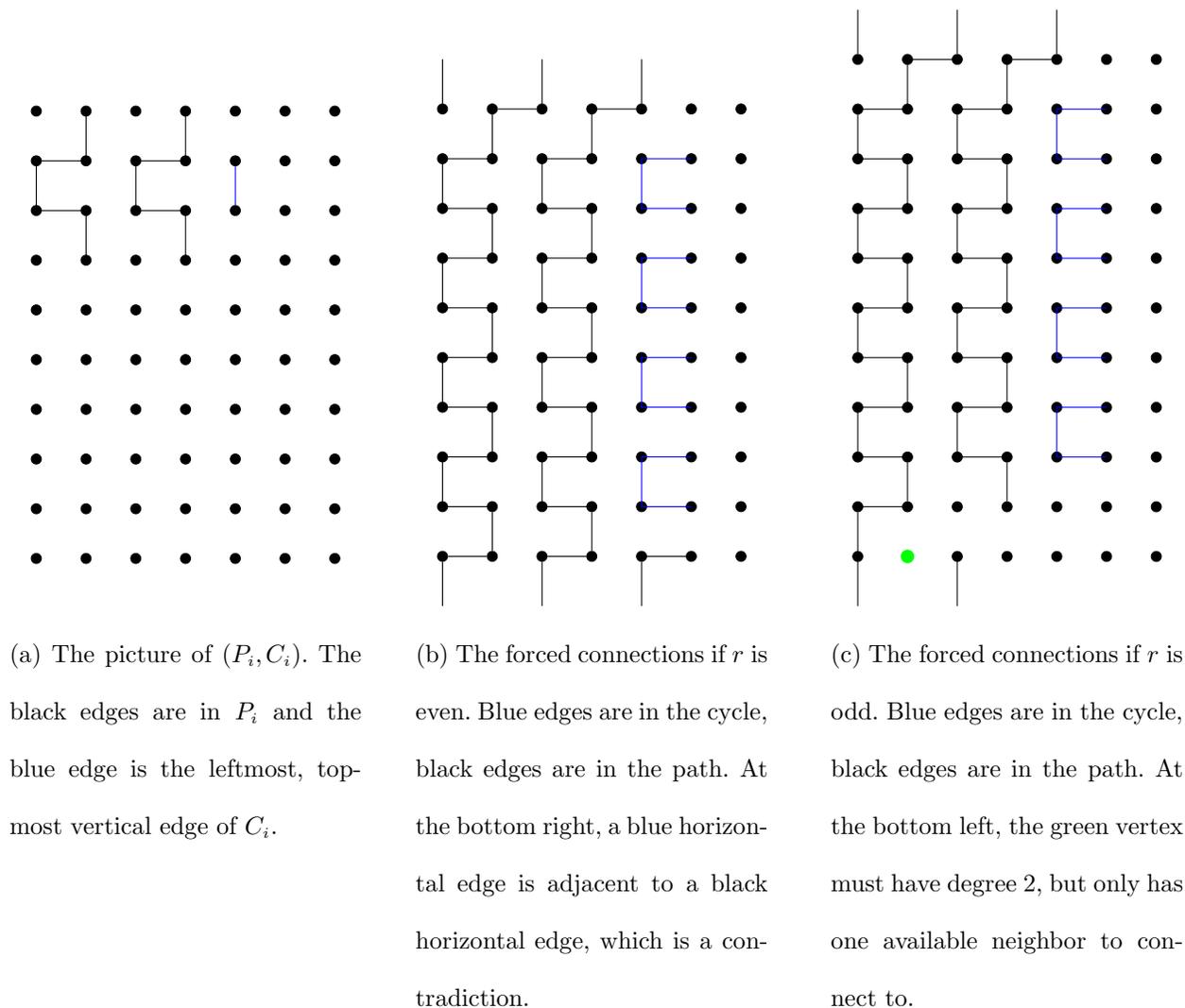

\subsection{Deriving the final contradiction} \label{subsec:contradiction}
Together, Claims \ref{cm:cyclecase1} and \ref{cm:caseb2} show that the only remaining case is for the one where all the cycles $C_1, \ldots, C_k$ are 2-cycles. We complete the proof of Theorem \ref{thm:main} by showing that this case is not possible in a minimal counterexample. 
\label{subsec:cyclecase3}
\begin{cm}
\label{cm:caseb3}
In a minimal counterexample, it is impossible for all the cycles $C_1, \ldots, C_k$ to be 2-cycles. 
\end{cm}
\begin{proof}
If we are in the case where cycles $C_1,\ldots,C_k$ are all 2-cycles, the picture must look like Figure \ref{fig:case_b3}.
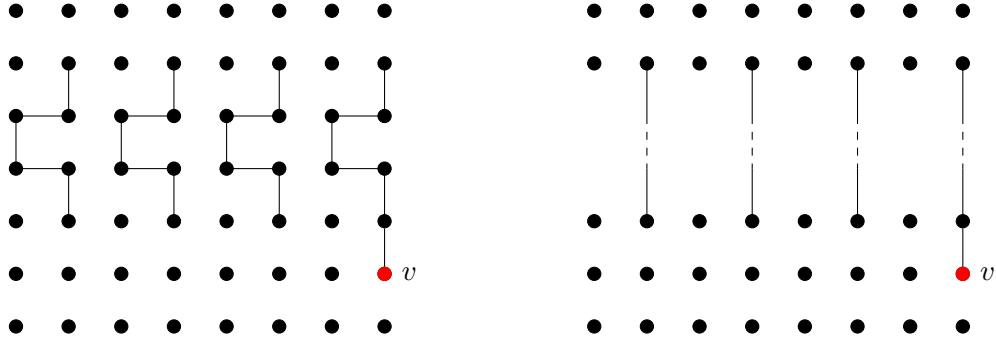
\begin{figure}
    \centering
    \begin{subfigure}[t]{0.45\textwidth}
        \centering
        \begin{tikzpicture}[scale=0.7]
    \foreach \x in {0,1,...,7} {
            \foreach \y in {0,1,...,6} {
                \vertex (\x\y) at (\x, \y) {};
            }
        }
    
        \foreach \i in {1,3,5,7} {
            \draw (\i,5) -- (\i, 4) -- (\i-1,4) -- (\i-1, 3) -- (\i, 3) -- (\i, 2);
        }
        \draw (72)--(71);
        \vertex[red] (71') at (7,1) [label=right:$v$] {};
    \end{tikzpicture}
        \caption{The picture of $P^*$ in the case where every cycle is a 2-cycle.}
        \label{subfig:case_b3a}
    \end{subfigure}%
    ~ 
    \begin{subfigure}[t]{0.45\textwidth}
        \centering
        \begin{tikzpicture}[scale=0.7]
    \foreach \x in {0,1,...,7} {
            \foreach \y in {0,1,2,5,6} {
                \vertex (\x\y) at (\x, \y) {};
            }
        }
    
        \foreach \i in {1,3,5,7} {
            \draw (\i,5) -- (\i, 4);
            \draw[dashed] (\i, 4) -- (\i, 3);
            \draw (\i, 3) -- (\i, 2);
        }
        \draw (72)--(71);
        \vertex[red] (71') at (7,1) [label=right:$v$] {};
    \end{tikzpicture}
        \caption{Reconnecting to a Hamiltonian path in the graph after deleting the two rows.}
        \label{subfig:case_b3b}
    \end{subfigure}
    \caption{An illustration of the proof of Claim \ref{cm:caseb3}.}
    \label{fig:case_b3}
\end{figure}
Note that, in particular, there are an even number of columns. 

To complete the proof, we will show that there is a GG path $T$ from 0 to $v$ with cost at most the cost of $P^*$. To do this, we will use the minimal counterexample assumption to find a GG path $T'$ in the cylinder graph with two rows deleted. We will then show that we can obtain our desired path $T$ from $T'$.

Suppose rows $i$ and $i+1$ are the two rows that contain the 2-cycles.
Consider the graph $G'$ obtained by deleting those two rows. Observe that we can turn $P^*$ into a Hamiltonian path $P'$ in $G'$ by doing the following:
\begin{itemize}
    \item Each time path $P$ visits the two deleted rows, it must be via a subpath of the form $x, x+a_1, x+a_1-a_2, x+2a_1-a_2, x+2a_1, x+3a_1$ (possibly in the reverse order.) \item To obtain $P'$ from $P^*$, replace all occurrences of the above subpath with the single vertical edge $\{x, x+3a_1\}$ (see Figure \ref{subfig:case_b3b}).
\end{itemize}
Recall that we defined $v := (v_1, c-1)$ to be the endpoint of $P^*$ in the last column of $G$. Similarly, define $v' := (v_1', c-1)$ to be the endpoint of $P'$ in the last column of $G'$. Then $v_1'$ is either $v_1$ or $v_1-2$: It is $v_1$ if $v_1 < i$, and it is $v_1-2$ if $v_1 \geq i$. 

Moreover, $P'$ uses exactly $2g_1$ fewer horizontal edges than $P$. By minimality of our counterexample, we know the cheapest 0 -- $v'$ Hamiltonian path in $G'$ is attained by a GG path. Hence, in $G'$, there exists a GG path $T'$ from 0 to $v'$ with cost at most the cost of $P'$. 

Let $T$ be the cheapest $0$-$v$ GG path in $G$. For a path $S$, let $c(S)$ denote its cost (which in our case is the number of horizontal edges it uses). Then, by property 5 of Proposition \ref{prop:AG}, we have that $c(T) \leq c(T') + 2$. Since $c(T') \leq c(P') = c(P) - 2g_1$, (and $g_1 \geq 1)$, it follows that $c(T) \leq c(P)$, as claimed.
\end{proof}


\section{Conclusion}
\label{sec:conclusion}
Circulant edge costs present an intriguing special case of the TSP.  Despite the substantial structure and symmetry, and despite being frequently stated as an open question, remarkably little has been known about the complexity of circulant TSP.  By giving  a polynomial-time algorithm for the two-stripe symmetric circulant TSP, we provide the first non-trivial complexity result. 

The natural next question is to consider general symmetric circulant TSP, where all the edge lengths can potentially have finite cost. For that problem, a 2-approximation algorithm is known, but it is open whether it is polynomial-time solvable. As an intermediate step, one might consider a variant with some constant number of edge lengths having finite costs. One might wonder, e.g., if there is a dynamic programming approach that extends work from this paper to the constant-stripe case.

We also note that there are many open polyhedral questions for circulant TSP.  Circulant TSP is one of the few non-trivial settings where the integrality gap of the subtour LP is exactly known (its integrality gap is 2; see Gutekunst and Williamson \cite{Gut19b}).  Gutekunst and Williamson \cite{gut20} prove a facet-defining inequality using circulant symmetry, motivated by a class of two-stripe instances where the subtour LP's integrality gap is arbitrarily close to 2.  One might wonder is there is an exact LP formulation for circulant TSP (in general or with a constant number of stripes), or as an intermediate step, a strengthening of the subtour LP.

\section*{Acknowledgements}

This material is based upon work supported by the National Science Foundation Graduate
Research Fellowship Program under Grant No. DGE-1650441,  by National Science Foundation Grants No. CCF-1908517, CCF-2007009, and CCF-2153331, and by Natural Sciences and Engineering Research Council of Canada fellowship PGSD3-532673-2019.  

\begin{APPENDICES}

\section{Proof of Proposition \ref{prop:ggall}}
\label{app:ggall}

\ggall*

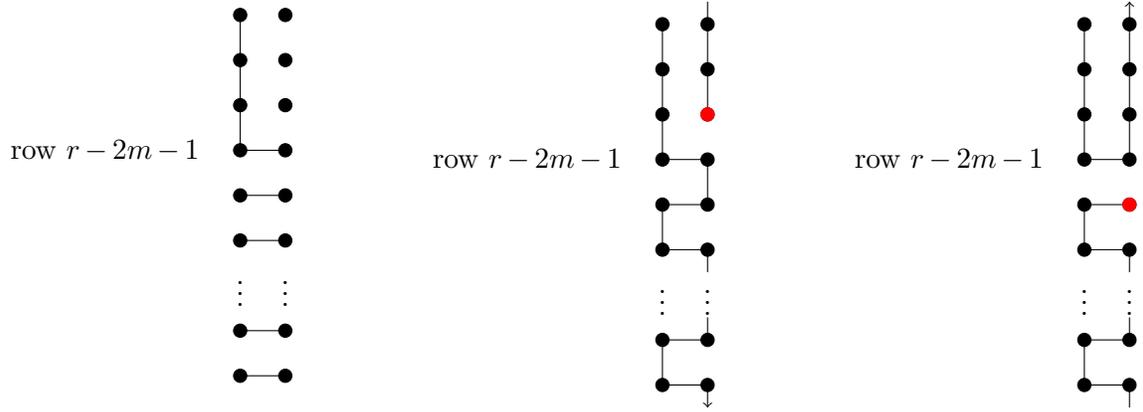
\begin{figure}[t!]

    \centering
        \begin{subfigure}[t]{0.3\textwidth}
        \centering
        \begin{tikzpicture}[scale=0.6]
        \foreach \x in {0,1} {
            \foreach \y in {0,1} {
                \vertex (\x\y) at (\x, \y) {};
            }
        }  
        
        \foreach \x in {0,1} {
            \foreach \y in {3, 4, ..., 8} {
                \vertex (\x\y) at (\x, \y) {};
            }
        }  
        \foreach \i in {0} {
            \foreach \j in {5, 6, 7} {
            \draw (\i, \j) to (\i, \j+1);}
        }
        
        \foreach \j in {0, 1, 3, 4, 5} {
        \draw (0, \j) to (1, \j);}
        
        \node (d1) at (0, 2) {$\vdots$};
        \node (d2) at (1, 2) {$\vdots$};
        \node (l1) at (-3, 5) {row $r-2m-1$};
        \node (blabk) at (0, -0.5) {};
        \node (blank) at (0, 8.5) {};
        \end{tikzpicture}
        \end{subfigure}
        \hspace{4mm}
                \begin{subfigure}[t]{0.3\textwidth}
        \centering
        \begin{tikzpicture}[scale=0.6]
        \foreach \x in {0,1} {
            \foreach \y in {0,1} {
                \vertex (\x\y) at (\x, \y) {};
            }
        }  
        
        \foreach \x in {0,1} {
            \foreach \y in {3, 4, ..., 8} {
                \vertex (\x\y) at (\x, \y) {};
            }
        }  
        \foreach \i in {0} {
            \foreach \j in {5, 6, 7} {
            \draw (\i, \j) to (\i, \j+1);}
        }
        
        \foreach \j in {0, 1, 3, 4, 5} {
        \draw (0, \j) to (1, \j);}
        
        \node (d1) at (0, 2) {$\vdots$};
        \node (d2) at (1, 2) {$\vdots$};
        \node (l1) at (-3, 5) {row $r-2m-1$};
        
        \draw (1, 5) to (1, 4);
        \draw (0, 4) to (0, 3);
        \draw (1, 3) to (1, 2.5);
        \draw (1, 1.5) to (1, 1);
        \draw (0, 1) to (0, 0);
        \draw[->] (1, 0) to (1, -0.5);
        \draw (1, 8.5) to (1, 8);
        \draw (1, 8) to (1, 7);
        \draw (1, 7) to (1, 6);
        \vertex[red] (r2) at (1, 6) {};
        
        \end{tikzpicture}
        \end{subfigure}
                \hspace{4mm}
                \begin{subfigure}[t]{0.3\textwidth}
        \centering
        \begin{tikzpicture}[scale=0.6]
        \foreach \x in {0,1} {
            \foreach \y in {0,1} {
                \vertex (\x\y) at (\x, \y) {};
            }
        }  
        
        \foreach \x in {0,1} {
            \foreach \y in {3, 4, ..., 8} {
                \vertex (\x\y) at (\x, \y) {};
            }
        }  
        \foreach \i in {0} {
            \foreach \j in {5, 6, 7} {
            \draw (\i, \j) to (\i, \j+1);}
        }
        
        \foreach \j in {0, 1, 3, 4, 5} {
        \draw (0, \j) to (1, \j);}
        
        \node (d1) at (0, 2) {$\vdots$};
        \node (d2) at (1, 2) {$\vdots$};
        \node (l1) at (-3, 5) {row $r-2m-1$};

        \draw (0, 4) to (0, 3);
        \draw (1, 3) to (1, 2.5);
        \draw (1, 1.5) to (1, 1);
        \draw (0, 1) to (0, 0);
        \draw (1, 0) to (1, -0.5);
        \draw[->] (1, 8) to (1, 8.5);
        \draw (1, 8) to (1, 7);
        \draw (1, 7) to (1, 6);
        \draw (1,5 ) to (1, 6);
        \vertex[red] (r2) at (1, 4) {};
        
        \end{tikzpicture}
        \end{subfigure}
    \caption{Argument that any Hamiltonian Path $P$ in Proposition \ref{prop:ggall} starting with $r-2m-1$ vertical edges is either a GG path or can be replaced by a GG path with fewer edges.  The leftmost image is for a generic path, including the vertical edges and forced horizontal edges between the first two columns. The middle and right image show the two possible ways to complete $P$'s path through the first two columns.}
    \label{fig:ForcedGG}
\end{figure}

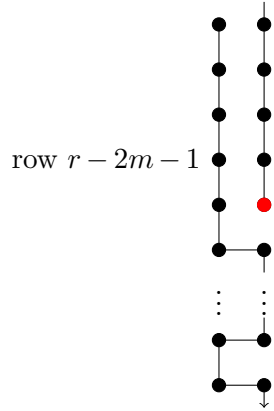
\begin{figure}[th!]

    \centering
        \begin{subfigure}[t]{0.3\textwidth}
        \centering
        \begin{tikzpicture}[scale=0.6]
        \foreach \x in {0,1} {
            \foreach \y in {0,1} {
                \vertex (\x\y) at (\x, \y) {};
            }
        }  
        
        \foreach \x in {0,1} {
            \foreach \y in {3, 4, ..., 8} {
                \vertex (\x\y) at (\x, \y) {};
            }
        }  
        \foreach \i in {0} {
            \foreach \j in {5, 6, 7} {
            \draw (\i, \j) to (\i, \j+1);}
        }
        
        \foreach \j in {0, 1, 3} {
        \draw (0, \j) to (1, \j);}
        
        \node (d1) at (0, 2) {$\vdots$};
        \node (d2) at (1, 2) {$\vdots$};
        \node (l1) at (-2.5, 5) {row $r-2m-1$};
        
        \draw (1, 5) to (1, 4);
        \draw (0, 4) to (0, 3);
        \draw (1, 3) to (1, 2.5);
        \draw (1, 1.5) to (1, 1);
        \draw (0, 1) to (0, 0);
        \draw[->] (1, 0) to (1, -0.5);
        \draw (1, 8.5) to (1, 8);
        \draw (1, 8) to (1, 7);
        \draw (1, 7) to (1, 6);
        \draw (1, 6) to (1, 5);
        \draw (0, 5) to (0, 4);
        \vertex[red] (r2) at (1, 4) {};
        
        \end{tikzpicture}
        \end{subfigure}

    \caption{Replacing the rightmost path in Figure \ref{fig:ForcedGG} with a GG path using two fewer horizontal edges.}
    \label{fig:ForcedGG2}
\end{figure}

\begin{proof}
If $P$ is a GG path, $x \in A_{r,c, m}$ definitionally. Because $P$ uses  1 horizontal edge between every  pair of columns after the first, it is only within the first two columns that it could have a non-GG path structure.   

\begin{cm}
If $m=0$, $P$ is a GG path.
\end{cm}

\begin{proof}
First consider the case when $m=0$, the path $P$ must visit every vertex in the first column before moving to the second column.  Since $P$ starts at $(0, 0)$, it either starts going vertically down (from $(0, 0)$ to $(1, 0)$) or wrapping up (from $(0, 0)$ to $(r-1, 0)$). In either case, it must be one of the tour types shown in Figure \ref{fig:GG1}, i.e., a GG path.  
\end{proof}

If $m\geq 1,$ there is at least one extra pair of horizontal edges between the first two columns.  We proceed by considering the first horizontal edge. 

\begin{cm}
If $m\geq 1$ and the first $r-2m-1$ edges of $P$ are vertical,  $x \in A_{r, c, m}$.
\end{cm}

\begin{proof}
Suppose that the first $r-2m-1$ edges of $P$ are vertical.  By cylindrical symmetry, it is without loss of generality to assume that $P$ starts by following $(0, 0), (1, 0), ..., (r-2m-1, 0).$  To include $2m+1$ horizontal edges between the first two columns,  we must have horizontal edges between $(0, i)$ and $(1, i)$ for $r-2m-1 \leq i \leq r-1.$  See Figure \ref{fig:ForcedGG}.  Note that there are two possible ways to complete $P$'s path through the first two columns.  If, from vertex $(1, r-2m-1),$ the path continues going down to $(1, r-2m)$ (in the same direction as the original edges $(0, 0), (1, 0), ..., (r-2m-1, 0)$), then $P$ is a GG path (the middle image of Figure \ref{fig:ForcedGG}).  Otherwise the path goes up from $(r-2m-1, 1)$ to $(r-2m-2, 1)$, and must be completed as in the rightmost image of Figure \ref{fig:ForcedGG}.  The first two columns of such a path can be replaced, as in Figure \ref{fig:ForcedGG2}, covering $P$ into a GG path, using two fewer horizontal edges, so that $(x, c-1) \in A_{r, c, m-1} \subset A_{r, c, m}.$  Formally, we delete the horizontal edges from $(r-2m-1, 0)$ to $(r-2m-1, 1)$ and from $(r-2m, 0)$ to $(r-2m, 1)$, and replace them with vertical edges from $(r-2m-1, 1)$ to $(r-2m, 1)$ and from $(r-2m-1, 1)$ to $(r-2m, 1).$
\end{proof}

At this point, what remains are paths $P$ that start at $(0, 0)$ and where at least one of the first $r-2m-1$ edges is horizontal.  We consider what happens if the first edge is vertical.  As above, it is without loss of generality to assume that the first edge is from $(0, 0)$ to $(1, 0).$

Suppose that the first horizontal edge is from $(i, 0)$ to $(i, 1)$ with $i\geq 1,$ so that $P$ contains the sequence $(i-1, 0), (i, 0), (i, 1).$  Since $P$ eventually returns to the first column (and uses one edge between the second and third columns), the next edge in this sequence is either $(i+1, 1)$ or $(i-1, 1)$.  

\begin{cm}\label{cm:PisGG}
Suppose that $P$ contains the sequence $(i-1, 0), (i, 0), (i, 1), (i+1, 1).$  Then $P$ must be a GG path and so $x \in A_{r, c, m}$.
\end{cm}

\begin{proof}
First, the sequence must continue  $(i-1, 0), (i, 0), (i, 1), (i+1, 1), (i+1, 0),$ as otherwise the vertex $(i+1, 0)$ cannot possibly have degree two (see red vertex in the left picture in Figure \ref{fig:PisGG}).  Analogously, $P$ must continue alternating and proceed as $(i-1, 0), (i, 0), (i,1), (i+1, 1), (i+1, 0), (i+2, 0).$  Otherwise the nodes $(i-1, 1)$ and $(i+2, 1)$ cannot both have degree two (see the red nodes in Figure \ref{fig:PisGG}); note that $(i-1, 0)$ is fully saturated: either $i=1$ and $(0, 0)$ has degree 1 in a Hamiltonian path, or $i>1$ and our assumption that $(i, 0)$ to $(i, 1)$ is the first horizontal edge in $P$ implies that $(i-1, 0)$ is also adjacent to $(i-2, 0)$ and has so $(i-1, 0)$ has degree 2.

These arguments continue iteratively, showing that $P$ must alternate horizontal and vertical edges, proceeding until all $2m+1$ horizontal edges have been used.  Since $P$ is Hamiltonian and visits every vertex in the first column, we must have that $P$ is a GG path of the form where the first edge goes from $(0, 0)$ to $(1, 0)$.
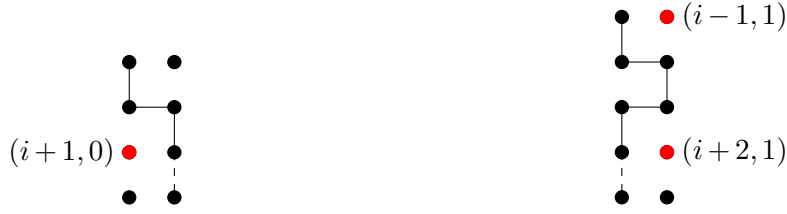
\begin{figure}[t!]
    \centering
        \begin{subfigure}[t]{0.45\textwidth}
        \centering
        \begin{tikzpicture}[scale=0.6]
        \foreach \x in {0,1} {
            \foreach \y in {0, 1, ..., 3} {
                \vertex (\x\y) at (\x, \y) {};
            }
        }  
        
        \draw (0, 3) to (0, 2);
        \draw (0, 2) to (1, 2);
        \draw (1, 2) to (1, 1);
        \draw[dashed] (1, 1) to (1, 0);
        \node (l1) at (-1.5, 1) {$(i+1, 0)$};
       
        \vertex[red] (r2) at (0, 1) {};
        \end{tikzpicture}
        \end{subfigure}
        \hspace{5mm}
         \begin{subfigure}[t]{0.45\textwidth}
        \centering
         \begin{tikzpicture}[scale=0.6]
        
        \foreach \x in {0,1} {
            \foreach \y in {-1, 0, ..., 3} {
                \vertex (\x\y) at (\x, \y) {};
            }
        }  
        
        \draw (0, 3) to (0, 2);
        \draw (0, 2) to (1, 2);
        \draw (1, 2) to (1, 1);
        \draw (1, 1) to (0, 1);
        \draw (0, 1) to (0, 0);
        \draw[dashed] (0, 0) to (0, -1);
        \node (l1) at (2.5, 0) {$(i+2, 1)$};
        \node (l2) at (2.5, 3) {$(i-1, 1)$};
       
        \vertex[red] (r2) at (1, 0) {};
        \vertex[red] (r2) at (1, 3) {};
        \end{tikzpicture}
        \end{subfigure}

    \caption{Illustration of the proof of Claim \ref{cm:PisGG}.}
    \label{fig:PisGG}
\end{figure}

\end{proof}

\begin{cm}
\label{cm:PisGG2}
Suppose that $P$ contains the sequence $(i-1, 0), (i, 0), (i, 1), (i-1, 1).$  Then  $x \in A_{r, c, m}$.
\end{cm}

\begin{proof}
This case proceeds largely as in Claim \ref{cm:PisGG}.  First we note that $P$ must continue as $(0, 0), (1, 0), ..., (i-1, 0),$ $(i, 0),(i, 1), (i-1, 1), (i-2, 1), ..., (1, 0), (r-1, 1)$ (see left picture of Figure \ref{fig:PisGG2}).

\begin{figure}[t!]
    \centering
        \begin{subfigure}[t]{0.45\textwidth}
        \centering
        \begin{tikzpicture}[scale=0.6]
        \foreach \x in {0,1} {
            \foreach \y in {0,1, 2, 4, 5, 6, 8, 9} {
                \vertex (\x\y) at (\x, \y) {};
            }
        }  
        
        \foreach \x in {0,1} {
            \foreach \y in {3, 7} {
                \node (\x\y) at (\x, \y) {$\vdots$};
            }
        }

        \draw (0, 9) to (0, 8);
        \draw (0, 6) to (0, 5);
        \draw (0, 7.5) to (0, 8);
        \draw (0, 6) to (0, 6.5);

        \draw (1, 9) to (1, 8);
        \draw (1, 6) to (1, 5);
        \draw (1, 7.5) to (1, 8);
        \draw (1, 6) to (1, 6.5);
        
        \draw (0, 5) to (1, 5);
        
        \draw[->] (1, 9) to (1, 9.5);
        \draw (1, -0.5) to (1, 0);
        
        \node (l1) at (-1.2, 9) {$(0, 0)$};
        \node (l2) at (-1.2, 5) {$(i, 0)$};
        \node (l3) at (-1.5, 0) {$(r-1, 0)$};
       
        \end{tikzpicture}
        \end{subfigure}
        \hspace{5mm}
         \begin{subfigure}[t]{0.45\textwidth}
        \centering
         \begin{tikzpicture}[scale=0.6]
        \foreach \x in {0,1} {
            \foreach \y in {0,1, 2, 4, 5, 6, 8, 9} {
                \vertex (\x\y) at (\x, \y) {};
            }
        }  
        
        \foreach \x in {0,1} {
            \foreach \y in {3, 7} {
                \node (\x\y) at (\x, \y) {$\vdots$};
            }
        }

        \draw (0, 9) to (0, 8);
        \draw (0, 6) to (0, 5);
        \draw (0, 7.5) to (0, 8);
        \draw (0, 6) to (0, 6.5);

        \draw (1, 9) to (1, 8);
        \draw (1, 6) to (1, 5);
        \draw (1, 7.5) to (1, 8);
        \draw (1, 6) to (1, 6.5);
        
        \draw (0, 5) to (1, 5);
        
        \draw[->] (1, 9) to (1, 9.5);
        \draw (1, -0.5) to (1, 0);
        
        \node (l1) at (-1.2, 9) {$(0, 0)$};
        \node (l2) at (-1.2, 5) {$(i,0)$};
        \node (l3) at (-1.5, 0) {$(r-1,0)$};
        
        \draw (1, 0) to (0, 0);
        \draw (0, 0) to (0, 1);
        \draw (0, 1) to (1, 1);
        \draw (1, 1) to (1, 2);
        \draw (1, 2) to (0, 2);
        
        \draw(0, 2) to (0, 2.5);
        \draw (0, 3.5) to (0, 4);
        \draw (0, 4) to (1, 4);
        
        \draw[->] (1, 4) to (1.5, 4);
       
        \end{tikzpicture}
        \end{subfigure}

    \caption{Illustration of the proof of Claim \ref{cm:PisGG2}.}
    \label{fig:PisGG2}
\end{figure}
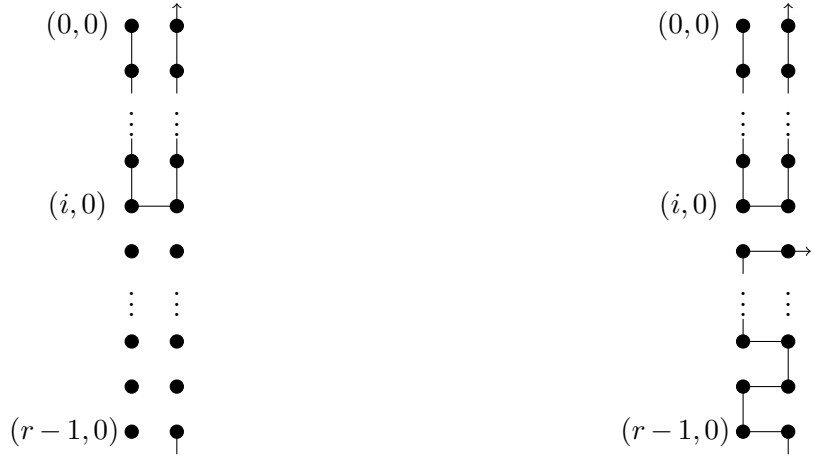

We can see that $P$ must continue from $(r-1,1)$ to $(r-1,0)$ (so that $(r-1,0)$ has degree 2) and then to $(r-2,0).$  As in Claim \ref{cm:PisGG}, at $(r-2, 0),$ $P$ must continue to $(r-2,1)$, and continue alternating between the first and second column.  Ultimately $P$ must proceed as illustrated in the right picture of Figure \ref{fig:PisGG2}.  Since $P$ uses exactly $2m+1$ horizontal edges between the first two columns, the horizontal edges must be in in the last $2m+1$ rows: $r-1, r-2, r-3, ..., r-2m-1.$  $P$ leaves the second column at the second horizontal edge, at row $r-2m$.  Note that we can replace the part of $P$ in the first two columns using a GG path with two fewer edges: delete the edges from $(i,0)$ to $(i,1)$ and from $( i+1,0)$ to $(i+1,1)$, and replace them with edges from $(i,0)$ to $(i+1,0)$ and from $( i,1)$ to $( i+1,1).$
\end{proof}

The only remaining case we must consider is if the first edge of $P$ is horizontal.  Again appealing to cylindrical symmetry, we assume that $P$ starts $(0, 0), (0,1), (1, 1)$.

\begin{cm}
\label{cm:PisGG3}
Suppose that $P$ contains the sequence $(0, 0), (0,1), (1, 1)$.  Then  $x \in A_{r, c, m}$.
\end{cm}

\begin{proof}
As in our previous claims, we trace out what must happen with $P$.  We note that $P$ must continue from $(1, 1)$ to $(1, 0)$ to $(2, 0)$ (as otherwise $(1,0)$ cannot have degree 2).  Similarly, if $P$ proceeds from $(2,0)$ to $(2,1)$, it must continue $(3,1)$ to $(3,0)$ to $(4,0).$  More generally, $P$ will alternate between the first and second column until either 1) $P$ visits every vertex in the first column or 2) $P$ contains at least two sequential vertical edges in the first column.  The former case can only happen if $r=2m+1$, in which case $P$ is definitionally a GG path, and so we restrict our attention to the latter case.

In this case, $P$ alternates between the first and second column until it reaches some sequence $(i-1,0), (i,0), (i+1,0)$ as shown below.  Consider the vertex $(i,1)$, indicated in red in the left picture of Figure \ref{fig:PisGG3}.  This vertex must have degree two, and so $P$ must contain the sequence $( i-1,1), (i,1), (i,2)$, indicated with dashed lines in the figure.

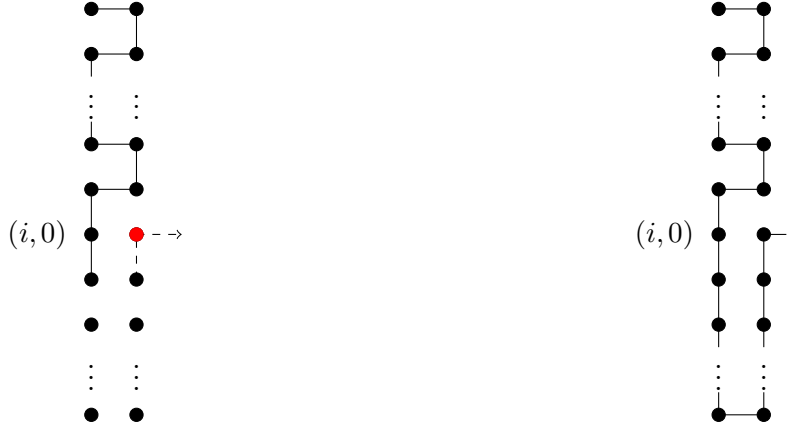
\begin{figure}[t!]
    \centering
        \begin{subfigure}[t]{0.45\textwidth}
        \centering
        \begin{tikzpicture}[scale=0.6]
        \foreach \x in {0,1} {
            \foreach \y in {0, 2, 3, 4, 5, 6, 8, 9} {
                \vertex (\x\y) at (\x, \y) {};
            }
        }  
        
        \foreach \x in {0,1} {
            \foreach \y in {1, 7} {
                \node (\x\y) at (\x, \y) {$\vdots$};
            }
        }  
        
        \draw (0, 9) to (1, 9) to (1, 8) to (0, 8) to (0, 7.5);
        \draw (0, 6.5) to (0, 6) to (1, 6) to (1, 5) to (0, 5) to (0, 4) to (0, 3);
        
        \draw[dashed] (1, 3) to (1, 4);
        \draw[dashed, ->] (1, 4) to (2, 4);

        \node (l2) at (-1.2, 4) {$(i,0)$};

        \vertex[red] (r2) at (1, 4) {};
        \end{tikzpicture}
        \end{subfigure}
        \hspace{5mm}
         \begin{subfigure}[t]{0.45\textwidth}
        \centering
         \begin{tikzpicture}[scale=0.6]
        \foreach \x in {0,1} {
            \foreach \y in {0, 2, 3, 4, 5, 6, 8, 9} {
                \vertex (\x\y) at (\x, \y) {};
            }
        }  
        
        \foreach \x in {0,1} {
            \foreach \y in {1, 7} {
                \node (\x\y) at (\x, \y) {$\vdots$};
            }
        }  
        
        \draw (0, 9) to (1, 9) to (1, 8) to (0, 8) to (0, 7.5);
        \draw (0, 6.5) to (0, 6) to (1, 6) to (1, 5) to (0, 5) to (0, 4) to (0, 3) to (0, 2) to (0, 1.5);
        \draw (0, 0.5) to (0, 0) to (1, 0) to (1, 0.5);
        \draw (1, 1.5) to (1, 2) to (1, 3);
        \draw (1, 3) to (1, 4);
        \draw (1, 4) to (1.5, 4);
        
        \node (l2) at (-1.2, 4) {$(i,0)$};

        \end{tikzpicture}
        \end{subfigure}

    \caption{Illustration of the proof of Claim \ref{cm:PisGG3}.}
    \label{fig:PisGG3}
\end{figure}
Provided $i+1 \neq r-1,$ there are more vertices to visit in the first and second column.  Hence $P$ cannot contain the edge from $(i-1,0)$ to $(i-1,1).$  Since both $(i-1,0)$ and $(i-1,1)$ must have degree 2, $P$ must contain edges from $(i-1,0)$ to $(i-2,0)$ and from $(i-1,1)$ to $(i-2,1)$.  Continuing this argument, we see that $P$ must contain sequential pairs of vertical edges in the first and second column, so that $P$ must proceed $(i,0), (i-1,0), (i-2,0), ..., (r-1,0), ( r-1,1), ( r-2,1), ( r-3,1), ..., ( i+1,1), ( i,1)$ as indicated in the right picture of Figure \ref{fig:PisGG3}.

Note that we can replace the part of $P$ in the first two columns using a GG path with two fewer edges: delete the edges from $(0, 0)$ to $(0,1)$ and from $(r-1,0)$ to $( r-1,1)$, and replace them with edges from $(0, 0)$ to $( r-1,0)$  and from $(0,1)$ to $(r-1,1).$
\end{proof}

The above claims complete our proof.
\end{proof}

\section{Proof of Proposition \ref{prop:AG}}
\label{app:AG}

Recall Proposition \ref{prop:AG}: 

\vspace{3mm}

\propAG*

    \begin{figure}[t!]

    \centering
        \begin{subfigure}[t]{0.3\textwidth}
        \centering
        \begin{tikzpicture}[scale=0.6]
        \foreach \x in {0,1} {
            \foreach \y in {0,1, 2, ..., 7} {
                \vertex (\x\y) at (\x, \y) {};
            }
        }

        \foreach \j in {0, 1, ..., 6} {
        \draw (0, \j) to (1, \j);}
    
        \foreach \j in {0, 2, 4, 6}{ \draw (0, \j) to (0, \j+1);}
        \foreach \j in {1, 3, 5}{ \draw (1, \j) to (1, \j+1);}
        \draw[->] (1, 0) to (1, -0.5);
        \draw (1, 7.5) to (1, 7);
        \vertex[red] (r1) at (1, 7) {};
        \end{tikzpicture}
        \end{subfigure}
        \hspace{4mm}
                \begin{subfigure}[t]{0.3\textwidth}
        \centering
        \begin{tikzpicture}[scale=0.6]
        \foreach \x in {0,1} {
            \foreach \y in {0,1, 2, ..., 7} {
                \vertex (\x\y) at (\x, \y) {};
            }
        }

        \foreach \j in {0, 1, ..., 4} {
        \draw (0, \j) to (1, \j);}
    
        \foreach \j in {0, 2, 4, 6, 5}{ \draw (0, \j) to (0, \j+1);}
        \foreach \j in {1, 3, 5, 6}{ \draw (1, \j) to (1, \j+1);}
        \draw[->] (1, 0) to (1, -0.5);
        \draw (1, 7.5) to (1, 7);
        \vertex[red] (r1) at (1, 5) {};
        \end{tikzpicture}
        \end{subfigure}
        \hspace{4mm}
                        \begin{subfigure}[t]{0.3\textwidth}
        \centering
        \begin{tikzpicture}[scale=0.6]
        \foreach \x in {0,1} {
            \foreach \y in {0,1, 2, ..., 7} {
                \vertex (\x\y) at (\x, \y) {};
            }
        }

        \foreach \j in {0} {
        \draw (0, \j) to (1, \j);}
    
        \foreach \j in {0, 1, 2, 3, 4, 5, 6}{ \draw (0, \j) to (0, \j+1);}
        \foreach \j in {1, 2, 3, 4, 5, 6}{ \draw (1, \j) to (1, \j+1);}
        \draw[->] (1, 0) to (1, -0.5);
        \draw (1, 7.5) to (1, 7);
        \vertex[red] (r1) at (1, 1) {};
        \end{tikzpicture}
        \end{subfigure}
                
    \caption{Argument for the third statement of Proposition \ref{prop:AG} when $r-(2k+1)=1.$  If $r-(2k+1)=1,$ then GG paths using $k-1$ and $0$ extra horizontal pairs respectively reach the vertex two below and two above the end of the GG path.}
    \label{fig:1left}
\end{figure}
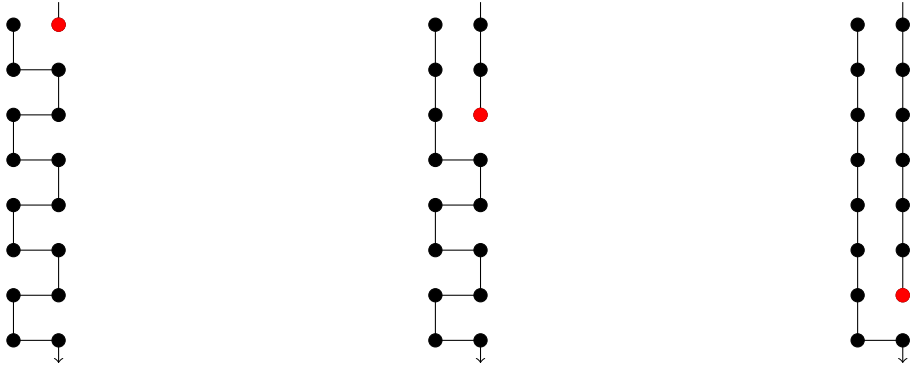

        \begin{figure}[t!]

    \centering
        \begin{subfigure}[t]{0.3\textwidth}
        \centering
        \begin{tikzpicture}[scale=0.6]
        \foreach \x in {0,1} {
            \foreach \y in {0,1, 2, ..., 8} {
                \vertex (\x\y) at (\x, \y) {};
            }
        }

        \foreach \j in {0, 1, ..., 8} {
        \draw (0, \j) to (1, \j);}
    
        \foreach \j in {0, 2, 4, 6}{ \draw (0, \j) to (0, \j+1);}
        \foreach \j in {1, 3, 5, 7}{ \draw (1, \j) to (1, \j+1);}
        \vertex[red] (r1) at (1, 0) {};
        \end{tikzpicture}
        \end{subfigure}
        \hspace{4mm}
                \begin{subfigure}[t]{0.3\textwidth}
        \centering
        \begin{tikzpicture}[scale=0.6]
        \foreach \x in {0,1} {
            \foreach \y in {0,1, 2, ..., 8} {
                \vertex (\x\y) at (\x, \y) {};
            }
        }

        \foreach \j in {0, 1, ..., 6} {
        \draw (0, \j) to (1, \j);}
    
        \foreach \j in {0, 2, 4, 6, 7}{ \draw (0, \j) to (0, \j+1);}
        \foreach \j in {1, 3, 5, 7}{ \draw (1, \j) to (1, \j+1);}
        \draw[->] (1, 0) to (1, -0.5);
        \draw (1, 8.5) to (1, 8);
        \vertex[red] (r1) at (1, 7) {};
        \end{tikzpicture}
        \end{subfigure}
        \hspace{4mm}
       \begin{subfigure}[t]{0.3\textwidth}
        \centering
        \begin{tikzpicture}[scale=0.6]
        \foreach \x in {0,1} {
            \foreach \y in {0,1, 2, ..., 8} {
                \vertex (\x\y) at (\x, \y) {};
            }
        }

        \foreach \j in {3, 4, ..., 7} {
        \draw (0, \j) to (1, \j);}
    
        \foreach \j in {0,1, 2, 4, 6}{ \draw (0, \j) to (0, \j+1);}
        \foreach \j in {0, 1, 3, 5, 7}{ \draw (1, \j) to (1, \j+1);}
        \draw[->] (0, 8) to (0, 8.5);
        \draw (0, -0.5) to (0, 0);
        \vertex[red] (r1) at (1, 2) {};
        \end{tikzpicture}
        \end{subfigure}
                
    \caption{Argument for the third statement of Proposition \ref{prop:AG} when $r-(2k+1)=0.$  If $r-(2k+1)=0,$ then GG paths using $k-1$ and $k-2$ extra horizontal pairs respectively reach the vertex two below and two above the end of the GG path.  Note that if $k=1$ (and $r=3$),  GG paths using $0$ extra horizontal edges should be used instead.}
    \label{fig:0left}
\end{figure}
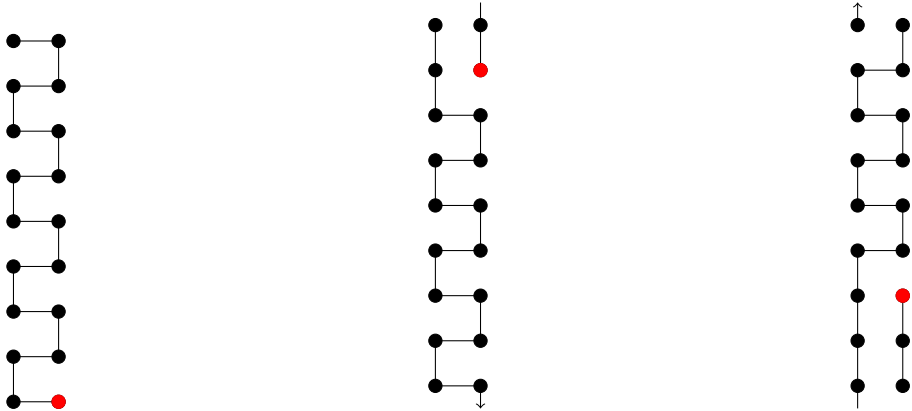
   
\begin{proof}
The first property follows by definition: vertices in $A_{r, c, m}$ are defined as using \emph{at most} $m$ extra horizontal pairs, and hence are also reachable using at most $m+1$ extra horizontal pairs.  

For the second property, it suffices to show that if $ i \in A_{r, 2, m}$, then $ i\pm 2 \mod r \in A_{r, 2, m+1}:$ For any $ j \in A_{r, c, m}$, consider a GG path $P$ reaching $(j, c-1)$ costing at most $2m+1$.  Let $(i, 1)$ be last vertex visited in the second column by $P$.  If $i\pm 2 \mod r$ are in  $A_{r, 2, m+1},$ then $ j \pm 2 \mod r \in A_{r, c, m+1}.$ 

So, suppose that $ i \in A_{r, 2, m}$ and is reachable using a GG path $P$ with $2k+1$ horizontal edges (where $k\leq m$).  If $k=0$, then $i\in \{0, 2, r-2\}$ (see, e.g., Figure \ref{fig:GG1}); but $A_{r, 2, 1}=\{0, 2, 4, r-4, r-2\}$ (see, e.g., Figures \ref{fig:GG3} and \ref{fig:GG1}).  If $k>0$ and $r-(2k+1)\geq 2,$ we can consider two GG paths of the same style as $P$ (i.e. starting vertically up if $P$ starts vertically up; starting vertically down otherwise). Considering one such path that uses $k-1$ extra horizontal pairs of edges and another such path that uses $k+1$ extra horizontal pairs of edges, we see that they reach $(i\pm 2 \mod r, 1).$  The only remaining cases are where $r-(2k+1) \in \{0, 1\}.$  Figure \ref{fig:1left} handles the case when $r-(2k+1)=1$ (appealing to cylindrical symmetry if the first edge is vertically up rather than vertically down), and Figure \ref{fig:0left} handles the case where $r-(2k+1) = 0.$  In both cases we see that $i\pm 2 \mod r \in A_{r, 2, m+1}.$

The third property summarizes the structure of GG paths beyond the first and second column: if there is a GG path to $(i, c-1)$, it can be extended to a GG path on $c+1$ columns, ending in row $i\pm 1:$ take the same path to $(i, c-1)$, add the edge to $(i, c)$, and then tour through the last column, starting vertically up (and ending at $(i+1 \mod r, c)$) or vertically down (and ending at $(i-1 \mod r, c)$).

Next, we prove property 4.  Let $m'$ be such that $c+2m'<r$ but $c+2(m'+1)\geq r.$ Note that
$$A_{r, c, m'} = \{-c-2m',  -c-2m'+2,  -c-2m'+4, ...,  c+2m'-2, c+2m'\}.$$
Now consider any $M>m'.$  Analogously, $$A_{r, c, M} = \{-c-2M,  -c-2M+2, -c-2M+4, ...,  c+2M-2,  c+2M\}.$$
The only values in $A_{r, c, M}$, but potentially not in $A_{r, c, m'}$, are those of the form $-c-2m'-2i$ and $ c+2m'+2i$ for some integer $i>0.$  We show that the latter is in $A_{r, c, m'};$ it is analogous to show for the former.   To do so, we want to find some $\lambda \in \Z, 0\leq \lambda \leq c+2m'$ such that $c+2m'+2i \equiv_r c+2m'-2\lambda.$  Doing so shows that $ c+2m'+2i$ is already included in $A_{r, c, m'}.$

First suppose that $r$ is even.  In this case, we note that, since $c\geq 2$, $$|(c+2m')-(-c-2m')| = 2c+4m'\geq c+2m'+2\geq r.$$  Thus we can subtract some integer multiple of $r$, say $\mu\times r$, from $c+2m'+2i$ so that 
\begin{align*}
    &c+2m'+2i-\mu r = c+2m' + 2i - 2\mu \f{r}{2} \equiv_r c+2m'+2i, \\ 
    &c+2m'+2i-2\mu\f{r}{2} \in \{-c-2m', c-2m'+2, ...., c+2m'\}. 
\end{align*} Setting $\lambda=\mu\f{r}{2}$ yields the desired integer, since $r$ is even.  If $r$ is odd, the proof proceeds similarly, but we must subtract a $\mu$ multiple of $2r.$ 
Note also that $c+2m'+2i- \mu\times 2r \equiv_2 c+2m'+2i$, but $ c+2m'+2i \not\equiv_2 \pm(c+2m'+1).$   Thus, if $c+2m'+2i- \mu\times 2r \in [-(c+2m'+1), c+2m'+1]$, it must also be in $[-(c+2m'), c+2m']$.
We also note that $$|(c+2m'+1)-(-c-2m'-1)| = 2c+4m'+2\geq 2r.$$  Hence there must be some integer multiple of $2r$, say $\mu\times 2r$, such that $c+2m'+2i-2\mu r$ is in $[-c-2m'-1, c+2m'+1].$  But $c+2m'+2i-2\mu r\equiv_2 c+2m',$ so we further have that $c+2m'+2i-2\mu r\in [-c-2m', c+2m'].$  Once again, $$c+2m'+2i- 2\mu r \equiv_r c+2m'+2i, \hspace{5mm} c+2m'+2i-2\mu r \in \{-c-2m', c-2m'+2, ...., c+2m'\},$$  and $\mu$ gives the desired integer.

Finally, we prove property 5. Recall that $A_{r, c, m} = \{c + 2m - 2i \mod r: 0 \leq i \leq c+2m\}$. It is useful to think of $A_{r,c,m}$ as $\{-(c+2m), -(c+2m)+2, \ldots, c+2m-2, c+2m\}$. That is, $A_{r, c, m}$ is the set of numbers between $-(c+2m)$ and $(c+2m)$ inclusive, with the same parity as $c+2m$, taken mod $r$. 

 Suppose $x = c+2m-2i \mod r$, where $0 \leq i \leq c+2m$. By property 4, we may assume that $c+2m < r$. Then 
 $$
x = 
\begin{cases}
c+2m-2i, &\text{if $0\leq c+2m-2i < r$,} \\
c+2m-2i+r, &\text{if $-r < c+2m-2i < 0$.}
\end{cases}
$$
We will now consider the two cases, and show that in each case, we have that both $x$ and $x+2$ are in $A_{r+2,c,m+1}$. Before moving on, recall that $A_{r+2,c,m+1}$ is the following set:
$$A_{r+2,c,m+1} = \{c+(2m+2) - 2i \mod r+2: \; 0 \leq i \leq c+2m+2\}.$$ 

\underline{Case 1:} $0 \leq c +2m - 2i < r$. Then we can write
$$x = c+2m-2i = c+(2m+2) - 2(i+1) \in A_{r+2,c,m+1},$$ 
and similarly
$$x+2 = c+2m-2i+2 = c+(2m+2) - 2i \in A_{r+2,c,m+1}.$$

\underline{Case 2:} $-r < c+2m-2i < 0$. In this case, we write
\begin{align*}
x &= c+2m-2i+r \\
&= c + 2m + 2 - 2(i+2) + (r+2) \\
&= c+(2m+2) - 2(i+2) \mod (r+2) \\
&\in A_{r+2,c,m+1}
\end{align*}
and similarly 
\begin{align*}
    x+2 &= c+2m-2i+r+2 \\
    &= c+(2m+2) - 2(i+1) + (r+2) \\
    &= c + (2m+2) - 2(i+1) \mod (r+2) \\
    &\in A_{r+2,c,m+1}
\end{align*}
Thus, in all cases, we have that both $x$ and $x+2$ are in $A_{r+2,c,m+1}$, which completes the proof.
\end{proof}

\section{Confirming the Conjecture of  Gerace  and Greco \cite{Ger08b, Grec07}} 
\label{app:weimplygg}

Recall our main result
\result*

We show that this result resolves Greco and Gerace's \cite{Ger08b, Grec07} conjecture.  Greco and Gerace defined $$S=\left\{y: 0\leq y < \f{n}{g_1}, \, (2y-g_1)a_1+g_1a_2 \equiv_n 0\right\}.$$   Their main result (Theorem 4.4 in  Gerace and Greco \cite{Ger08b}) states:
\begin{enumerate}
\item If $S$ is empty, the cost of the optimal tour is $2c-2$.
\item Let $y_1=\min\{S\}$ and $y_2 = \max\{S\}.$ If $y_1 \leq g_1$, the cost of the optimal tour is $c$.
\item Otherwise, let $m=\min\{y_1 - g_1, \f{n}{g_1}-y_2\}.$  There exists a tour of cost $c+2m$.
\end{enumerate}
 Greco and Gerace \cite{Ger08b, Grec07} conjectured that, in cases where $S$ is nonempty, either the tours in case 1 or 3 above are optimal. That is, for cases where $S$ is nonempty, the cost of the optimal tour is $\min\{c+2m, 2c-2\}.$
 
 We now prove that our characterization confirms this conjecture.  To prove this conjecture, we recall from Proposition \ref{prop:getLB} and Theorem \ref{thm:result} that the lower bound $c$ is optimal 
if and only if there exists an integer $y$, with $0\leq y\leq g_1$, such that $$(2y-g_1)a_1+g_1a_2 \equiv_n 0.$$  Recall also that, if $\f{n}{g_1}-1\leq g_1$, then $S$ is either empty or $\min{S}\leq \f{n}{g_1}-1\leq g_1$ so that the lower bound is optimal.

To resolve Greco and Gerace's conjecture, we show the following:
\begin{thm}
Consider an instance of two-stripe TSP where $g_1+1<\f{n}{g_1}.$  Let $$S=\left\{y: 0\leq y < \f{n}{g_1}, \, (2y-g_1)a_1+g_1a_2 \equiv_n 0\right\}.$$  Assume that $S$ is nonempty, that $y_1=\min\{S\}$ and $y_2 = \max\{S\},$ and that $y_1 > g_1$.  Finally, assume that the cost of the optimal solution to the two stripe instance is strictly between $c$ and $2c-2$. Then the cost of the optimal solution to the two-stripe instance is $c+2m^*$  if and only if $m^*=\min\{y_1-g_1, \f{n}{g_1}-y_2\}.$
\end{thm}


\begin{proof}

\begin{cm}\label{cm:GGequiv}
Suppose that the optimal solution to the two-stripe  is $c+2m^*$ with $c<c+2m^* <2c-2$.  Then $\min\{y_1-g_1, \f{n}{g_1}-y_2\}\leq m^*$. 
\end{cm}

Recall from Corollary \ref{cor:newrows} that, if the cost of the optimal solution is $c+2m^*$ with $c<c+2m^* <2c-2$ (i.e., if the optimal solution is a GG path plus a $(0-(-a_2))$ edge), then the corresponding GG path ends in the last column at row $c+2m^* \mod r$ or at row $-(c+2m^*) \mod r$.

\noindent {\bf Case 1: The optimal tour corresponds to a GG path ending at row index $c+2m^* \mod r.$}

First, suppose that the optimal tour corresponds to a GG path ending at row index $c+2m^*.$  The corresponding GG path is then a path that starts going vertically up (as in the right of Figure \ref{fig:GG3full}), using $2m+1$ horizontal edges between the first two columns, and, in every subsequent column, moving up  and ending at $-a_2.$  This path uses $g_1-1+2m^*$ edges of length $\pm a_2$ with $2m^*$ of those edges alternating between the first two columns.  Since it uses $n-1$ total edges, and all other edges are vertically up (i.e. of length $-a_1$), we have that it uses:
\begin{itemize}
\item $m^*$ edges of length $-a_2$
\item $g_1-1+m^*$ edges of length $a_2$
\item $n-1-(g_1-1+2m^*)=n-g_1-2m^*$ edges of length $-a_1.$ 
\end{itemize}
Since it ends at $-a_2$, we know that
$$-a_2 \equiv_n -m^*a_2 + (g_1-1+m^*)a_2 + (n-g_1-2m^*)(-a_1) \equiv_n (g_1-1)a_2 + (g_1+2m^*)a_1.$$  Rearranging, we see that $$2m^*a_1 \equiv_n -g_1a_2 -g_1a_1.$$  Now let $y=m^*+g_1\geq 0$  Then
\begin{align*}
(2y-g_1)a_1 + g_1 a_2 &= (2(m^*+g_1)-g_1)a_1+g_1a_2\\
&=2m^*a_1 +g_1a_1 + g_1a_2 \\
&\equiv_n -g_1a_2-g_1a_1 + g_1a_1 + g_1a_2 \\
&= 0
\end{align*}
Hence, $y\in S$ provided $y<\f{n}{g_1}.$  To see that $y<\f{n}{g_1}$ suppose not. Let $y'=y-\f{n}{g_1}$.  By assumption, $y'\geq 0.$ Moreover 
$$y' = y-\f{n}{g_1} = m^* + g_1 - \f{n}{g_1} <m^*,$$ since we assume  $g_1+1<\f{n}{g_1}.$  Moreover, since we assume $c+2m^* <2c-2,$ we have that $m^* < c = g_1.$  Thus $0\leq y' < \f{n}{g_1}.$  Finally, 
$$y' a_1 = (y-\f{n}{g_1})a_1 \equiv_n ya_1,$$ so that $y'\in S$ and $\min\{S\} < g_1.$  Then, however, the optimal solution must cost $c$ (by, e.g., Proposition \ref{prop:GGFull}).  Thus we cannot have $y\geq \f{n}{g_1}$ and must have $y\in S$.  Since $y_1=\min{S}$, $y_1\leq y.$  Finally, 
$$\min\{y_1-g_1, \f{n}{g_1}-y_2\}\leq y_1-g_1 \leq y-g_1 = m^*.$$

\noindent {\bf Case 2: The optimal tour corresponds to a GG path ending at row index $-(c+2m^*) \mod r.$}

This case proceeds almost identically to Case 1.  If the GG path ends at row index  $-(c+2m^*) \mod r,$ then the only difference is that it uses $(n-g_1-2m^*)$ edges of length $a_1$.  Since it ends at $-a_2$, we know that
$$-a_2 \equiv_n -m^*a_2 + (g_1-1+m^*)a_2 + (n-g_1-2m^*)(a_1) \equiv_n (g_1-1)a_2 - (g_1+2m^*)a_1.$$  Rearranging, we see that $$-2m^*a_1 \equiv_n -g_1a_2 +g_1a_1.$$  Now let $y=\f{n}{g_1}-m^*$. Then
\begin{align*}
(2y-g_1)a_1 + g_1 a_2 &= (2(\f{n}{g_1}-m^*)-g_1)a_1+g_1a_2\\
 &\equiv_n -2m^*a_1 -g_1a_1+g_1a_2 \\
  &\equiv_n  -g_1a_2 +g_1a_1 -g_1a_1+g_1a_2 \\
&= 0
\end{align*}

Moreover $y<\f{n}{g_1}$ since $m^*\geq 1.$  Since  $c+2m^* <2c-2$ we again have that $m^* < c = g_1$ and, further,  since $g_1<\f{n}{g_1},$ we have that $m^* < \f{n}{g_1}.$  Thus $0\leq y <\f{n}{g_1}$ and again $y\in S$.  Since $y_2=\max{S}$, $y_2\geq y.$  Finally, 
$$\min\{y_1-g_1, \f{n}{g_1}-y_2\}\leq \f{n}{g_1}-y_2 \leq\f{n}{g_1}-y = m^*.$$

These two cases resolve Claim \ref{cm:GGequiv}.  To finish proving the theorem, we want to show the following

\begin{cm}\label{cm:GGequiv2}
Suppose that $m=\min\{y_1-g_1, \f{n}{g_1}-y_2\}$ with $m> 0.$  Then the cost of the optimal two-stripe solution is at most $c+2m.$
\end{cm}

To prove this claim, we need only construct a Hamiltonian tour of cost $c+2m.$  This is done in Theorem 4.4 of  Gerace and Greco \cite{Ger08b} using GG paths.  For completeness, we sketch the proof here. We have two cases that parallel the previous two cases.

\noindent {\bf Case 1: $m=y_1-g_1.$}
Suppose that $m=y_1-g_1.$  Consider a GG path where all edges are vertically up (of length $-a_1$) of cost $(c-1)+2m.$  This is, by construction, a Hamiltonian path.  It gives rise to a Hamiltonian tour of cost $c+2m$ provided it ends at $-a_2.$  As in Case 1 of Claim \ref{cm:GGequiv}, we can calculate that it ends at $$(g_1-1)a_2 + (g_1+2m)a_1 \mod n.$$  Since $y_1\in S$, we have that $$g_1a_2 + (2y_1-g_1)a_1 \equiv_n 0.$$  Using the fact that $m=y_1-g_1$, we get:
\begin{align*}
(g_1-1)a_2 + (g_1+2m)a_1 &= (g_1-1)a_2 + g_1a_1 + 2ma_1 \\
&= (g_1-1)a_2 + g_1a_1 +2(y_1-g_1)a_1 \\
&= (g_1-1)a_2 + (2y_1-g_1)a_1 \\
&= -a_2 + g_1a_2 + (2y_1-g_1)a_1 \\
&\equiv_n -a_2,
\end{align*}
using the fact that $g_1a_2 + (2y_1-g_1)a_1 \equiv_n 0$  in the last line.  Thus this GG path gives rise to a Hamiltonian tour of cost $c+2m$.

\noindent {\bf Case 2: $m=\f{n}{g_1}-y_2.$}
Suppose that $m=\f{n}{g_1}-y_2.$  Consider a GG path where all edges are vertically down (of length $a_1$) of cost $(c-1)+2m.$  This is, by construction, a Hamiltonian path.  It gives rise to a Hamiltonian tour of cost $c+2m$ provided it ends at $-a_2.$  As in Case 2 of Claim \ref{cm:GGequiv}, we can calculate that it ends at $$(g_1-1)a_2 - (g_1+2m)a_1 \mod n.$$  Since $y_2\in S$, we have that $$g_1a_2 + (2y_2-g_1)a_1 \equiv_n 0.$$  Using the fact that $m=\f{n}{g_1}-y_2$, we get:
\begin{align*}
(g_1-1)a_2 - (g_1+2m)a_1 &= (g_1-1)a_2 - (g_1+2\f{n}{g_1}-2y_2)a_1 \\
&\equiv_n (g_1-1)a_2 + (2y_2-g_1)a_1\\
&= -a_2 + g_1 a_2 +(2y_2-g_1)a_1\\
&\equiv_n -a_2,
\end{align*}
using the fact that $g_1a_2 + (2y_2-g_1)a_1 \equiv_n 0.$  in the last line.  Thus this GG path gives rise to a Hamiltonian tour of cost $c+2m$.

\hfill
\end{proof}
\end{APPENDICES}

\bibliography{bibliog} 
\bibliographystyle{abbrv}
\end{document}